\documentclass[acmtog]{acmart}

\makeatletter

\setcopyright{cc}
\setcctype{by}
\acmJournal{TOG}
\acmYear{2025} \acmVolume{44} \acmNumber{4} \acmArticle{} \acmMonth{8} \acmPrice{}\acmDOI{10.1145/3730851}

\usepackage{bbm}
\usepackage{dsfont}

\newcommand{\ourmethod}{TetWeave}
\newcommand{\ourtitle}{\ourmethod{}: Isosurface Extraction using On-The-Fly Delaunay Tetrahedral Grids for Gradient-Based Mesh Optimization}

\newcommand{\figref}[1]{Fig.\ \ref{#1}}
\newcommand{\secref}[1]{Sec.\ \ref{#1}}

\newcommand{\tableref}[1]{Table \ref{#1}}
\newcommand{\appref}[1]{Appendix \ref{#1}}
\newcommand{\bigo}{\mathcal{O}}

\usepackage[percent]{overpic}
\usepackage{multirow}
\usepackage{pifont}
\usepackage{nicefrac}
\usepackage{colortbl}
\usepackage{xcolor}
\usepackage{bm}
\usepackage{subcaption}

\definecolor{acmblue}{RGB}{0, 102, 204}  
\definecolor{githubpurple}{RGB}{110, 84, 148}  

\newcommand{\cmark}{\ding{51}}%
\newcommand{\xmark}{\ding{55}}%
\definecolor{darkgreen}{rgb}{0.0, 0.6, 0.0}
\definecolor{darkred}{rgb}{0.85, 0.0, 0.0}
\newcommand{\cmarkcol}{{\color{darkgreen}\cmark}}
\newcommand{\xmarkcol}{{\color{darkred}\xmark}}

\usepackage{enumitem}

\setlist[itemize]{left=0pt, itemsep=0pt, topsep=0pt}
\setlist[enumerate]{left=0pt, itemsep=0pt, topsep=0pt}

\title{\ourtitle{}}

\author{Alexandre Binninger}
\affiliation{\institution{ETH Zurich}\city{Zurich}\country{Switzerland}}
\email{alexandre.binninger@inf.ethz.ch}
\orcid{0000-0002-9833-4126}
\author{Ruben Wiersma}
\affiliation{\institution{ETH Zurich}\city{Zurich}\country{Switzerland}}
\email{ruben.wiersma@inf.ethz.ch}
\orcid{0000-0001-7900-7253}
\author{Philipp Herholz}
\affiliation{\institution{Zurich}\country{Switzerland}}
\email{ph.herholz@gmail.com}
\orcid{0000-0002-8389-792X}
\author{Olga Sorkine-Hornung}
\orcid{0000-0002-8089-3974}
\affiliation{\institution{ETH Zurich}\city{Zurich}\country{Switzerland}}
\email{sorkine@inf.ethz.ch}

\begin{document}

\begin{abstract}
  We introduce \ourmethod{}, a novel isosurface representation for gradient-based mesh optimization that jointly optimizes the placement of a tetrahedral grid used for Marching Tetrahedra and a novel directional signed distance at each point.
\ourmethod{} constructs tetrahedral grids on-the-fly via Delaunay triangulation, enabling increased flexibility compared to predefined grids. The extracted meshes are guaranteed to be watertight, two-manifold and intersection-free.
The flexibility of \ourmethod{} enables a resampling strategy that places new points where reconstruction error is high and allows to encourage mesh fairness without compromising on reconstruction error.
This leads to high-quality, adaptive meshes that require minimal memory usage and few parameters to optimize.
Consequently, \ourmethod{} exhibits near-linear memory scaling relative to the vertex count of the output mesh — a substantial improvement over predefined grids.
We demonstrate the applicability of \ourmethod{} to a broad range of challenging tasks in computer graphics and vision, such as multi-view 3D reconstruction, mesh compression and geometric texture generation.
Our code is available at 
%
\href{https://github.com/AlexandreBinninger/TetWeave}{%
\textcolor{githubpurple}{https://github.com/AlexandreBinninger/TetWeave}}.
\end{abstract}


\begin{teaserfigure}
  \includegraphics[width=\linewidth]{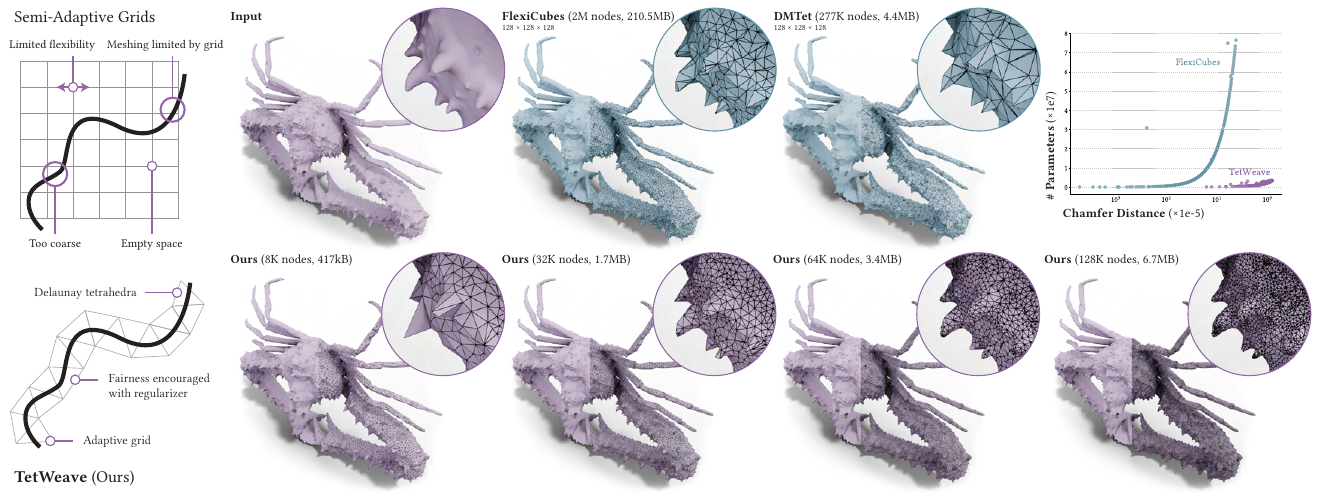}
  \Description{\ourmethod{} Teaser}
  \caption{\textbf{\ourmethod{}} jointly optimizes a tetrahedral grid and a directional signed distance function used for Marching Tetrahedra. Our method \textit{weaves} a background grid around the surface, which is regularized to give fair output meshes. The results are compared with semi-adaptive grid methods (top row), which start from a predefined grid. These methods have limited flexibility and their meshing is influenced by the configuration of the predefined grid. They require many more nodes in the background grid because of unused empty space to get a similar level of detail in the output mesh (graph in top right).}
  \label{fig:teaser}
\end{teaserfigure}

\maketitle

\section{Introduction}
\label{sec:introduction}

\begin{figure*}[t!]
    \centering
	\small
	\setlength{\tabcolsep}{5pt}
	\begin{tabular}{cccccc}
\includegraphics[width=0.22\linewidth]{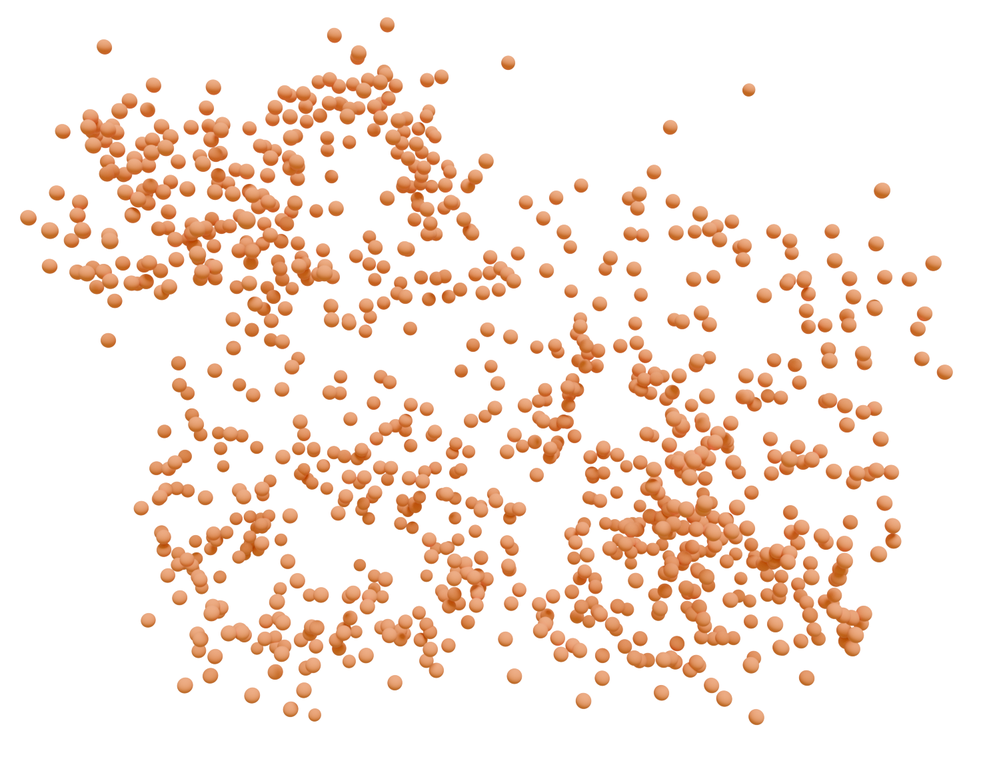} & \includegraphics[width=0.22\linewidth]{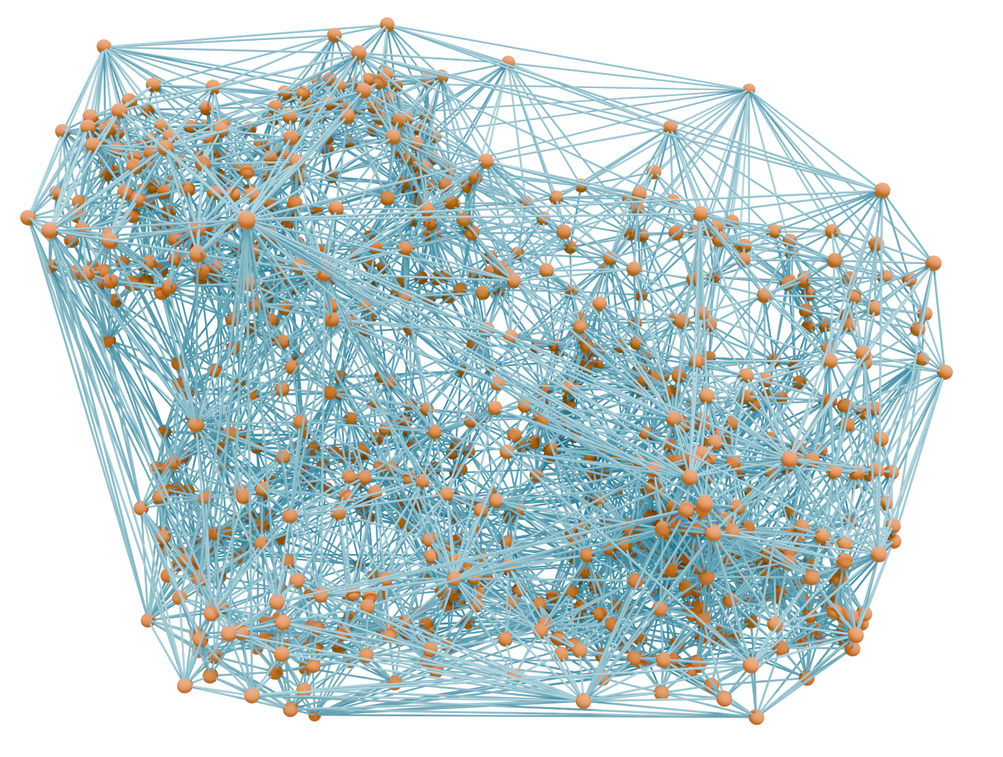} &
\includegraphics[width=0.22\linewidth]{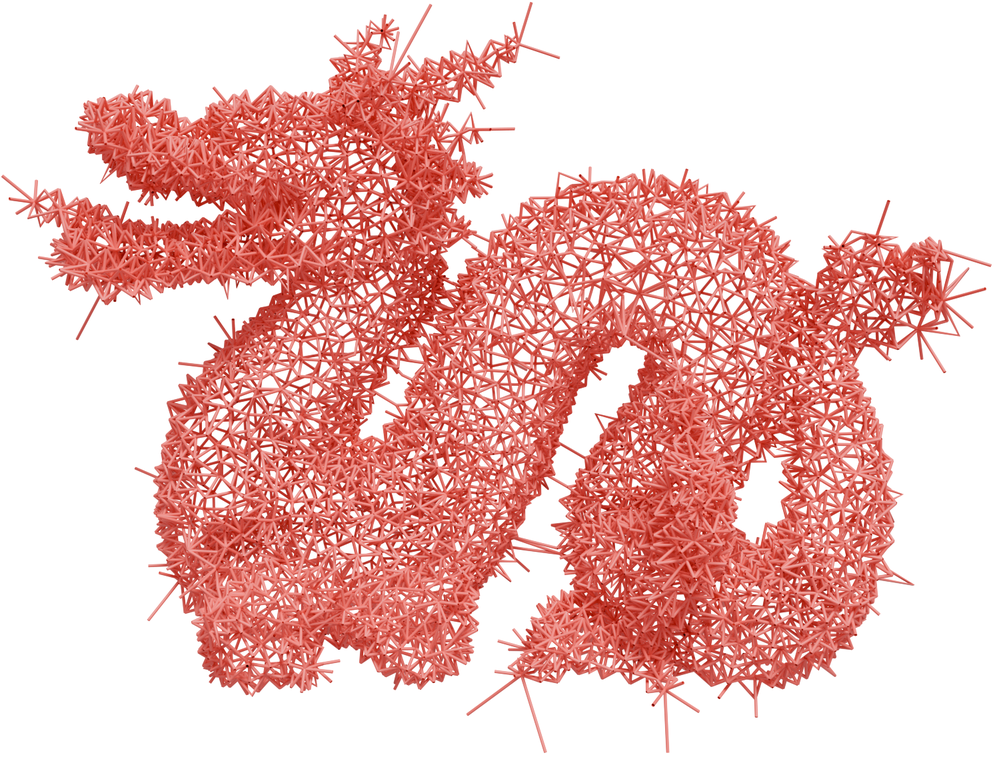} &
\includegraphics[width=0.22\linewidth]{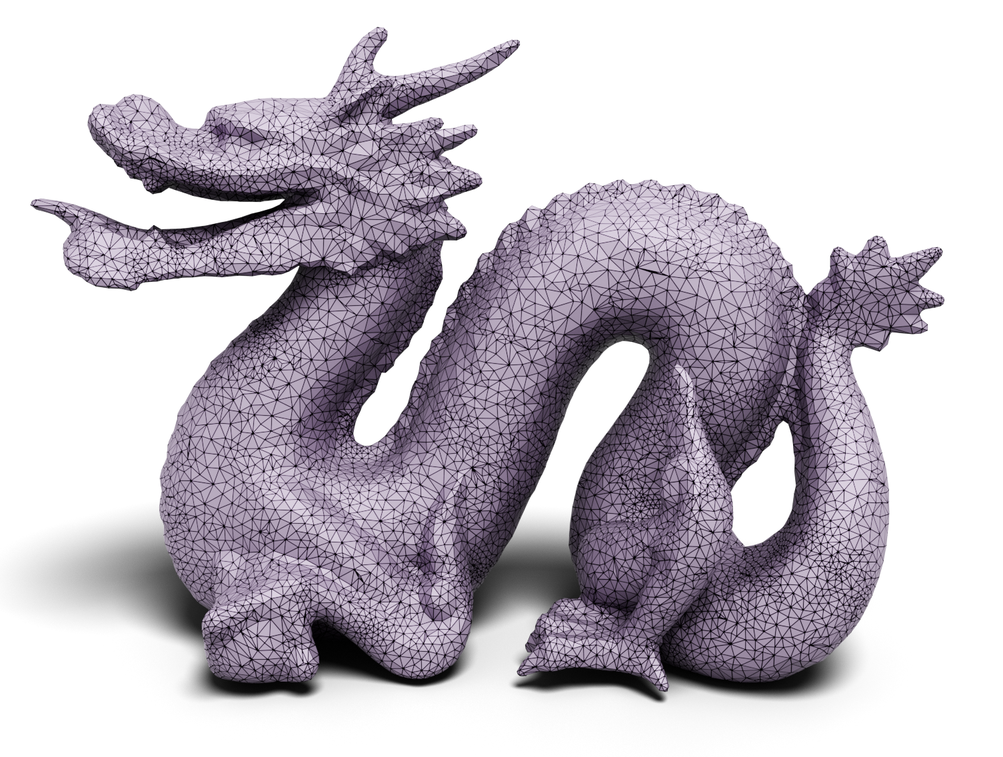}  \\
Point cloud & Delaunay triangulation & Active edges &  Mesh extraction \\
        \end{tabular}
    \caption{Illustration of our mesh extraction pipeline, which begins with a point cloud where each point is associated with a signed distance value. The process starts by generating a tetrahedral grid through Delaunay triangulation. Next, active edges are identified, and a directional signed distance is computed for each active point using spherical harmonics (\secref{sec:directional_sdf}). The final mesh is extracted using the Marching Tetrahedra algorithm (\secref{sec:mesh_extraction}). Our method iteratively refines the randomly initialized point cloud, distributing points to closely align with the target shape (\secref{sec:resampling}), which ensures scalability and adaptability. In this figure, only a portion of the point cloud and Delaunay Triangulation is displayed to enhance clarity.
    }
    \label{fig:pipeline_figure}
    \Description{Pipeline figure}
\end{figure*}
Many recent 3D applications require shape representations that are both expressive for artists and differentiable for optimization.
Differentiable shape representations enable appealing applications, such as shape generation \cite{gao2022get3d}, text-to-3D synthesis \cite{poole2022dreamfusion}, inverse rendering \cite{Munkberg_2022_CVPR} and geometric texture synthesis \cite{Liu:Paparazzi:2018,Hertz2020deep}. 
{An illustrative task using differentiable shape representations is surface reconstruction from multi-view images. In this task, the objective is to recover the 3D geometry of an object from a set of input views. A common approach is to optimize a shape representation to match the input views using gradient descent. This requires a differentiable renderer and, of interest to us, a differentiable shape representation.}
While triangle meshes are widely used in computer graphics as an efficient and practical shape representation and thus highly desirable, they are difficult to use in a differentiable setting, especially when the topology is not known a priori{, as is the case with surface reconstruction}. To address these limitations, considerable effort has been invested in developing representations capable of generating meshes for gradient-based optimization, rather than relying on direct mesh optimization.

A typical approach {in shape optimization} is to optimize a signed distance function and then generate a mesh using an isosurface extraction method, like Marching Cubes \cite{marchingcube}. This approach inherits limitations from the isosurface extraction technique. For instance, Marching Cubes is prone to staircasing artifacts. To overcome these challenges, works like DMTet~\cite{shen2021dmtet} and FlexiCubes~\cite{shen2023flexicubes} jointly optimize the implicit representation and the spatial grid structure from which the final mesh is extracted. We refer to these approaches as \emph{grid-adaptive isosurface representations}. Still, these approaches face issues, such as high memory consumption due to poor scaling properties, susceptibility to self-intersections and challenges with adaptive meshing across complex multi-scale grids.

{To address these issues,} we propose \ourmethod{}, a novel differentiable isosurface representation that allows joint optimization of the placement of an unstructured tetrahedral grid and the signed distance at each node of the grid. Contrary to previous methods like DMTet or FlexiCubes that rely on deforming a precomputed grid structure, we use Delaunay triangulation \cite{delaunay1934} on an arbitrary point cloud to generate a support structure (which we refer to as a \textit{grid}). This grid is used for isosurface extraction with Marching Tetrahedra \cite{doi_marchingtet}, which guarantees meshes to be watertight, 2-manifold and intersection-free. These are crucial properties for many downstream applications.

{The flexibility of TetWeave comes at a much lower cost than prior work using a predefined grid. For example, FlexiCubes~\cite{shen2023flexicubes} starts from a voxel grid, where each grid cell is equipped with parameters that adjust the position of the grid points and the placement of the resulting mesh's vertices. This flexibility requires a considerable number of parameters per grid cell: $21$ parameters per grid cell and $3$ per grid point. Moreover, because FlexiCubes starts from a voxel grid, the surface resolution scales poorly as the grid resolution increases. This results in a memory-intensive representation that inevitably fails to reconstruct high-frequency details (see \autoref{fig:teaser}). TetWeave does not rely on such a predefined grid structure. The grid points are free to move anywhere in ambient space and only require storing a position and the value of the signed distance function at that location.

To allow additional flexibility for the placement of mesh vertices, we introduce the notion of \emph{directional signed distance}, encoded with spherical harmonics coefficients at each point. This allows distinct implicit surface positions along different grid edges (unlike a single SDF per point), providing finer control over vertex placement during extraction. To ensure that running Marching Cubes on the resulting background grid produces high-quality triangulations with minimal sliver triangles, we propose} a simple loss function measuring fairness. Finally, our approach incorporates a resampling technique to refine details during optimization, enabling adaptive meshing tailored to customizable objectives, such as reducing reconstruction error. Thus, we achieve linear memory scaling relative to the resolution of the resulting meshes, while producing high-quality, highly detailed meshes with adaptive resolution, {as shown in \figref{fig:teaser}}. We demonstrate that this is useful in applications of gradient-based mesh optimization, such as multi-view 3D reconstruction, mesh compression and geometric texture generation.

Summarizing our main contributions:
\begin{itemize}
\setlength{\itemsep}{0pt}
    \setlength{\parskip}{0pt}
    \setlength{\parsep}{0pt}
    \item We use Marching Tetrahedra on Delaunay triangulations of arbitrary point clouds in gradient-based mesh optimization pipelines.
    \item A directional signed distance function to more accurately capture the distance to the surface along tetrahedral edges.
    \item A method to adapt the tetrahedral grid to an unknown surface.
    \item Two regularization terms to improve the quality of our meshes.
\end{itemize}

{We accentuate that \ourmethod{} is not designed as a general-purpose adaptive meshing technique or for isosurface extraction from fixed scalar fields. Rather, we propose a specialized representation optimized for gradient-based mesh processing, particularly suited for applications like multi-view 3D reconstruction.}
\section{Related work}
\label{sec:related_work}

In this section, we briefly review related work on differentiable surface representations, with a particular emphasis on isosurface extraction methods. We refer the reader to the survey by \citet{araujo_survey_poly} for a more extensive presentation. \tableref{tab:methodCriteria} provides a compact comparison of closely related isosurfacing methods. The approaches are compared according to the following criteria, inspired by \cite{shen2023flexicubes}:

\begin{itemize}
    \item \emph{Grad}: differentiation with respect to the generated mesh is possible and enables robust gradient-based pipelines.
    \item \emph{Sharp}: sharp features are correctly reconstructed.
    \item \emph{Fair}: the tessellation is fair, with few sliver triangles.
    \item \emph{Intersection-free}: the mesh is guaranteed not to self-intersect.
    \item \emph{2-manifold}: the topology is guaranteed to be two-manifold.
    \item \emph{Scalable resolution}: the number of vertices scales well with the number of parameters of the underlying representation.
\end{itemize}

\begin{figure}[t]
    \centering
    \small
    \setlength{\tabcolsep}{0pt}
    \begin{tabular}{c c}
    \includegraphics[width=.5\linewidth]{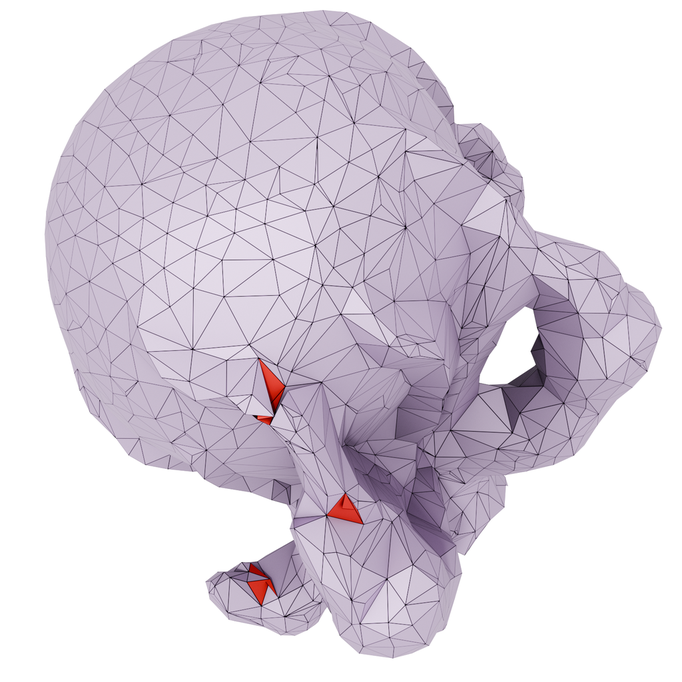} & 
    \includegraphics[width=.5\linewidth]{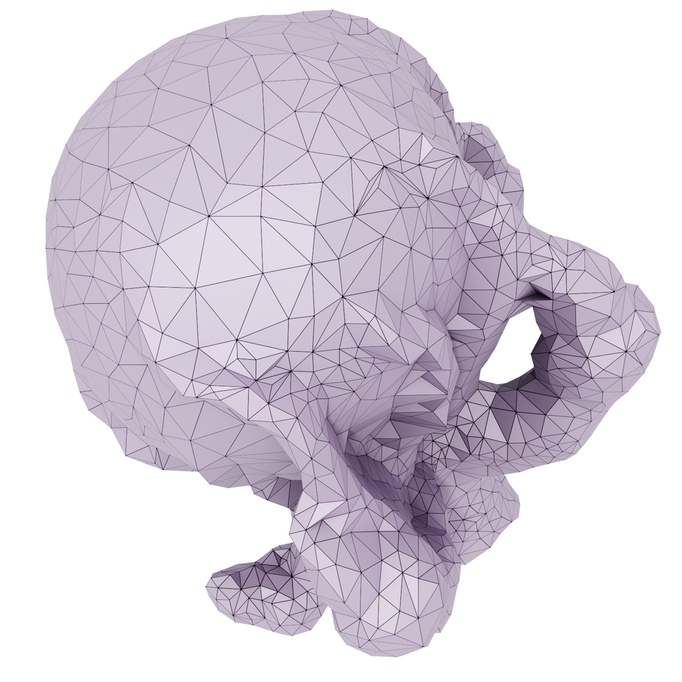} \\
    FlexiCubes - $128^3$ grid resolution & \ourmethod{} - $32$K points \\
    \end{tabular}
    \caption{Contrary to FlexiCubes \cite{shen2023flexicubes}, our method is guaranteed free from self-intersections.}
    \label{fig:self_intersections}
    \Description{TODO}
\end{figure}

\paragraph{Differentiable shape representations.}
Some methods differentiate directly on the positions of a (triangular) mesh. Such approaches require careful regularization to avoid degeneracy \cite{Nicolet2021Large} and are bound by the topology of the initialization \cite{wang2018pixel2mesh,hanocka_2020_point2mesh}. Mesh R-CNN \cite{gkioxari2020meshrcnn} uses a two-stage approach to support arbitrary topologies, first predicting the topology before optimizing the vertex positions of a template mesh. AtlasNet \cite{groueix2018atlasnetpapiermacheapproachlearning} predicts parametric patches to reconstruct a mesh, but the patches are not guaranteed to be connected.
MeshSDF \cite{remelli2020meshsdf} allows for the computation of gradients from the mesh extraction process.

Several learning methods reconstruct a mesh from a point cloud. PointTriNet \cite{sharp2020ptn} iteratively predicts triangle meshes, IERMeshing \cite{Liu_2020_IERMeshing} constructs the mesh based on the ratio between the geodesic and Euclidean distance, and
DSE \cite{Rakotosaona_2021_CVPR_DSE} builds a triangle mesh via 2D Delaunay triangulation of local projections. DeepDT \cite{LuoMT21DeepDT} and DMNet \cite{Zhang_2023_ICCV_DMNet} utilize Delaunay triangulation for point cloud meshing, employing neural networks to classify tetrahedra for surface reconstruction. These methods have difficulties ensuring that the resulting meshes are watertight or manifold. By introducing a continuous latent connectivity space at each vertex,
SpaceMesh \cite{spacemesh2024} trains a neural network that generates a watertight 2-manifold mesh from a point cloud, but cannot guarantee intersection-free outputs.

Neural implicit representations emerged as a powerful representation for continuously generating or interpolating shapes with various topologies. Early work represents shapes with a unique latent vector, and a neural network maps this latent code to a signed distance function (SDF) \cite{DeepSDF_Park_2019_CVPR,mescheder2019occupancynetworkslearning3d}. Many improvements over these methods involve latent code decomposition, being in a regular grid structure \cite{Peng2020ECCV} or a multiresolution grid \cite{takikawa2021nglod,mueller2022instant}, as a point cloud \cite{zhang20223dilg,petrov2024gem3d}, as a set of local grids \cite{Yariv_2024_CVPR,yang2024galageometryawarelocaladaptive}, as a set of 3D Gaussians \cite{SPAGHETTI:2022}, or as a general set \cite{3DShape2VecSetzhang2023} where the weight of the latent code is given via an attention mechanism, rather than spatial proximity. Other methods focus on learning parts of the shape with different frequencies via positional encoding schemes \cite{sitzmann2019siren,tancik2020fourfeat,SAPE:2021} or by learning a high-frequency displacement map on top of a base implicit shape \cite{yifan2022geometryconsistent}. 
Despite these improvements, the resulting representations still need to be converted to meshes via isosurface mesh extraction techniques for use in downstream applications or in optimization pipelines operating directly on meshes.\\

\begin{table}[b]
    \centering
    \scriptsize
\setlength{\tabcolsep}{2pt}
    \caption{Comparison of isosurface mesh extraction methods in three categories: classic isosurfacing methods are typically used on top of a sign distance field, neural methods uses a neural network to estimate the parameters of an isosurface extraction technique, while grid adaptive methods allow for joint optimization of the grid structure and the mesh extraction's parameters. Criteria are explained in \secref{sec:related_work}. 
    }
    \begin{tabular}{c@{\hspace{0pt}}cl@{\hspace{0pt}}cccccc}
        \toprule
        & & & \multirow{2}{*}{Grad} & \multirow{2}{*}{Sharp} & \multirow{2}{*}{Fair} & Intersection &  \multirow{2}{*}{2-manifold} & Scalable\\
        & & &  &  &  & Free & & Resolution \\  \midrule
        
        \multicolumn{2}{c}{\multirow{4}{*}{\rotatebox{90}{classic}}} & MC \cite{marchingcube} & \cmarkcol & \xmarkcol & \xmarkcol & \cmarkcol & \xmarkcol & \xmarkcol  \\
        & & Dual Contouring \cite{ju_dualcontouring_2002} & \xmarkcol & \cmarkcol & \xmarkcol & \xmarkcol & \xmarkcol & \xmarkcol \\
        & &  DMC—centroid \cite{nielson_dualmarchingcube_2004} 
 & \cmarkcol & \xmarkcol & \cmarkcol & \cmarkcol & \cmarkcol & \xmarkcol  \\
        & &  DMC—QEF \cite{schaefer2007manifolfdualcontouring} & \xmarkcol & \cmarkcol & \cmarkcol & \cmarkcol & \cmarkcol & \xmarkcol \\
        \midrule

        \multicolumn{2}{c}{\multirow{3}{*}{\rotatebox{90}{\scriptsize neural}}} & NMC \cite{chen2021nmc} & \cmarkcol & \cmarkcol & \xmarkcol 
        & \xmarkcol & \cmarkcol & \xmarkcol \\
        & & NDC \cite{chen2022ndc} & \xmarkcol & \cmarkcol & \cmarkcol & \cmarkcol & \xmarkcol & \xmarkcol  \\
        & & Voromesh \cite{Maruani_2023_ICCV} & \cmarkcol & \cmarkcol & \xmarkcol & \cmarkcol & \cmarkcol & \cmarkcol \\
        \midrule
        
        \multirow{3}{*}{\rotatebox{90}{\scriptsize grid}} & \multirow{3}{*}{\rotatebox{90}{\scriptsize adaptive}} & DMTet \cite{shen2021dmtet} & \cmarkcol & \cmarkcol & \xmarkcol & \cmarkcol & \cmarkcol & \xmarkcol \\
       & & FlexiCubes \cite{shen2023flexicubes} & \cmarkcol & \cmarkcol & \cmarkcol & \xmarkcol & \cmarkcol & \xmarkcol \\
        & & \ourmethod{} \textbf{(Ours)} & \cmarkcol & \cmarkcol & \cmarkcol & \cmarkcol & \cmarkcol & \cmarkcol \\
        \bottomrule
    \end{tabular}
    \label{tab:methodCriteria}
\end{table}

\textit{Isosurface mesh extraction }
is the process of generating a surface mesh from an implicit function, typically by evaluating an SDF on a regular grid \cite{marchingcube} or a tetrahedral grid \cite{doi_marchingtet}. Vertices of the mesh are placed on the edges of the background grid and their connectivity is based on a lookup table. While widely used, this approach fails to reconstruct sharp features and produces unfair tessellations. 
Dual Contouring \cite{ju_dualcontouring_2002} leverages the dual of an octree grid and optimizes vertices inside cubes based on normals to recover sharp features better, but produces non-manifold, self-intersecting meshes. \citet{schaefer2007manifolfdualcontouring} also use an optimization setting within a Dual Marching Cube formulation by optimizing quadratic error functions (QEF). Placing vertices at the face centroids ensures the differentiability of the isosurface extraction but comes at the expense of flexibility, which is essential for accurately reconstructing sharp features \cite{nielson_dualmarchingcube_2004}. Marching Triangles \cite{marchingtriangles} iteratively expands a surface mesh from a starting vertex by enforcing a Delaunay property at each newly added triangle. DelIso \cite{DelaunayIsosurfacingLevine2007} proposes a two-stage strategy that extracts a coarse surface from an initial 3D Delaunay triangulation and refines it with a surface-restricted Delaunay triangulation. However, the method is prone to artifacts when the isosurface exhibits sharp edges. 
More recently, \citet{Sellan2023RFTS,Sellan2024RFTA} introduce techniques that improve reconstruction by leveraging tangency information derived from signed distance values to extract additional insights from point clouds.

Some methods leverage the use of a neural network to predict the parameters of a regular structure from which they can generate a mesh. DefTet \cite{gao2020deftet} predicts the deformation and the occupancy value of a regular tetrahedral grid, Deep Marching Cubes \cite{Liao2018CVPR} and Neural Marching Cubes \cite{chen2021nmc} learn the vertex positions and mesh topologies in a regular grid, Neural Dual Contouring (NDC) \cite{chen2022ndc} predicts the edge crossing and the vertex locations, then extracts a mesh via dual contouring, but can produce non-manifold surfaces. Voromesh \cite{Maruani_2023_ICCV} learns the position of Voronoi cell generators and then extracts the mesh from the boundary of the occupied cells, but produces many small facets.

\begin{figure}[t]
    \center
    \includegraphics[width=0.9\linewidth]{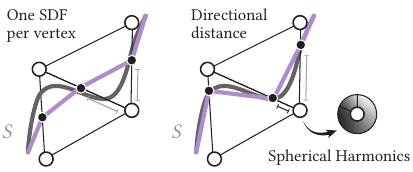}
    \caption{Storing a single SDF value at each point can lead to inaccurate surface reconstruction (left), where a directional distance encoded as spherical harmonics allows our method to place mesh vertices differently on each tetrahedral edge, resulting in a more accurate shape reconstruction (right).}
    \label{fig:directional_sdf}
    \Description{TODO}
\end{figure}

{Methods for adaptive grid generation in isosurface extraction are well-established. While most approaches combine grid subdivision with local refinement to preserve grid quality \cite{LiuAdaptive1995,Hui1999TetAdaptive,Kim2000PolygonizationON,Bey1995TetrahedralGR,Ju2024Adaptive}, their process follows purely geometric criteria. Because our method focuses on mesh optimization, we refine the grid using optimization-derived error signals, typically from 2D inputs, and use an energy-based approach to produce a fair tessellation. Delaunay triangulation has been used in the context of isosurface extraction techniques with adaptive grids. \citet{Zhao2021progressivedomains} progressively refine a 3D Delaunay triangulation to enable coarse-to-fine isosurface extraction. Their method aims at reconstructing a mesh from a point cloud, while \ourmethod{} is designed for gradient-based mesh optimization and reconstruction pipelines. As such, their refinement is driven by a pure geometric analysis based on local curvature estimation, surface smoothness and grid fairness, which can cause their reconstruction to struggle on fine details. In contrast, our error-based refinement strategy is more robust to this limitation and can generalize to different meshing targets.} McGrids \cite{mcgrids2024} (Monte-Carlo grids) generates a mesh by extracting it from the Delaunay triangulation of an iteratively grown point cloud. Although their approach starts with a similar concept to ours, their objectives and respective contributions differ significantly: McGrids start from a given SDF and develop contributions for sampling the SDF using a Monte-Carlo method. We jointly optimize the SDF and the background grid for an unknown shape. This presents unique challenges in optimization, as this setting is less constrained, but offers opportunities to influence the mesh quality and adaptivity through application-specific metrics.

To overcome the limitations of a fixed background grid {for gradient-based mesh optimization}, DMTet \cite{shen2021dmtet} and FlexiCubes \cite{shen2023flexicubes} jointly optimize the implicit representation and the spatial grid structure from which they extract the final mesh. While FlexiCubes can achieve adaptive meshing by leveraging an octree structure instead of a uniform grid, this approach requires specific constraints to reduce the occurrence of non-manifold edges--without guaranteeing their complete removal. Furthermore, FlexiCubes does not provide a principled method to determine where the voxel grid should be refined. These \emph{grid-adaptive isosurface representations} are closest to our work. We focus on removing some of the limitations of these methods, such as high memory consumption due to poor scaling properties and over-parameterized flexibility, susceptibility to self-intersections {(see \figref{fig:self_intersections})} and challenges with adaptive meshing. 

\section{Shape representation}
\label{sec:shape_representation}

{We start with a description of the differentiable shape representation of \ourmethod. In \autoref{sec:reconstruction}, we show how this representation is used in an optimization pipeline.} \ourmethod's shape representation is strikingly simple: it consists of a point cloud $P = \{p_1, p_2, \ldots, p_n\}$, with $p_i \in \mathds{R}^3$, and each point is associated with a base signed distance value $s_i \in \mathds{R}$ and a feature vector $\bm{c}_i \in \mathds{R}^q$. From this representation, our algorithm constructs a surface mesh in three steps, illustrated in \autoref{fig:pipeline_figure}:
\begin{enumerate}
    \item From $P$, compute a tetrahedral grid $(P, T)$ using Delaunay triangulation.
    \item For the endpoints of edges whose base signed distances have opposite signs, compute the corresponding directional signed distance $\hat{s}_i(e)$ from $s_i$ and $\bm{c}_i$. We call such edges \emph{active edges}.
    \item Use Marching Tetrahedra to extract the surface mesh $(V, F)$.
\end{enumerate}

\begin{figure}[t]
    \centering
    \footnotesize
    \setlength{\tabcolsep}{0pt}
    \begin{tabular}{cccc}
    \includegraphics[width=0.35\linewidth]{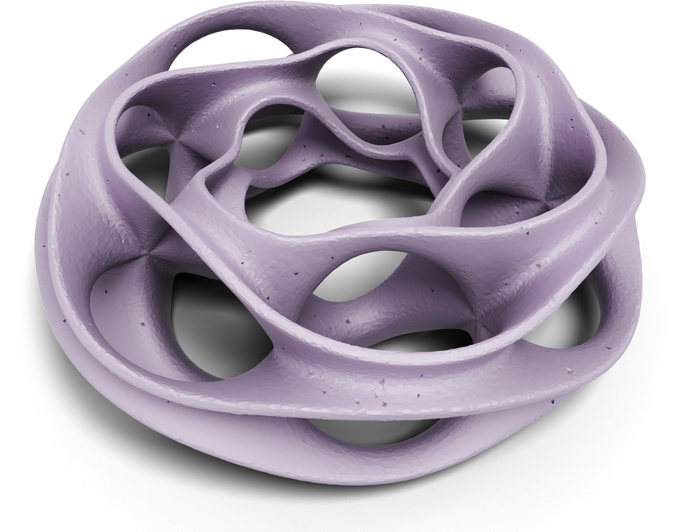}&   
    \includegraphics[width=0.12\linewidth]{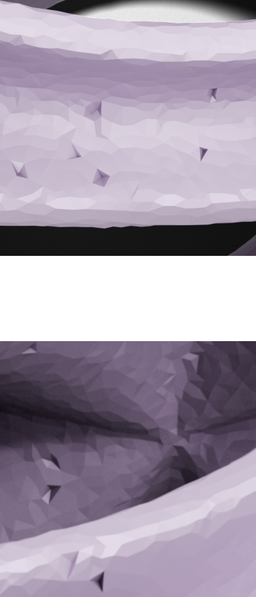} & 
    \includegraphics[width=0.35\linewidth]{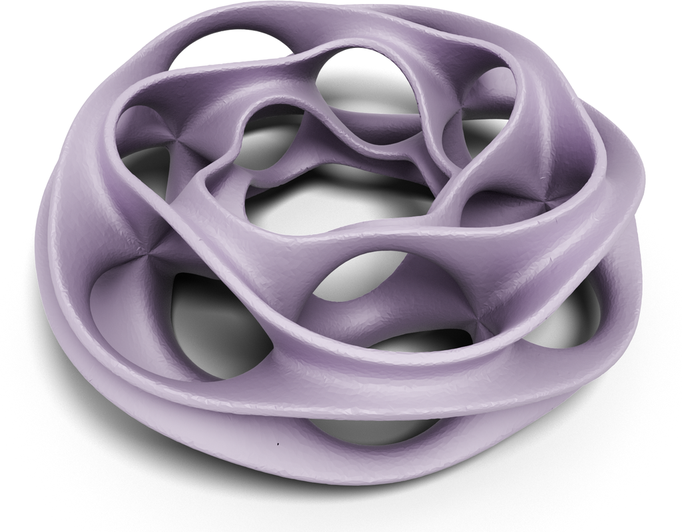}  &   
    \includegraphics[width=0.12\linewidth]{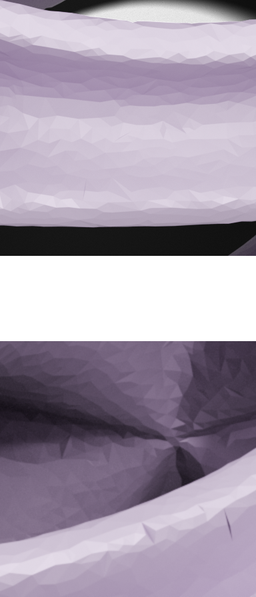} \\
    \multicolumn{2}{c}{No Spherical Harmonics} &
    \multicolumn{2}{c}{Degree 2 SH}
    \end{tabular}
    \caption{{We compare the reconstructed meshes with and without the use of directional signed distance (spherical harmonics).} Using spherical harmonics enhances detail preservation, particularly for shapes with complex topology, where a single grid point can influence multiple parts of the shape.}
    \label{fig:sh_degree_comparison}
\end{figure}

\subsection{Directional signed distance}
\label{sec:directional_sdf}

While \ourmethod{} can generate high-quality meshes with only one signed distance per point, this representation can be suboptimal (see \autoref{fig:directional_sdf}). During the Marching Tetrahedra step, we require the distance along the edges of the tetrahedral grid to place the mesh vertices. These distances differ per edge and storing one distance results in a compromise between all active edges connected to a point. Therefore, we define a \emph{directional} signed distance function $\hat{s}$ to gain more flexibility in positioning mesh vertices. This approach enables our method to position vertices uniquely based on the specific edge (\figref{fig:directional_sdf}), better aligning the local normal with the target surface and effectively preventing artifacts when edges of the tetrahedral grid span separated regions of the mesh (\figref{fig:sh_degree_comparison}).

Our directional signed distance is encoded as spherical harmonics (see \citet{wieczorek2018shtools} for an accessible reference). Each point $p_i$ is associated with a base sign distance $s_i$ and spherical harmonic coefficients $\bm{c}_i \in \mathds{R}^q$ with $q = (d+1)^2$, where $d$ is the desired degree of the spherical harmonics representation. Given an active edge $e$ connecting points $p_i$ and $p_j$, we define the directional signed distance at point $p_i$ for edge $e$ as
\begin{equation}
\label{eq:directional_SDF}
\hat{s}_i(e) = \left( 1+\tanh \left(SH\left(\theta_{i \rightarrow j}, \phi_{i \rightarrow j}, \bm{c}_i \right) \right)\right) s_i,
\end{equation}
where $\theta_{i \rightarrow j}, \phi_{i \rightarrow j}$ are the polar and azimuthal angle of the vector $p_j - p_i$ and $SH(\theta, \phi, \bm{c})$ evaluates the spherical harmonics at the given polar angles for coefficients $\bm{c}$. Note that our formulation enforces that $\hat{s}_i(e)$ has the same sign as $s_i$ for every edge $e$ and takes values in $\left( 0, 2s_i \right)$. These two properties align our directional signed distance with Marching Tetrahedra while allowing for more flexibility.

{We choose spherical harmonics (SH) over alternatives, such as gradient vectors at grid points, for the following reasons: SH are evaluated as a linear combination of basis functions, which is simple to differentiate; by optimizing only low-frequency spherical harmonics, we can enforce a smooth directional function; and it is straightforward to expand the degrees of freedom by adding more coefficients. This can be useful to handle cusps in the SDF (e.g., points between the dragon’s neck and body in \figref{fig:pipeline_figure}).}

\subsection{Mesh extraction}
\label{sec:mesh_extraction}

\ourmethod{} extracts the final mesh using a variant of the Marching Tetrahedra algorithm.
We call \emph{active points} the points of $(P, T)$ that belong to at least one active edge. Each active edge contains exactly one point of the extracted mesh $(V, F)$: given two points $(p_1, p_2)$ from an active edge $e$, the corresponding vertex is given by 
\begin{equation}
\label{eq:directionalMTet}
    v_e = \frac{\hat{s}_2(e) p_1 - \hat{s}_1(e)p_2}{\hat{s}_2(e) - \hat{s}_1(e)}.
\end{equation}

The surface connectivity $F$, is retrieved from a lookup table, based on the sign of the edges, which is the same lookup table as for regular Marching Tetrahedra. {Configurations are displayed in \figref{fig:marching_tets_configurations}.}

\section{Optimization Pipeline}
\label{sec:reconstruction}

{A typical optimization task for \ourmethod{} is multi-view 3D reconstruction, which we use to illustrate the optimization components of our method. Most of the components described in this section could operate in any optimization pipeline where gradients are required. We make clear along the way when a component is specific to multi-view 3D reconstruction.}

\subsection{Multi-view 3D Reconstruction Objective}
\label{sec:multiview_objective}
{
For multi-view 3D reconstruction, we are given a set of input views and corresponding camera intrinsic and extrinsics. We aim to recover a mesh $(V, F)$, representing the surface geometry of the captured object. The set of input views can be photographs captured in the wild (\autoref{sec:photogrammetry}) or, in the case of our validation experiments in \autoref{sec:results}, rendered masks $M_{\text{gt}}$, depth maps $D_{\text{gt}}$, and normal maps $N_{\text{gt}}$ of a given ground-truth mesh $(V_\text{gt}, F_\text{gt})$. We render the mesh resulting from \ourmethod{} with a differentiable rasterizer~\cite{Laine2020diffrast} and compute a loss in image space. We then use autograd to compute the gradients of the loss function w.r.t. the parameters in \ourmethod{} (point positions $p_i$ and their coefficients $s_i$ and $\mathbf{c}_i$) and use gradient descent to minimize the loss function.

For our multi-view 3D reconstruction experiments in \autoref{sec:results}, we use the following loss function (the weights $\lambda_M, \lambda_D, \lambda_N$ are detailed in \autoref{sec:implementation_details})
\begin{equation}
\label{eq:reconstruction}
\begin{split}
\mathcal{L}_{\text{recons}} =\ &\lambda_{\text{M}} \Vert M - M_{\text{gt}} \Vert + \lambda_{\text{D}} \Vert M_{\text{gt}} (D - D_{\text{gt})} \Vert^2 \\ &+ \lambda_{\text{N}} \Vert M_{\text{gt}} (N - N_{\text{gt})} \Vert^2.
\end{split}
\end{equation}
}

\subsection{Regularization}
\label{sec:regularization}
{
We propose to use several regularizers that enhance the quality of the mesh output. These regularizers can be used for any application.
}
Because the grid topology in TetWeave is not fixed, our representation allows full flexibility over the position of the grid points while ensuring that the grid is embedded (i.e., there are no self-intersections). This flexibility allows us to use regularizers that improve the quality of the resulting mesh.
Marching Tetrahedra generally produces meshes of low quality, containing many sliver triangles. This is particularly visible with DMTet. 
We propose two regularization terms in our gradient-based mesh optimization pipelines to encourage fair triangulations.

\subsubsection*{Optimal Delaunay triangulations} {We would like to encourage tetrahedra  in the background grid to be uniform and well-conditioned, which should lead to better elements in the output triangular mesh. Mesh-quality metrics have been extensively studied in the FEM community \cite{shewchuk:hal-04614934}. Our tetrahedral regularizer adopts} the Optimal Delaunay Triangulation (ODT) energy from this literature, which minimizes interpolation error across all possible triangulations for a fixed vertex set \cite{ODTChen2004}. {This is well-suited to our approach because 
\ourmethod{} uses Delaunay triangulation to construct its grid on-the-fly. Since ODT identifies the Delaunay connectivity as optimal for a given point set, it ensures consistency with our pipeline. Furthermore,} \citet{alliez2005variationaltetmeshing} {show that ODT generates well-shaped tetrahedra and} derive a formulation that can be used in an optimization setting for vertex positions, {aligning with our goal of producing fair, adaptive meshes without additional connectivity overhead.}
The ODT energy for a tetrahedron $T_i$ is expressed as $ E_{\text{ODT}}(T_i) = \vert M_{S_{T_i}} - M_{T_i} \vert $, where $ M_{T_i} $ represents the sum of the principal moments of $T_i$ relative to its circumcenter. Similarly, $ M_{S_{T_i}} $ denotes the moment of inertia of $S_{T_i}$, defined as the spherical shell matching the circumsphere of $T_i$ and having an equivalent mass. We define an ODT loss for the tetrahedral grid $(P, T)$ as
\[
\textstyle \mathcal{L}_{\text{ODT}} = \sum_{T_i \in T}{\vert M_{S_{T_i}} - M_{T_i} \vert}.
\]
Minimizing this energy reduces distortions in tetrahedral shapes, leading to more uniform and well-conditioned triangulations. We provide more information regarding implementation and show that this energy vanishes for a regular tetrahedron in {\appref{sec:odt_proof}}.

\begin{figure}
    \centering
    \footnotesize
    \setlength{\tabcolsep}{0pt}
    \begin{tabular}{cccc}
    \includegraphics[width=0.35\linewidth]{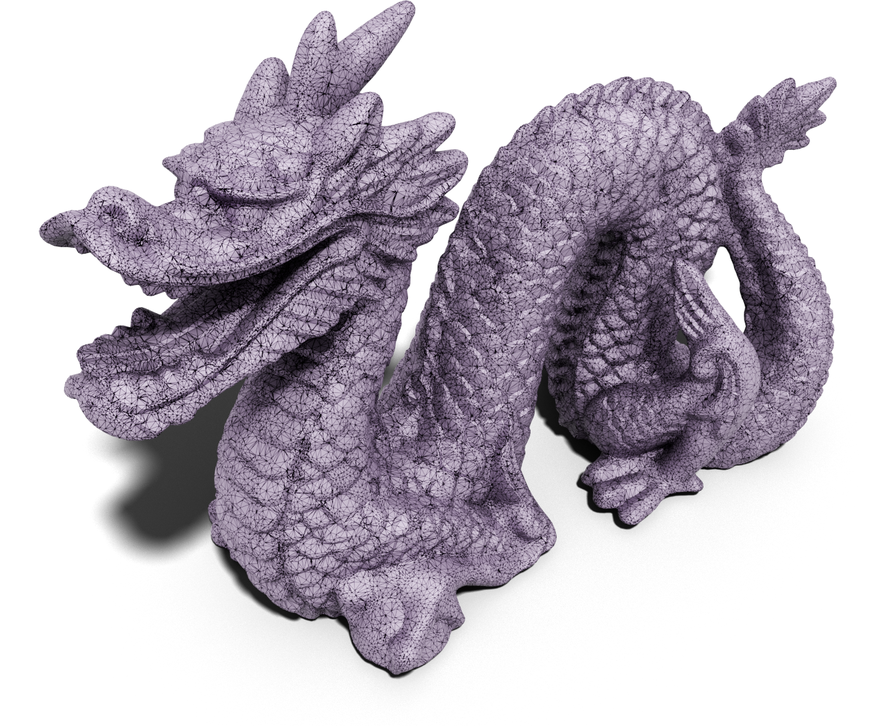} &   
    \includegraphics[width=0.12\linewidth]{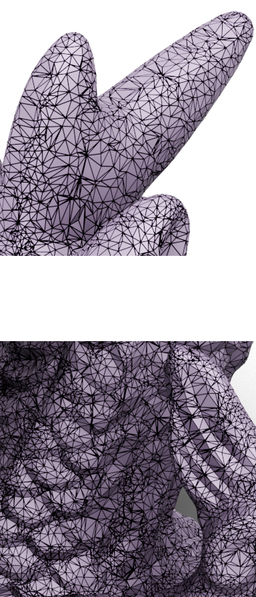} & 
    \includegraphics[width=0.35\linewidth]{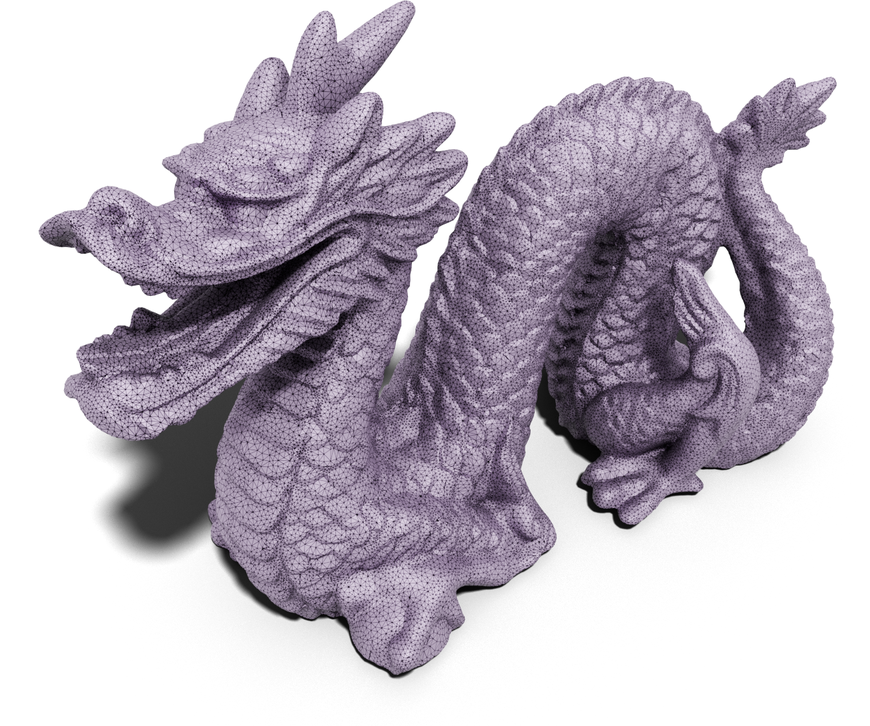}  &   
    \includegraphics[width=0.12\linewidth]{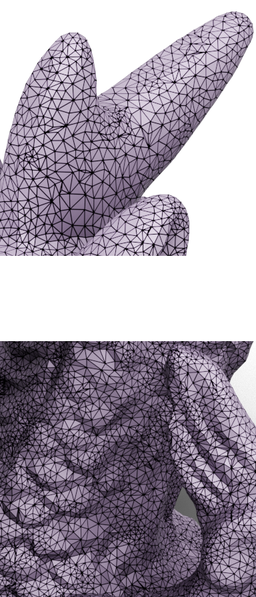} \\
    \multicolumn{2}{c}{No fairness loss} &
    \multicolumn{2}{c}{Fairness loss}
    \end{tabular}
    \caption{{Comparison of reconstructed meshes with and without the use of the fairness loss.} We show that incorporating a simple fairness loss improves tessellation quality without compromising shape fidelity.}
    \label{fig:fairness_loss_comparison}
\end{figure}

\subsubsection*{Triangle fairness loss} We want to encourage the formation of uniform triangles on the extracted mesh. Therefore, we encourage equilateral triangles by penalizing angles that deviate from $\frac{\pi}{3}$:
\[
\textstyle \mathcal{L}_{\text{fairness}} = \sum_{f \in F}{\frac{1}{3}\sum_{i=1}^3{\left( \theta_i - \tfrac{\pi}{3} \right)^2}}.
\]
As shown in \figref{fig:fairness_loss_comparison}, this effectively encourages a fair tessellation.

\subsubsection*{Sign change regularizer} Similar to FlexiCubes \cite{shen2023flexicubes}, our method is prone to spurious geometry in unsupervised parts of the space. We adopt the same regularizer and penalize
sign changes of the base sign distance value for every active edge:
\[
\textstyle \mathcal{L}_{\text{sign}} = \sum_{(a, b) \in E_A}{H\left( \sigma(s_a), \text{sgn}(s_b) \right)},
\]
where $E_A$ designates the set of active edges of the tetrahedral grid, $H$ is the cross-entropy function, and $\sigma$ the sigmoid function.

\subsection{Point cloud refinement}
\label{sec:resampling}

\begin{figure}[t]
    \center
    \includegraphics[width=0.8\linewidth]{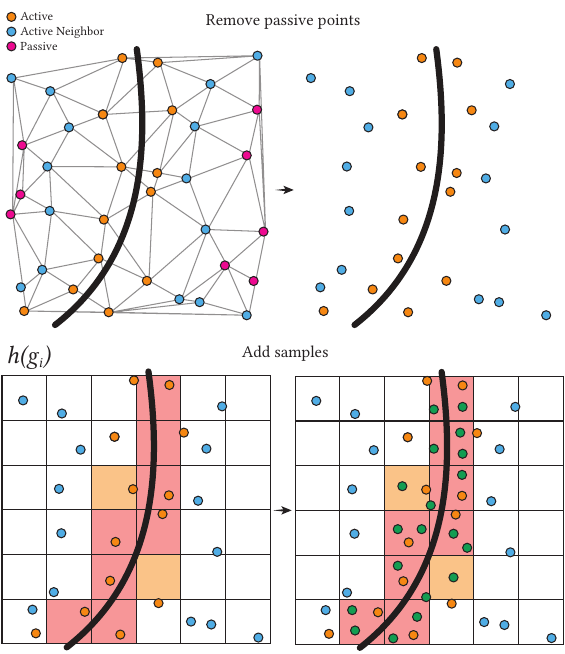}
    \caption{Point cloud refinement is done in two steps. First, we identify and remove points that do not neighbor an active point (passive points). Next, we voxelize the space around the mesh and estimate $h$. Then we sample points within the voxels according to $h$. Finally, the sdf is interpolated to the new points based on a barycentric average of the containing tetrahedron.}
    \label{fig:refinement}
    \Description{TODO}
\end{figure}

{Since \ourmethod{} does not rely on a fixed grid, we can adaptively place points where they are necessary.}
Our approach for sampling focuses on resampling points that are not connected to any active points, referred to as \emph{passive points} (see \figref{fig:refinement}, top row), as they do not contribute to the optimization. In addition to resampling passive points, we incrementally introduce new points where they matter most during the optimization.
We devise a sampling strategy illustrated in \figref{fig:refinement}, bottom row. Given the current mesh $(V, F)$, we compute a bounding box of the mesh, which is subdivided into a voxel grid $\mathcal{G}$. Each voxel $g_i \in \mathcal{G}$ is assigned an importance value $h(g_i)$, which we normalize over the voxel grid to get the probability distribution $\rho(g_i)$.
Hereby, $\rho(\mathcal{G}) = \{ \rho(g_i) \vert g_i \in \mathcal{G}\}$ defines a probability distribution over the grid $\mathcal{G}$.

Given $K$ points to sample, we sample $k_i$, the number of points to sample in voxel $g_i$, from a multinomial distribution parameterized by $\rho(\mathcal{G})$ with $K$ trials. We then randomly sample $k_i$ points in the voxel $g_i$. {For each sampled point, we determine its barycentric coordinates within the containing tetrahedron $T$ and initialize its signed distance and spherical harmonics coefficients via barycentric interpolation of $T$’s vertex parameters.}
{As exemplified in \figref{fig:scaling_ablation}}, we can define $h$ according to specific requirements. For example, uniform sampling is achieved by setting $h(g_i) = 1$ for voxels intersecting the mesh and $h(g_i) = 0$ otherwise. In the following section, we introduce a rendering-based method to estimate $h$. This enables our approach to dynamically adapt the mesh resolution.

Methods that use a predefined grid often use a constant level of detail for every region in space. This leads to significant memory inefficiency when reconstructing surfaces, which sparsely occupy 3D space. With a dense grid representation, much of the space remains unused, and the scaling of output vertices is suboptimal.
For instance, in the case of FlexiCubes, doubling the grid resolution results in a cubic increase of the number of parameters, but only a quadratic increase in the number of output vertices. As a result, FlexiCubes scales poorly and requires using an octree-based datastructure to recover high-frequency details on meshes. DMTet can address this issue by subdividing tetrahedra close to the current mesh but lacks a systematic approach for doing so.

\subsection{Adaptive meshing}
\label{sec:adaptive}

We propose an importance function $h$, suitable for inverse rendering applications, {such as multi-view 3D reconstruction,} showcased in \figref{fig:adaptive_meshing}. The underlying strategy -- to use error as a guide for sampling -- can be applied outside of inverse rendering. We suppose that we have a current mesh $(V, F)$ produced by \ourmethod{}. Let $\mathcal{I}$ be a rasterizer that takes as input camera parameters $\theta_k$ for view $k$. We assume that target images $\tilde{I}_k$ are given, e.g., photographs or renders of an object we aim to reconstruct. For each pixel $(u, v)$, we compute the error $E_k(u, v) = \vert \tilde{I}_k(u, v) - \mathcal{I}(V, F, \theta_k)(u, v) \vert$. Next, each pixel is projected onto the current mesh in $\mathbb{R}^3$ using rasterization. The error values for all pixels that land in voxel grid cell $g_i$ are accumulated over all views and then normalized to compute $h(g_i)$
\begin{equation}
\label{eq:importance}
    h(g_i) = \frac{\sum_k\sum_{(u, v) \in \mathcal{N}_{g_i}} E_k(u, v)}
    {|\mathcal{N}_{g_i}|},
\end{equation}
where $\mathcal{N}_{g_i}$ is the set of all pixels $(u, v)$ that are contained in voxel cell $g_i$. Pixels that land outside the mesh are ignored.
Next to resampling passive points, we progressively add points to our shape representation. Therefore, the resulting meshes adapt to the importance function $h$. \figref{fig:adaptive_meshing_comparison} illustrates an example that demonstrates that denser tessellation occurs in high-frequency areas, enhancing reconstruction in those regions.

\begin{figure}[t]
    \center
    \includegraphics[width=1.0\linewidth,trim={15 65 0 35
    },clip]{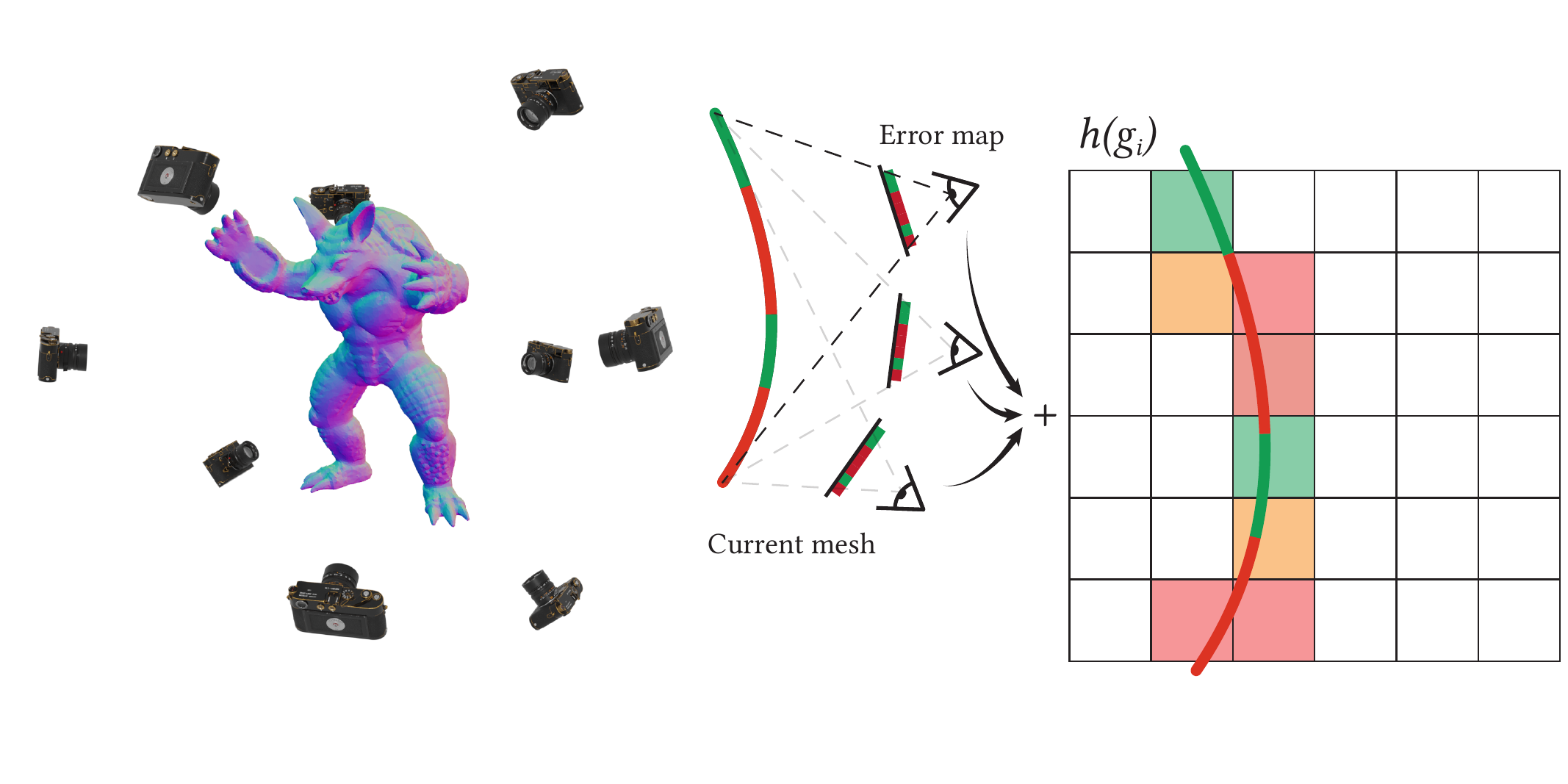}
    \caption{Our refinement strategy can be applied to adaptive meshing. By rendering the shape from multiple viewpoints, we compute the error for each pixel. Since each visible pixel corresponds to a point on the shape's surface, we accumulate these errors within the voxel that contains the corresponding point. We normalize the accumulated errors across all voxels to define the importance value function $h$.}
    \label{fig:adaptive_meshing}
\end{figure}

\subsection{Multi-stage optimization}
\label{sec:stages}

\begin{figure}
    \centering
    \footnotesize
    \setlength{\tabcolsep}{0pt}
    \begin{tabular}{cccc}
    \includegraphics[width=0.35\linewidth]{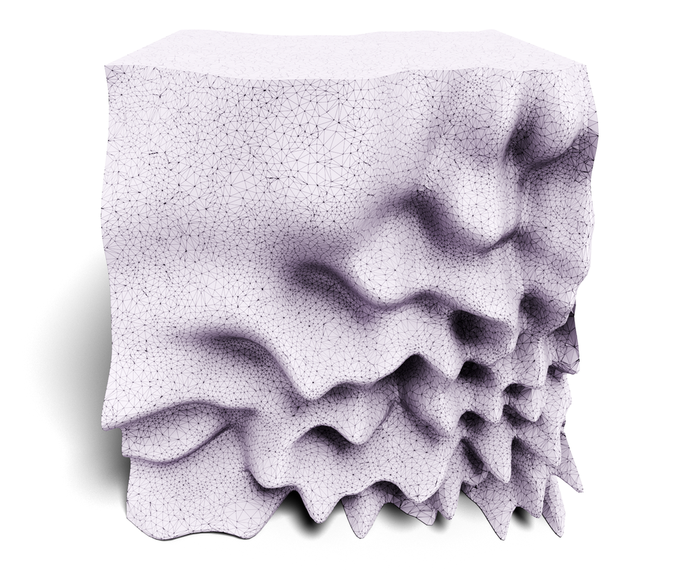} &   
    \includegraphics[width=0.12\linewidth]{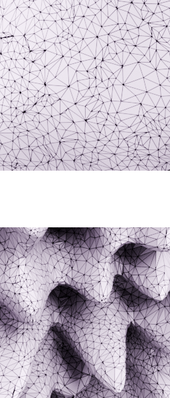} & 
    \includegraphics[width=0.35\linewidth]{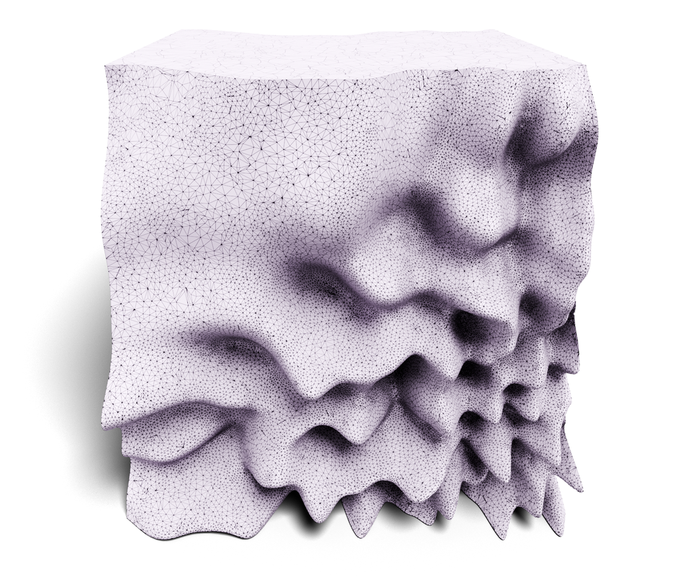}  &   
    \includegraphics[width=0.12\linewidth]{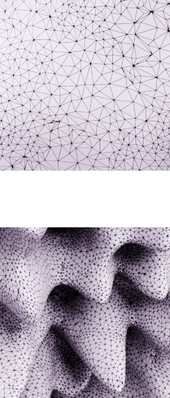} \\
    \multicolumn{2}{c}{Uniform meshing} &
    \multicolumn{2}{c}{Adaptive meshing}
    \end{tabular}
    \caption{{Comparison between uniform and adaptive meshing.} Resampling based on normal map error enables curvature-adaptive sampling, concentrating more triangles in high-curvature regions, which enhances reconstruction quality.}
    \label{fig:adaptive_meshing_comparison}
\end{figure}

We propose a multi-stage optimization for \ourmethod. The increased flexibility of our method and the higher grid resolution around the reconstructed surface results in more degrees of freedom. As reported by the authors of NVDiffRec (\citet{hasselgren2022nvdiffrecmc}, section 8.5) and observed in our own experiments, these degrees of freedom could lead to noisy geometry when optimized naively. Another motivation to implement a multi-stage approach is to limit the number of Delaunay triangulation calls, {which improves the overall speed}.

During the \emph{main stage} {(5000 iterations)}, we update both grid point positions and SDF values, and apply all regularizers. The number of grid points is incrementally increased via resampling until a target point number is reached. As the number of points grows, recomputing the Delaunay Triangulation becomes computationally expensive. To mitigate this, we only recompute the Delaunay Triangulation {every $m$ iterations by default, and keep point positions fixed between updates. We sum point positions' gradients over these $m$ iterations before updating them, which ensures that the grid remains non-degenerate while incorporating information from multiple viewpoints, akin to gradient accumulation in machine learning pipelines.}

We also set up a \emph{late stage} {(2000 iterations)} which acts as a fine-tuning stage to ensure higher-fidelity results. Point positions are fixed, and the Delaunay triangulation is no longer recomputed. The optimization focuses solely on refining the signed distance values and spherical harmonics coefficients. We also do not use the ODT and triangle fairness regularizers. As shown in \figref{fig:multistage_comparison}, the late stage incurs a significant boost in reconstruction quality.

\begin{figure}
    \centering
    \footnotesize
    \setlength{\tabcolsep}{0pt}
    \begin{tabular}{cccc}
    \includegraphics[width=0.35\linewidth]{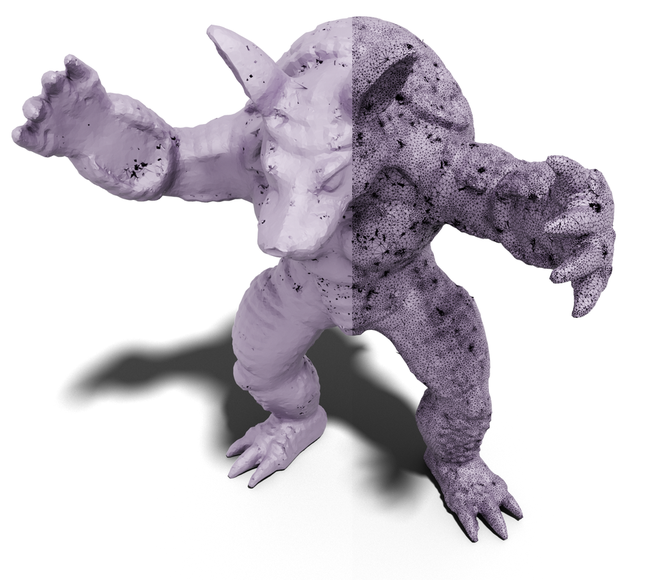} &  
    \includegraphics[width=0.12\linewidth]{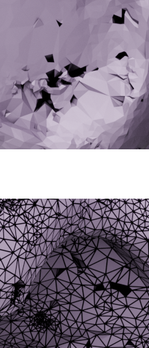} & 
    \includegraphics[width=0.35\linewidth]{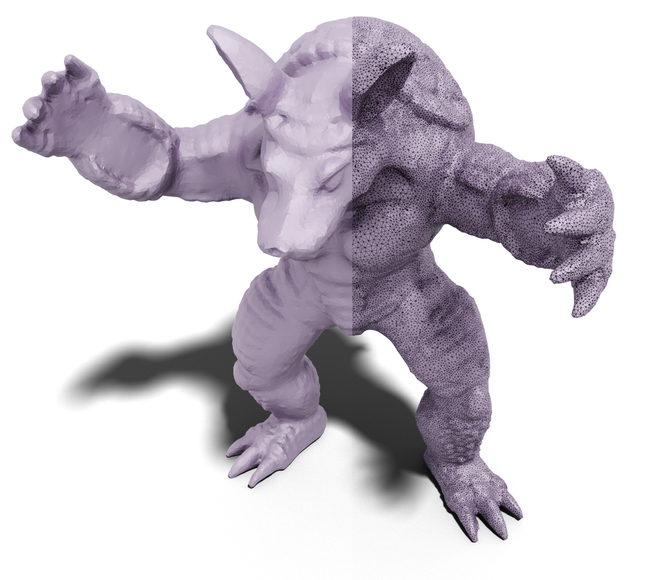}  &   
    \includegraphics[width=0.12\linewidth]{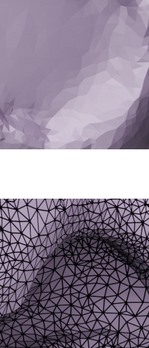} \\
    \multicolumn{2}{c}{No multi-stage training} &
    \multicolumn{2}{c}{Multi-stage training}
    \end{tabular}
    \caption{{Comparison between the resulting meshes with and without a multi-stage pipeline during optimization.} Recomputing the Delaunay triangulation only periodically and fixing the grid to focus on SDF optimization in a second stage is necessary to avoid artifacts.}
    \label{fig:multistage_comparison}
\end{figure}

\begin{figure*}[ht]
{
\centering
\footnotesize
\setlength{\tabcolsep}{0pt}
\begin{tabular}{cccccc}
DMTet \cite{shen2021dmtet} & FlexiCubes \cite{shen2023flexicubes} & \ourmethod{} & \ourmethod{} & \ourmethod{} & \multirow{2}{*}{Reference} \\
$128^3$ grid & $128^3$ grid & 16K points & 64K points & 128K points\\
  \includegraphics[width=0.16\linewidth]{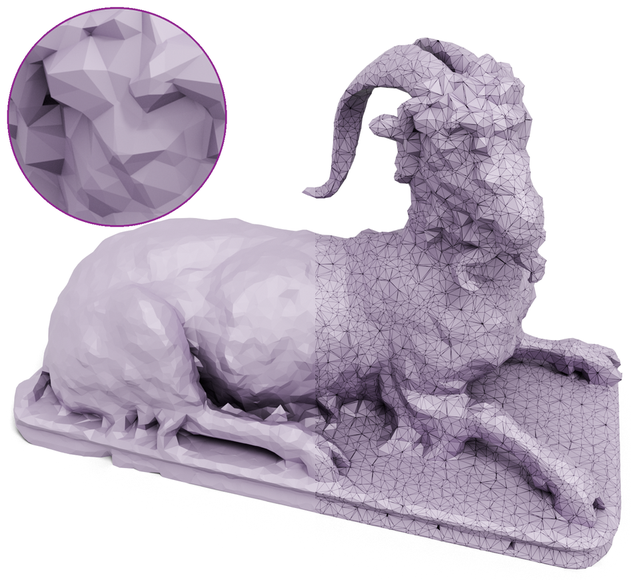}
  & \includegraphics[width=0.16\linewidth]{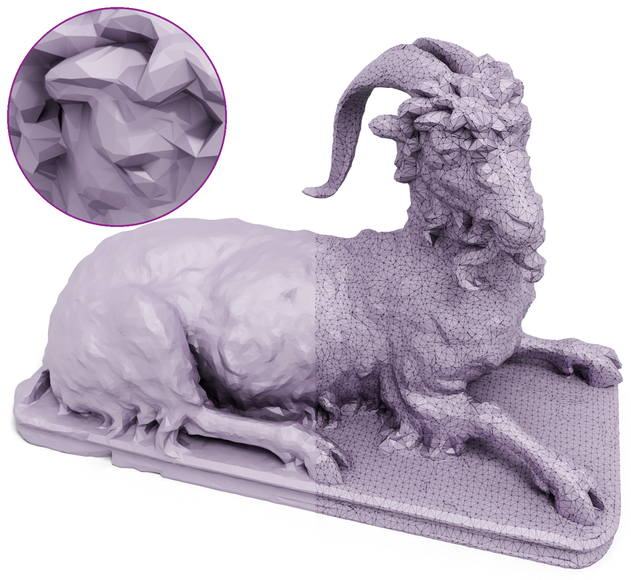}
  & \includegraphics[width=0.16\linewidth]{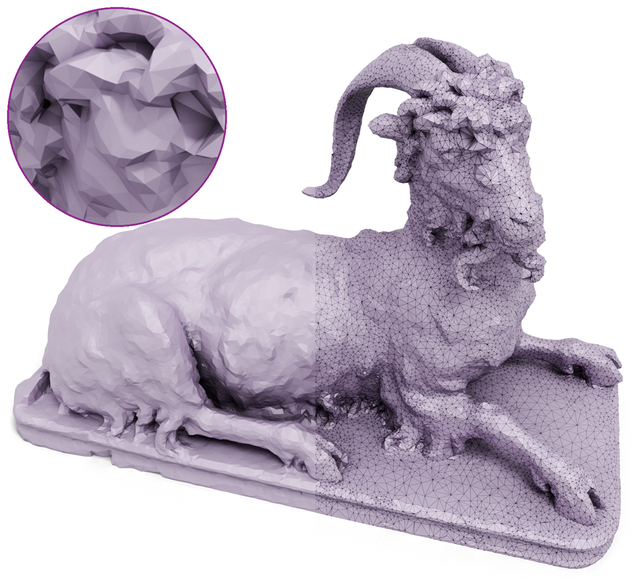}
  & \includegraphics[width=0.16\linewidth]{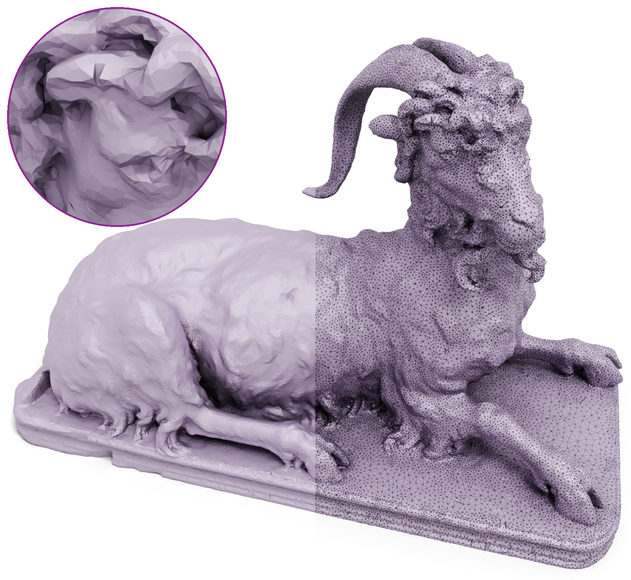}
  & \includegraphics[width=0.16\linewidth]{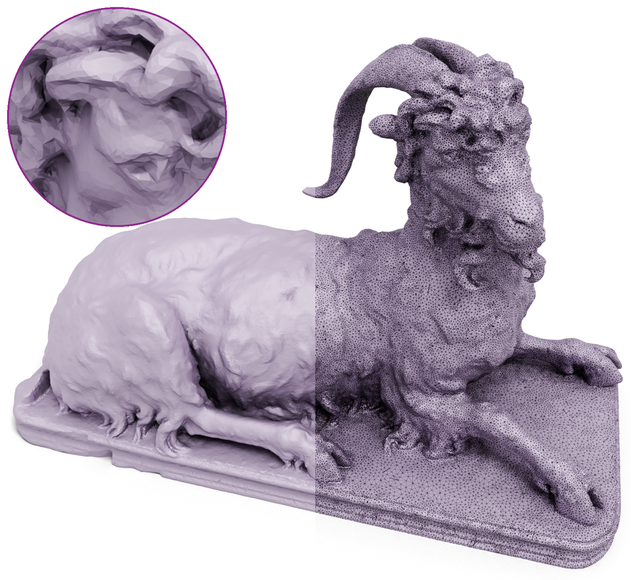}
  & \includegraphics[width=0.16\linewidth]{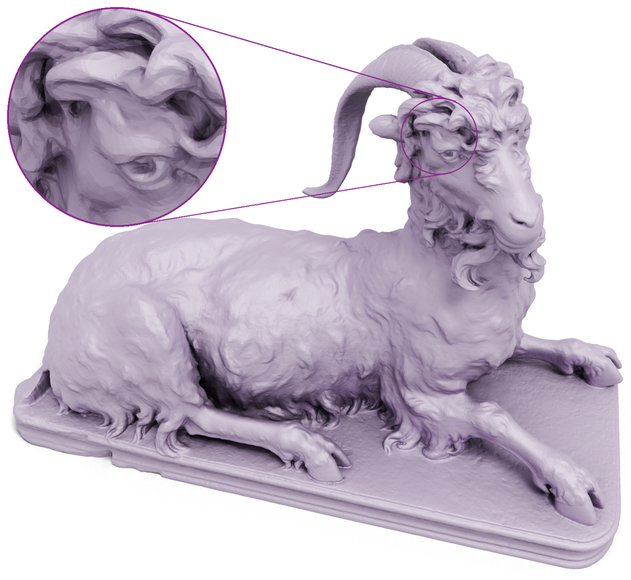} \\

  \includegraphics[width=0.16\linewidth]{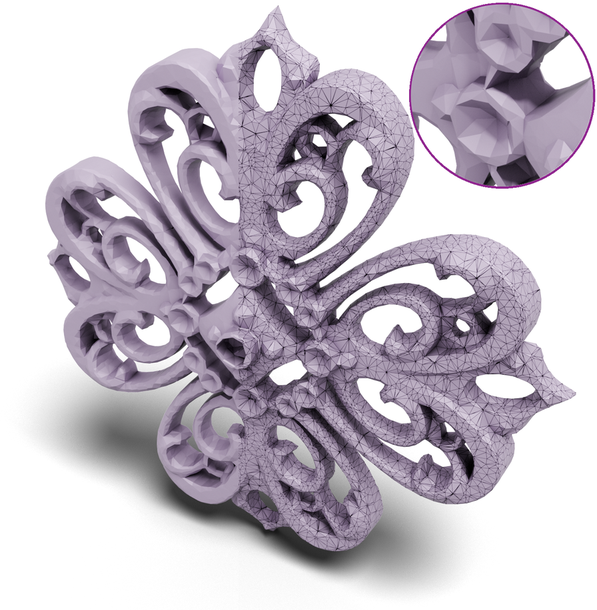} & \includegraphics[width=0.16\linewidth]{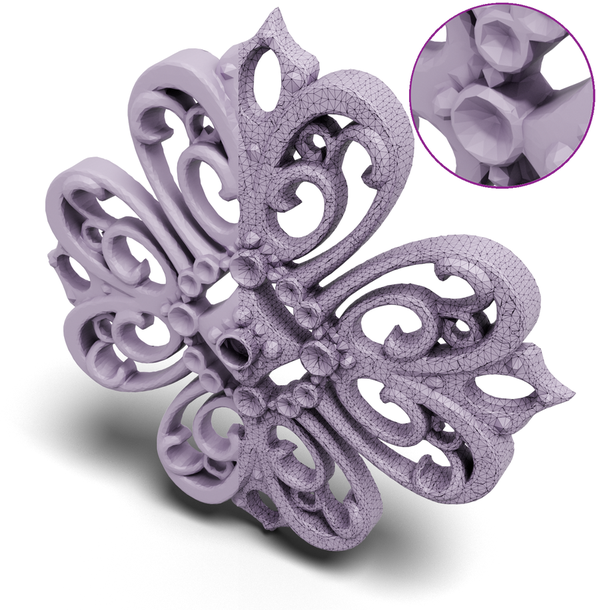} & \includegraphics[width=0.16\linewidth]{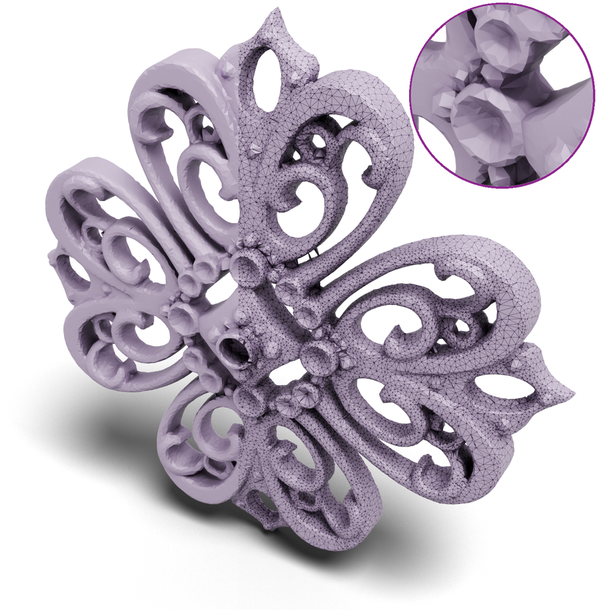} & \includegraphics[width=0.16\linewidth]{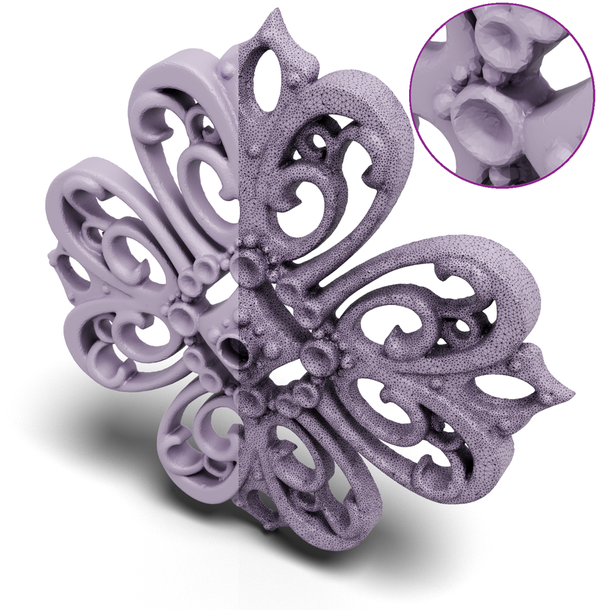} & \includegraphics[width=0.16\linewidth]{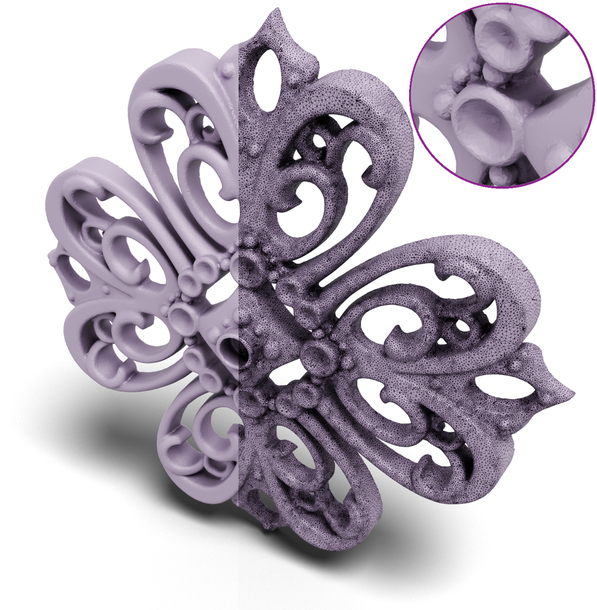} & \includegraphics[width=0.16\linewidth]{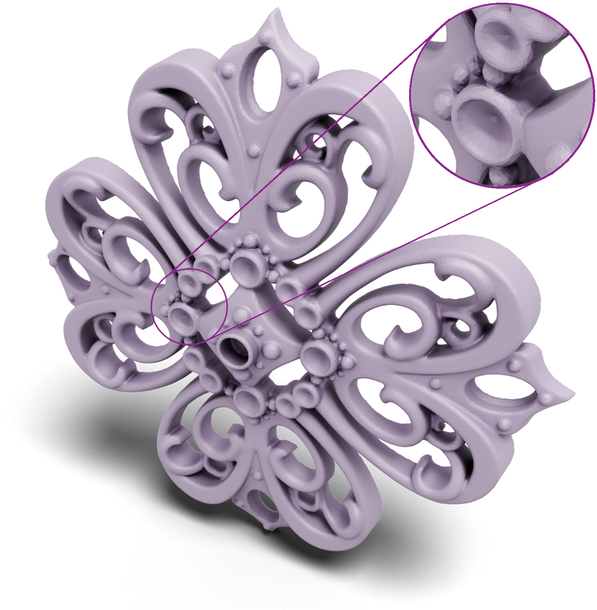}\\

  \includegraphics[width=0.16\linewidth]{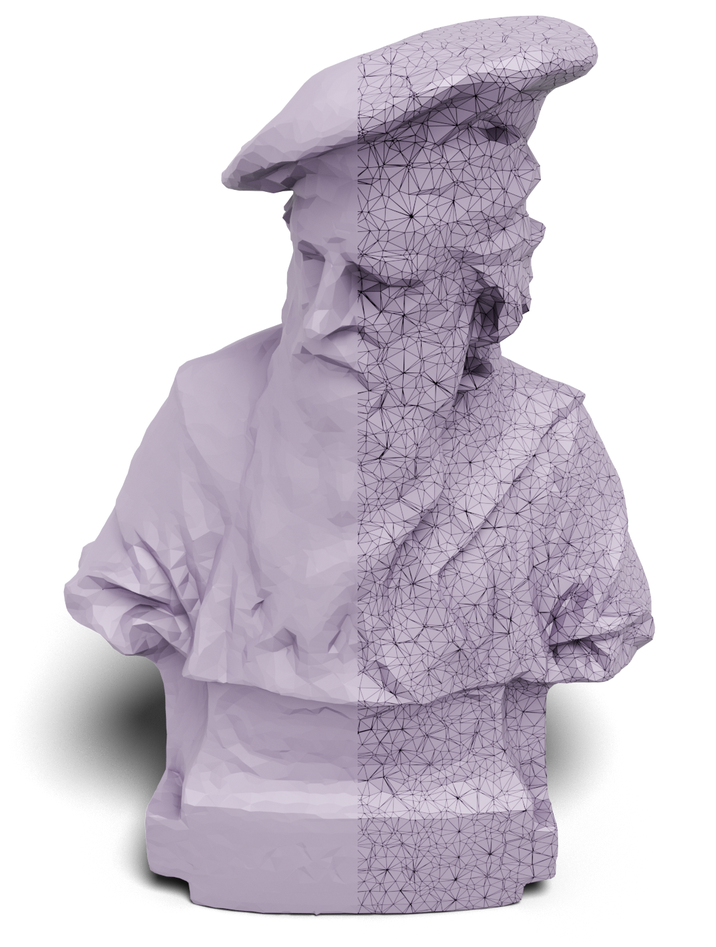} & \includegraphics[width=0.16\linewidth]{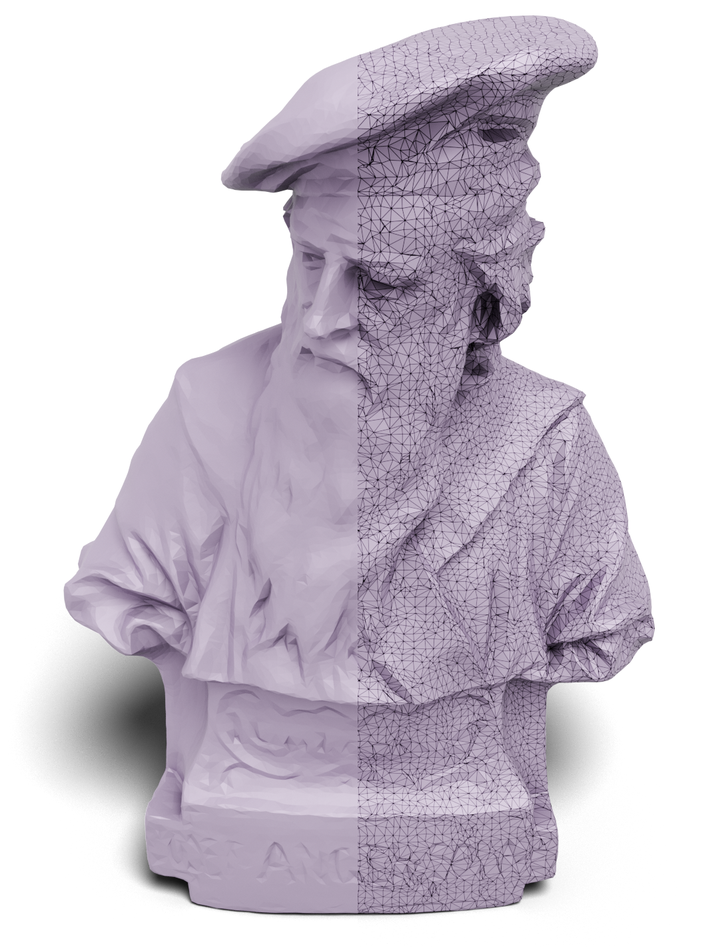} & \includegraphics[width=0.16\linewidth]{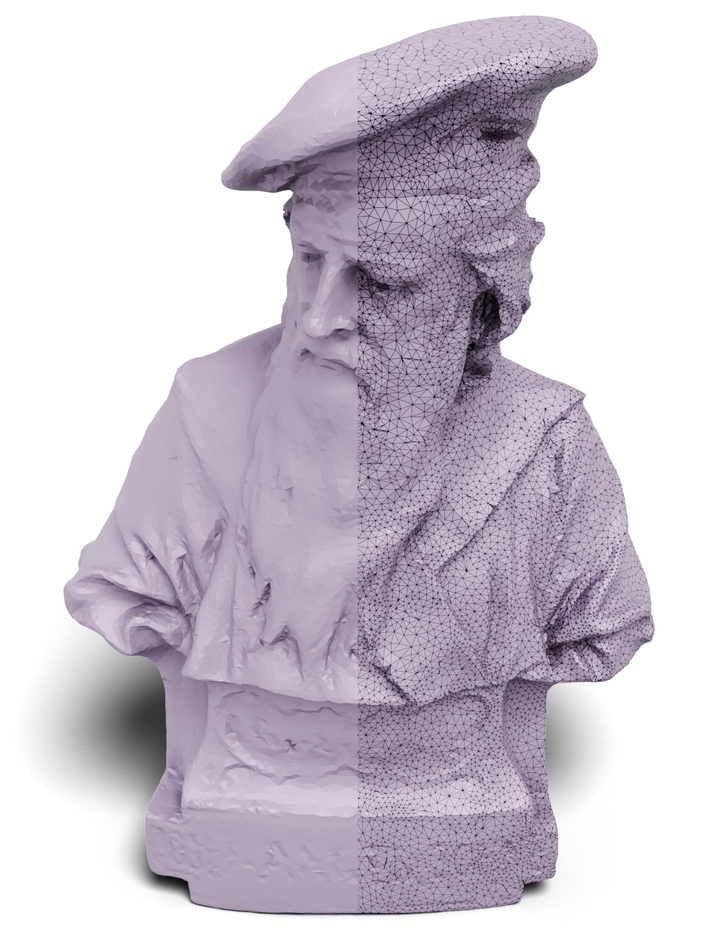} & \includegraphics[width=0.16\linewidth]{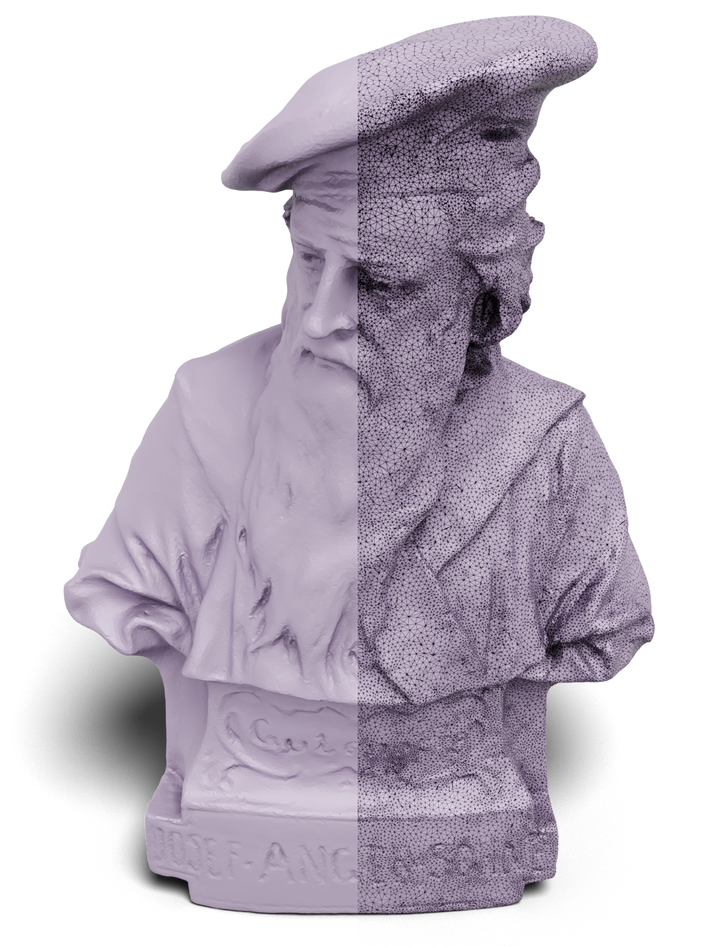} & \includegraphics[width=0.16\linewidth]{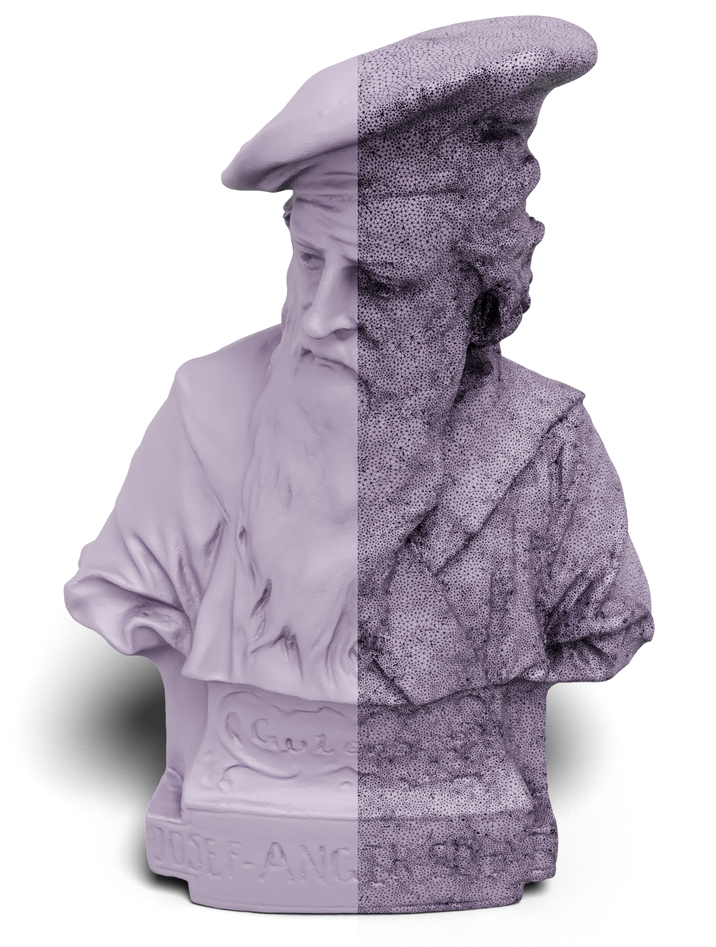} & \includegraphics[width=0.16\linewidth]{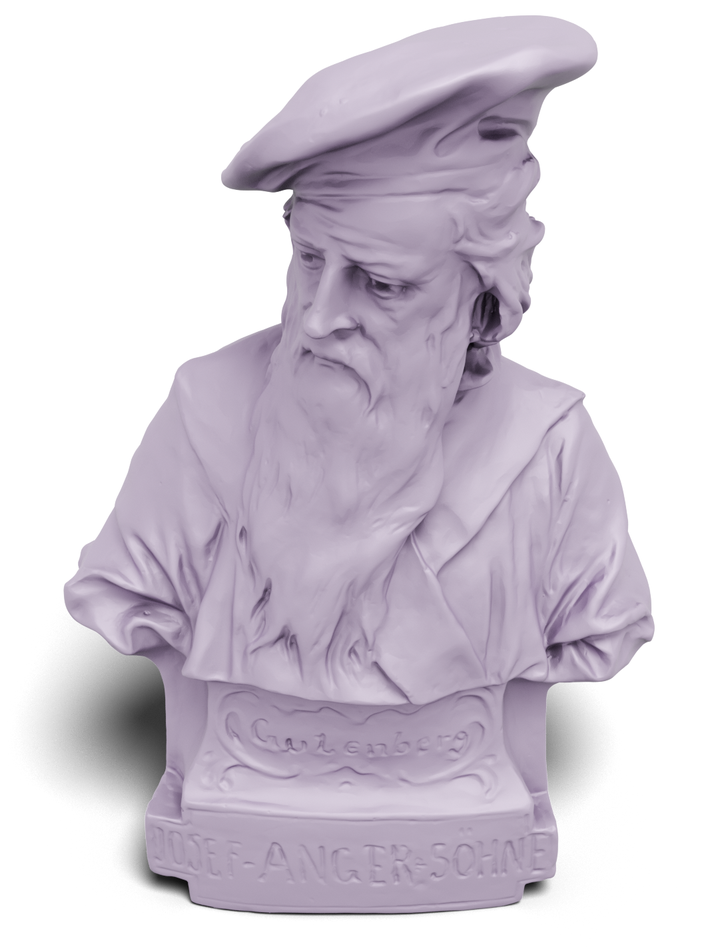}\\

  \includegraphics[width=0.16\linewidth]{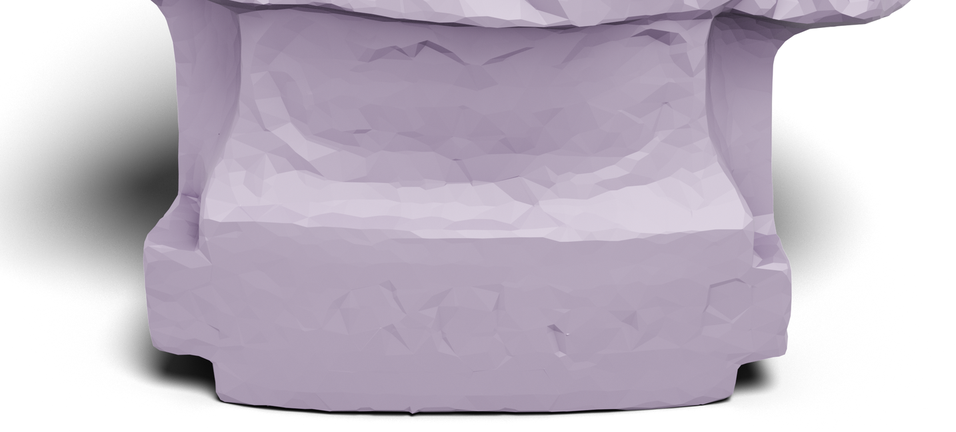} & \includegraphics[width=0.16\linewidth]{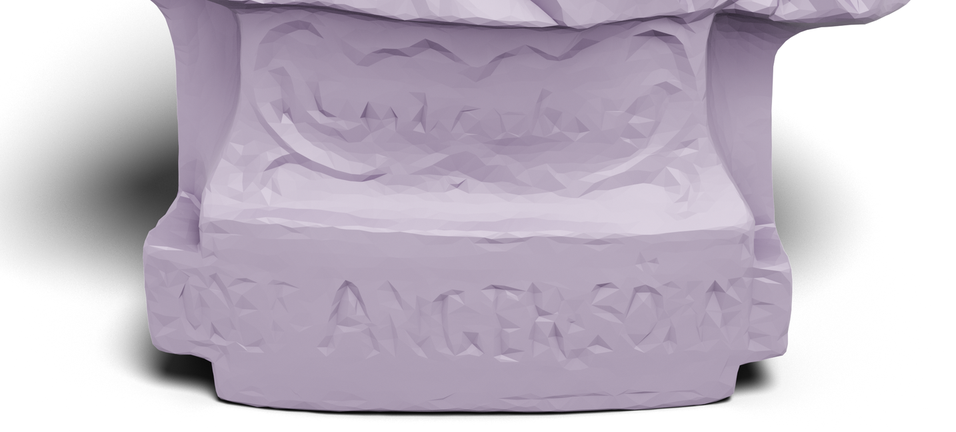} & \includegraphics[width=0.16\linewidth]{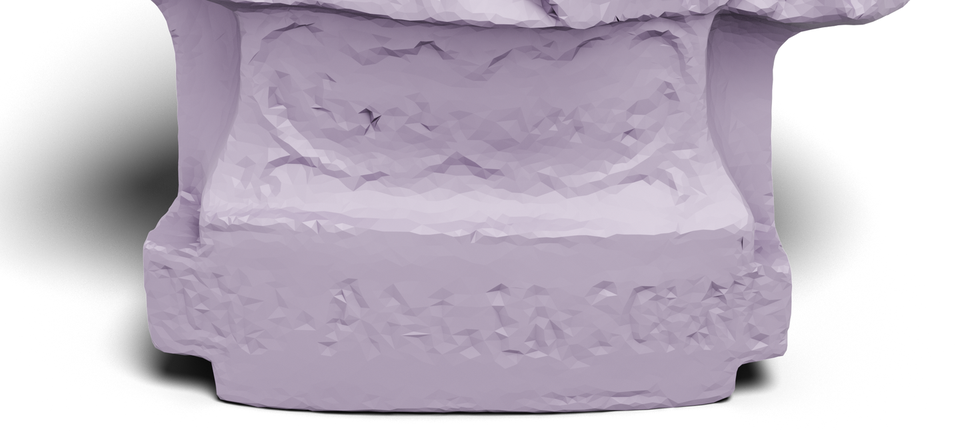} & \includegraphics[width=0.16\linewidth]{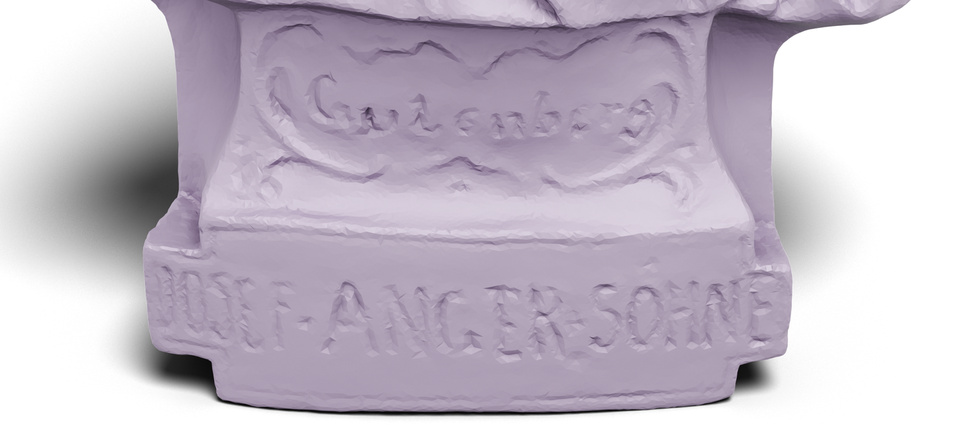} & \includegraphics[width=0.16\linewidth]{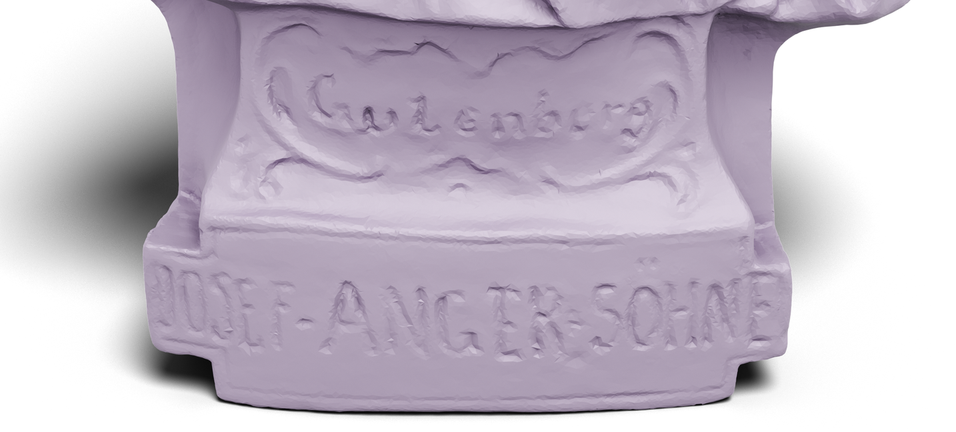} & \includegraphics[width=0.16\linewidth]{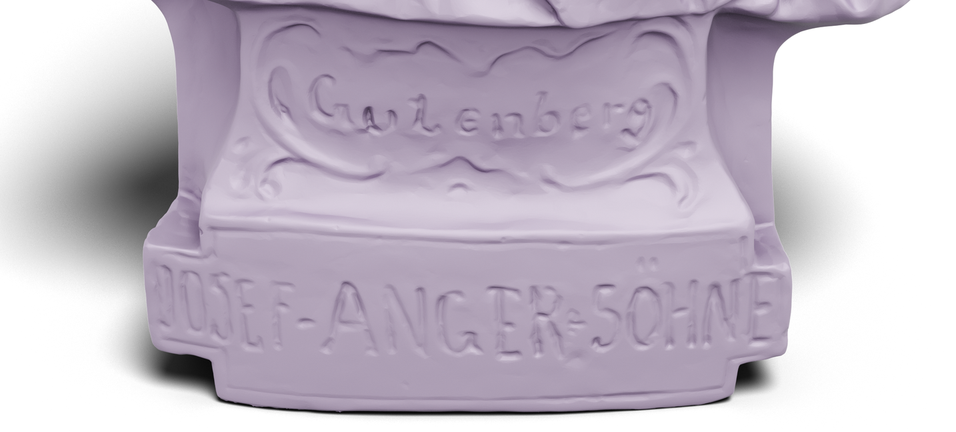}
\end{tabular}
\caption{
Visual comparison of different grid-adaptive isosurface mesh extraction methods, following the experimental setup described in \secref{sec:results}: DMTet, FlexiCubes and \ourmethod{} evaluated at low, mid and high resolutions. Our method achieves high-quality reconstruction across a wide variety of shapes, including highly detailed statues and objects with complex topologies, and excels in capturing intricate high-frequency details, such as the writing on the Gutenberg statue's socle. Furthermore, our approach provides fairer tessellation. We recommend zooming in on a digital display.
}
\label{fig:visual_comparisons}
\Description{Visual comparisons figure.}
}
\end{figure*}

\begin{table*}[t!]
\centering
\scriptsize
\setlength{\tabcolsep}{3pt}
\caption{Quantitative evaluation on the mesh reconstruction task. We sample 1 million points per shape in our dataset. We report the chamfer distance (CD, 1e-5), the F1 score, the edge chamfer distance (ECD, 1e-2), its F1 score (EF1), the normal consistency (NC), the percentage of inaccurate normals (IN$>5^{\circ}$), the percentage of triangles with aspect ratio (AR) and radius ratio (RR) above $4$, the percentage of small angles (SA$<10^{\circ}$) and the percentage of self-intersecting faces (SI). After the generation, only the largest connected component is kept. An expanded table is provided in the appendix.}
\resizebox{1.0\linewidth}{!}{
\begin{tabular}{l c c c c c c c c c c c c}
\toprule
Method & CD $\downarrow$ & F1 $\uparrow$ & ECD $\downarrow$ & EF1 $\uparrow$ & NC $\uparrow$ & IN$>5^{\circ}$(\%) $\downarrow$ & AR$>4$(\%) $\downarrow$ & RR$>4$(\%) $\downarrow$ & SA$<10^{\circ}$(\%) $\downarrow$ & SI(\%) $\downarrow$ & \#V & \#F \\
\midrule

DMTet ($128^3$) & 1.043 & 0.339 & 1.681 & 0.272 & 0.965 & 48.393 & 12.026 & 11.826 & 12.351 & \textbf{0.000} & 20677 & 41364 \\
FlexiCubes ($128^3$)  & 0.752 & 0.416 & 1.254 & 0.393 & 0.979 & 36.911 & 5.418 & 6.701 & 4.588 & 0.203 & 28430 & 56873 \\ 
\ourmethod{} (16K) & 0.517 & 0.409 & 1.475 & 0.353 & 0.974 & 43.380 & 2.350 & 3.344 & 1.743 & \textbf{0.000} & 26484 & 53015 \\
\ourmethod{} (64K) &  0.419 & 0.446 & 0.962 & 0.518 & 0.984 & 33.700 & \textbf{2.251} & \textbf{3.252} & \textbf{1.616} & \textbf{0.000} & 81027 & 162102\\
\ourmethod{} (128K) & \textbf{0.393} & \textbf{0.455} & \textbf{0.708} & \textbf{0.588} & \textbf{0.987} & \textbf{29.361} & 2.507 & 3.556 & 1.829 &\textbf{0.000} & 146514 & 293074 \\
\bottomrule
\end{tabular}
}
\label{fig:metrics}
\end{table*}

\section{Results}
\label{sec:results}

{This section assesses the ability of our shape representation to reconstruct a mesh from unambiguous image input, as discussed in \secref{sec:multiview_objective}.} In our experiments, the ground truth signed distance function is unavailable; the reconstruction is entirely inferred through an inverse rendering pipeline. {We conclude this section with a detailed ablation over several key components of our optimization pipeline.} Practical applications under real-world conditions are explored in \secref{sec:applications}.

\begin{figure*}[ht!]
    \centering
    \small
    \setlength{\tabcolsep}{0pt}
    \begin{tabular}{c c c}
    \includegraphics[width=.33\linewidth]{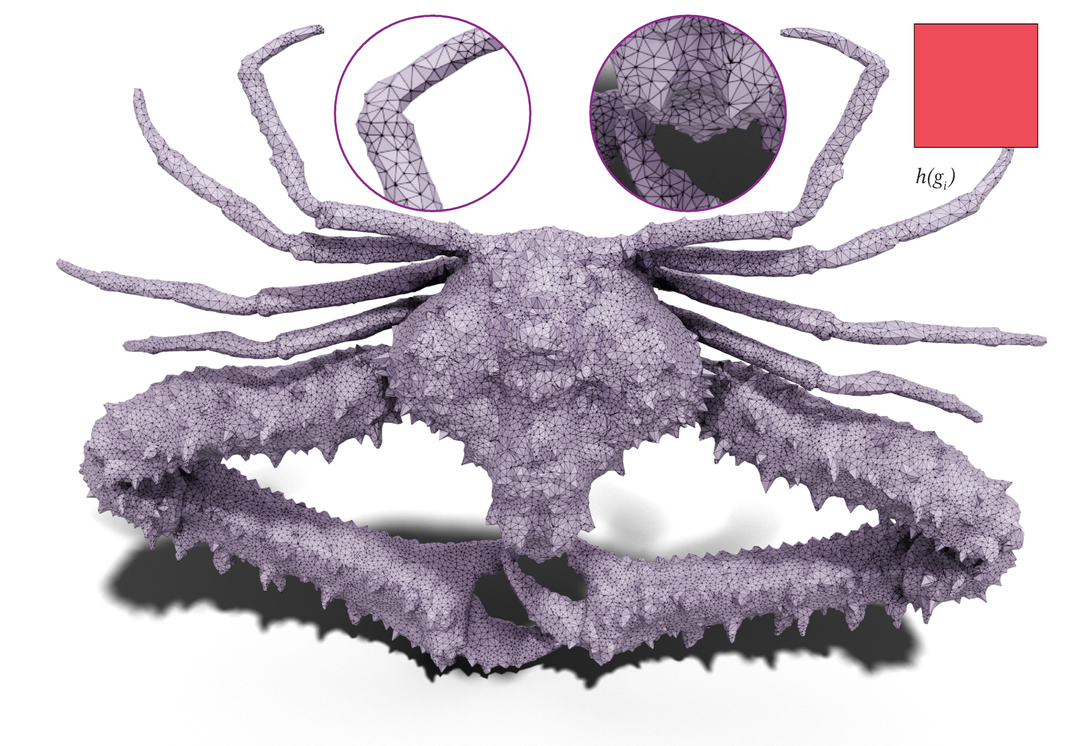} & 
    \includegraphics[width=.33\linewidth]{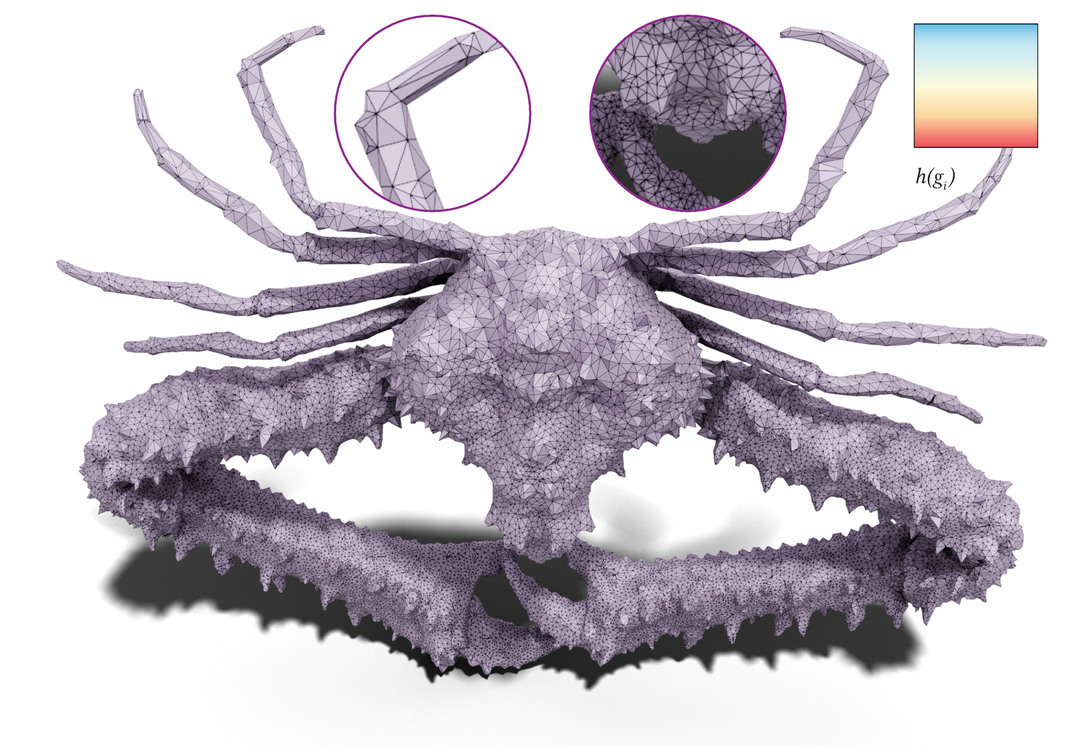} &
    \includegraphics[width=.33\linewidth]{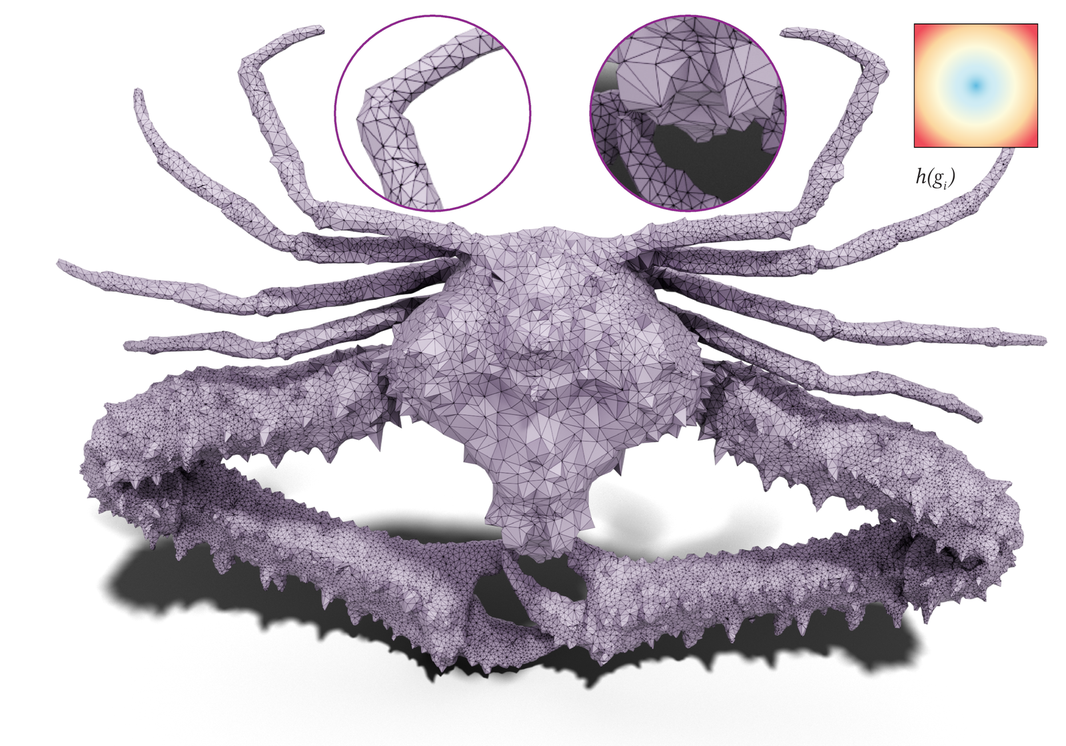} \\
    $h: \mathbf{p} \rightarrow 1$ & $h: \mathbf{p} \rightarrow \vert\mathbf{p}_y-1\vert^3$ & $h: \mathbf{p} \rightarrow \Vert \mathbf{p}\Vert^2$ \\
    uniform & axis-wise increase & radial
    \end{tabular}
    \caption{Our resampling is guided by the importance value function $h$. In our experiments, we compute $h$ based on normal map errors to enable adaptive sampling, but we can use any custom formulation. Here, we present three variations of $h$. Uniform tessellation is achieved by defining $h$ as a constant value. By setting $h$ as a cubic function along the $y$-axis, tessellation density increases progressively from top to bottom. Alternatively, using a radial function results in a higher triangle density in regions farther from the center of the mesh.}
    \label{fig:scaling_ablation}
    \Description{TODO}
\end{figure*}

\subsection{Experimental setting}

To conduct our experiments, we render a mask, depth map, and normal map of the target shape and the mesh generated by \ourmethod{} using Nvdiffrast \cite{Laine2020diffrast}. We employ the regularizers discussed in \secref{sec:regularization}, the refinement from \secref{sec:resampling}, use the $L_1$ loss over 
 the normal maps as a rendering target for adaptive meshing from \secref{sec:adaptive}, and the multistage training method from \secref{sec:stages}.
Additional details, including implementation specifics and parameter settings, are provided in {\appref{sec:implementation_details}}.

Our dataset comprises shapes from the ThreeDScan repository \cite{threedscans}, which offers high-quality meshes of up to more than 2 million vertices. This degree of detail demonstrates \ourmethod{}'s ability to accurately reconstruct high-frequency details. Additionally, the dataset provides a testbed to highlight the potential of our approach for mesh compression applications, as detailed in \secref{sec:compression}.
We preprocess the dataset by retaining only the largest connected component of each shape and by excluding non-watertight models. Each mesh is then normalized. Additionally, we remove redundant shapes, such as multiple versions of the same object.
After processing, the dataset consists of 75 shapes. 


\subsection{Comparisons}

We compare against the two most closely related methods, regarded as the state-of-the-art, DMTet and FlexiCubes. In our comparisons, we only extract the largest connected component from the resulting mesh for each method. This filters out internal geometry, present in each compared method. In the experiments by \citet{shen2023flexicubes}, a different strategy is employed to solve the problem of internal geometry, however, the ground truth mesh itself is used as input. We found our filtering method to work comparably well, while being more widely applicable in realistic applications.

\textit{Quantitative evaluation.\ }
We follow the evaluation method from FlexiCubes, inspired by NDC \cite{chen2022ndc}, and report the average value of the following metrics: the chamfer distance (CD), its F1 score, the edge chamfer distance (ECD), its F1 score (EF1), the normal consistency (NC), the percentage of inaccurate normals (IN$>5^{\circ}$(\%)), the percentage of triangles whose aspect ratio (AR) and radius ratio (RR) is greater than $4$, the percentage of angles below a threshold of $<10^{\circ}$ (SA$<10^{\circ}$), the percentage of self-intersecting triangles (SI), and the number of vertices and faces. The metrics are explained in more detail in {\appref{sec:metric_def}}.

As shown in \tableref{fig:metrics}, our method consistently yields lower chamfer distances, higher normal consistency, and fewer degenerate triangles compared to both DMTet and FlexiCubes at comparable complexity levels. We also exhibit a near-monotonic improvement across all geometry metrics as the number of points increases from 16K to 128K, while producing no self-intersections at any resolution by design, as shown in \figref{fig:self_intersections}. In particular, the lower percentages of angles under $10^\circ$ and triangles with high aspect or radius ratios demonstrate the high geometric quality of our tessellation. These results indicate that \ourmethod{} not only outperforms existing approaches in reconstructing high-frequency details, but also scales more effectively to larger resolutions.

\textit{Visual comparison.\ } Visual comparisons in \figref{fig:visual_comparisons} feature DMTet and FlexiCubes at a $128^3$ grid resolution alongside our method, shown at 16K, 64K, and 128K points in the underlying tetrahedral grid. Notably, FlexiCubes can hardly scale beyond $128^3$ due to GPU memory constraints, while \ourmethod{} can handle far higher resolutions. This difference proves critical for capturing fine details—such as the inscription on the Gutenberg statue’s base, which FlexiCubes cannot reconstruct at its maximum grid size. Additionally, our meshing is more adaptive, using fewer, larger triangles on flat surfaces, and our overall tessellation tends to be fairer. However, FlexiCubes does achieve slightly sharper edges compared to ours at 16K grid points, consistent with its stronger ECD and EF$1$ scores at that resolution.

\begin{figure*}[t]
    \centering
    \small
    \setlength{\tabcolsep}{1pt}
    \begin{tabular}{cccc @{\hskip 0.15cm} cccc}
    \multicolumn{4}{c}{\includegraphics[width=.485\linewidth]{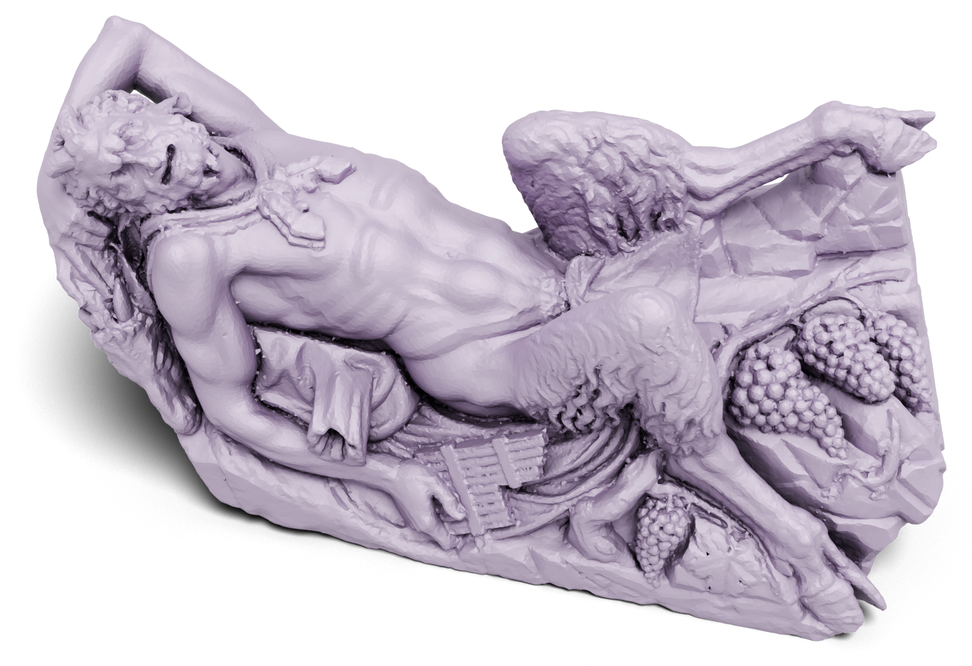}} &
    \multicolumn{4}{c}{\includegraphics[width=.485\linewidth]{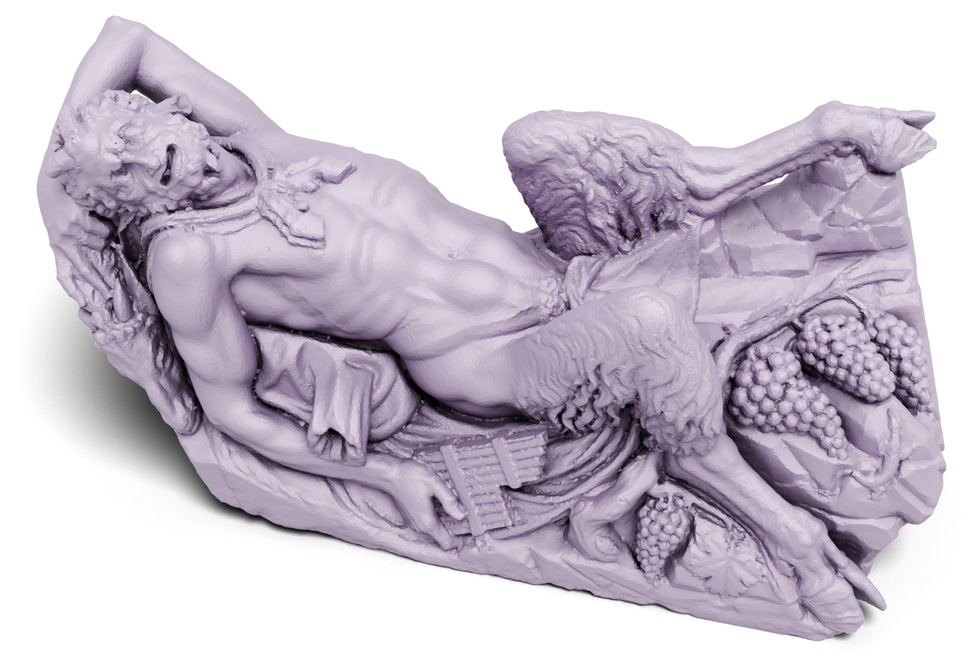}}\\ 
    \includegraphics[width=0.12\linewidth]{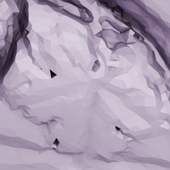} &
    \includegraphics[width=0.12\linewidth]{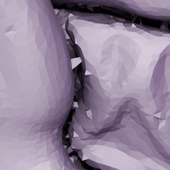} &
    \includegraphics[width=0.12\linewidth]{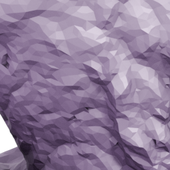} &
    \includegraphics[width=0.12\linewidth]{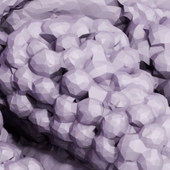} &
    
    \includegraphics[width=0.12\linewidth]{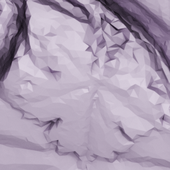} &
    \includegraphics[width=0.12\linewidth]{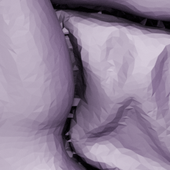} &
    \includegraphics[width=0.12\linewidth]{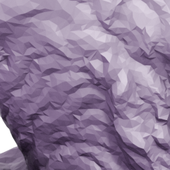} &
    \includegraphics[width=0.12\linewidth]{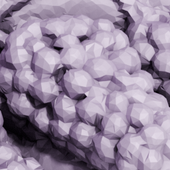}
    \\
    
    \multicolumn{4}{c}{No Spherical Harmonics} & \multicolumn{4}{c}{Spherical Harmonics - degree 1} \\
    \multicolumn{4}{c}{$568$kB -- EF$1$: $0.568$} & \multicolumn{4}{c}{$1.044$MB -- EF$1$: $0.617$}
    \end{tabular}
    \caption{Compressed model obtained using a grid with $128$K points. Incorporating spherical harmonics provides slight improvements at the edges as evidenced by the reported EF$1$ metric. However, it also doubles the memory requirements of the representation.}
    \label{fig:compression_SH}
    \Description{TODO}
\end{figure*}

\begin{figure}[t]
    \centering
    \small
    \setlength{\tabcolsep}{2pt}
    \begin{tabular}{c c}
    \includegraphics[width=.49\linewidth]{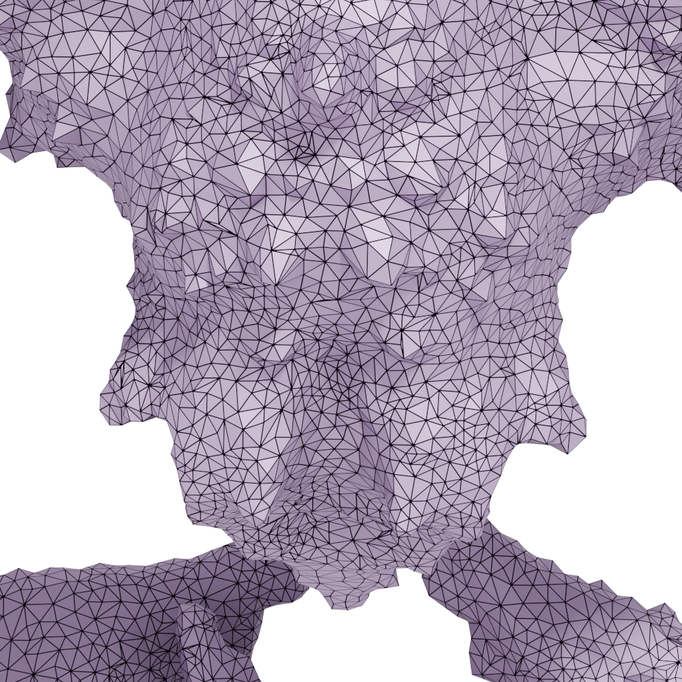} & 
    \includegraphics[width=.49\linewidth]{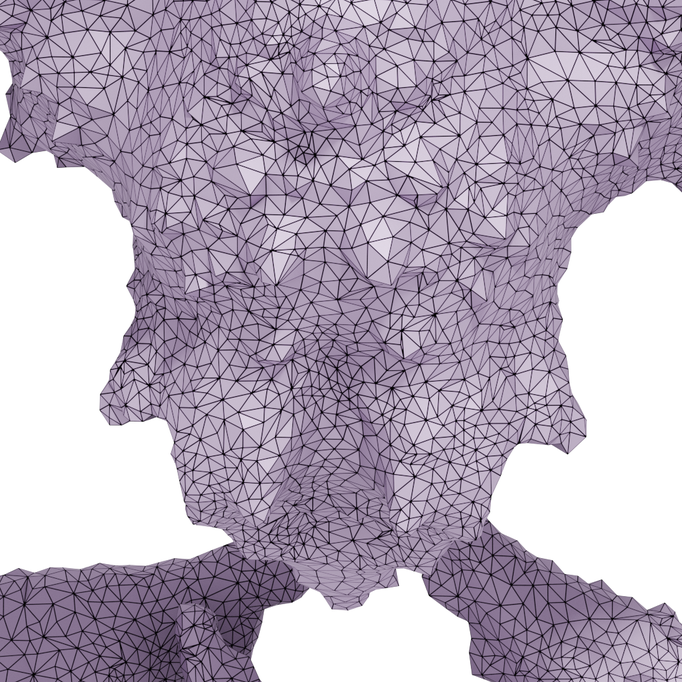} \\
    full point cloud - double precision & trimmed point cloud - single precision \\
    $2.1$MB & $587.3$kB
    \end{tabular}
    \caption{Our shape representation with 64K points and spherical harmonics of degree 1 occupies $2.1$MB when stored in double precision. By discarding non-active points and using single precision storage, the memory footprint can be reduced by nearly a factor of four, with minimal changes that remain imperceptible, as shown in this close-up.}
    \label{fig:compression_simplification}
    \Description{TODO}
\end{figure}

\subsection{Ablation}
\label{sec:ablation}

In this section, we conduct ablation studies on isolated components of our technique to highlight their significance within the pipeline.

\subsubsection{Directional signed distance}

We examine the impact of incorporating spherical harmonics for computing directional signed distances in \tableref{tab:metrics_extended}. While the use of spherical harmonics does not significantly affect the chamfer distance or F1-score, it consistently enhances the EF1-score and reduces the percentage of inaccurate normals. This improvement aligns with our observations: spherical harmonics enable directional adjustments of the signed distance, effectively aligning normals at a highly localized level. While this has minimal influence on chamfer distance, it proves beneficial and is visible in worst-case scenarios, {as shown in \figref{fig:sh_degree_comparison}}.

\subsubsection{Regularization}

We evaluate the impact of regularizers in our method in the last three rows of \tableref{tab:metrics_extended}. Introducing a fairness term substantially improves reconstruction, particularly in triangle quality metrics such as aspect ratio, radius ratio, and the reduction of small-angle percentages. This improvement explains the $36\%$ reduction in vertex count, as the fairness term eliminates triangles with extremely small areas while preserving reconstruction quality. However, this comes at the cost of a slight increase in the percentage of inaccurate normals. Additionally, incorporating an ODT energy improves the edge chamfer distance by over $10\%$.

\subsubsection{Adaptive meshing}

We conducted experiments with \ourmethod{} using our resampling strategy. {As described in \secref{sec:resampling}}, our resampling relies on an importance value function $h$, which operates over a voxel grid decomposition of the current reconstructed shape. In the adaptive setting, $h$ is computed based on rendering errors of the normal map. This approach concentrates points in regions where normals misalign, typically areas of high curvature like highly detailed regions.

In this ablation, we test alternatives of $h$. By setting $h$ to a constant value and $0$ in voxels that do not intersect the shape, we uniformly sample points within each voxel intersecting the mesh. This configuration is referred to as ``Uniform'' in \tableref{tab:metrics_extended}. Compared to adaptive sampling, uniform sampling achieves similar chamfer distance (CD) and F1-score, demonstrating its ability to produce meshes of comparable quality. However, for sharp features, uniform sampling underperforms, as evidenced by significantly worse ECD, EF1, and percentages of inaccurate normals. This highlights the advantage of adaptive sampling in targeting high-frequency details. On the other hand, uniform sampling marginally improves triangle quality, with less than $1\%$ of triangles having an angle smaller than $10^\circ$.

To demonstrate the flexibility of our importance value function $h$, we illustrate three different configurations for $h$ in \figref{fig:scaling_ablation}: $h: \mathbf{p} \rightarrow 1$ corresponds to a uniform meshing strategy, where all regions are sampled equally; $h: \mathbf{p} \rightarrow |\mathbf{p}_y - 1|^3$ introduces an axis-wise variation which increases the density of triangles progressively along the $y$-axis from top to bottom; and $h: \mathbf{p} \rightarrow |\mathbf{p}|^2$ concentrates triangles in regions farther from the center. These examples showcase how $h$ can be tailored to prioritize specific regions of interest, resulting in meshes that emphasize different geometric features of the shape.


\begin{figure*}[htbp!]
	\centering
	\small
	\setlength{\tabcolsep}{1pt}
	\begin{tabular}{cccccccc}
\multicolumn{2}{c}{\includegraphics[width=0.24\linewidth]{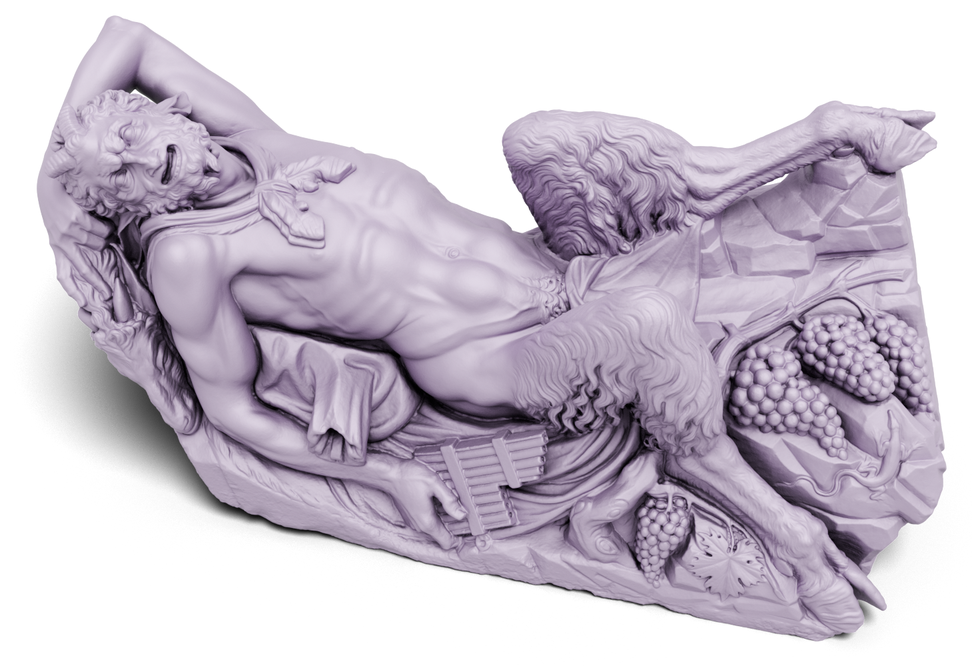}} &
\multicolumn{2}{c}{\includegraphics[width=0.24\linewidth]{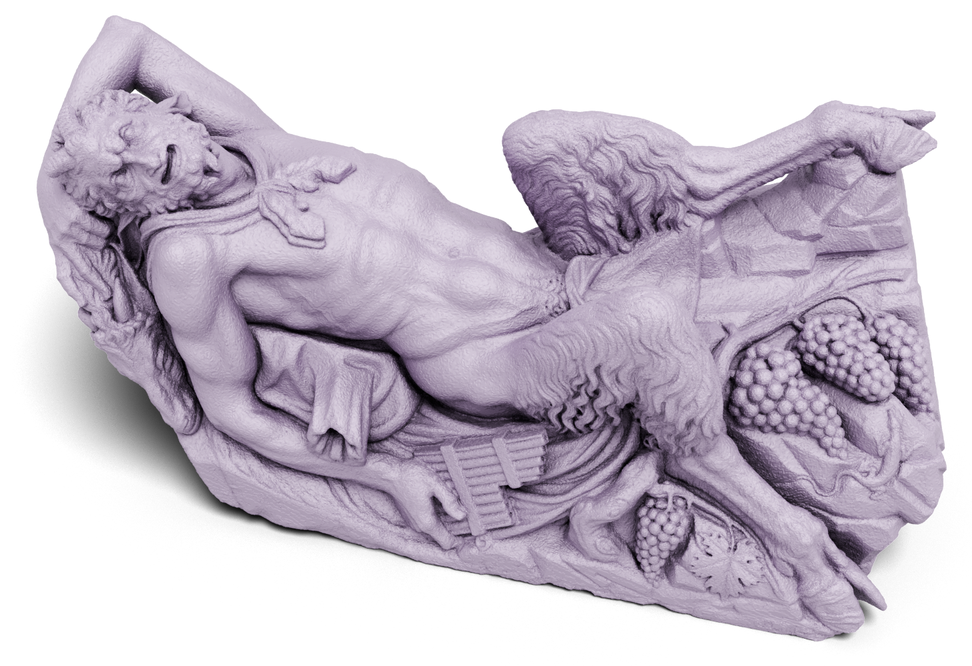}} &
\multicolumn{2}{c}{\includegraphics[width=0.24\linewidth]{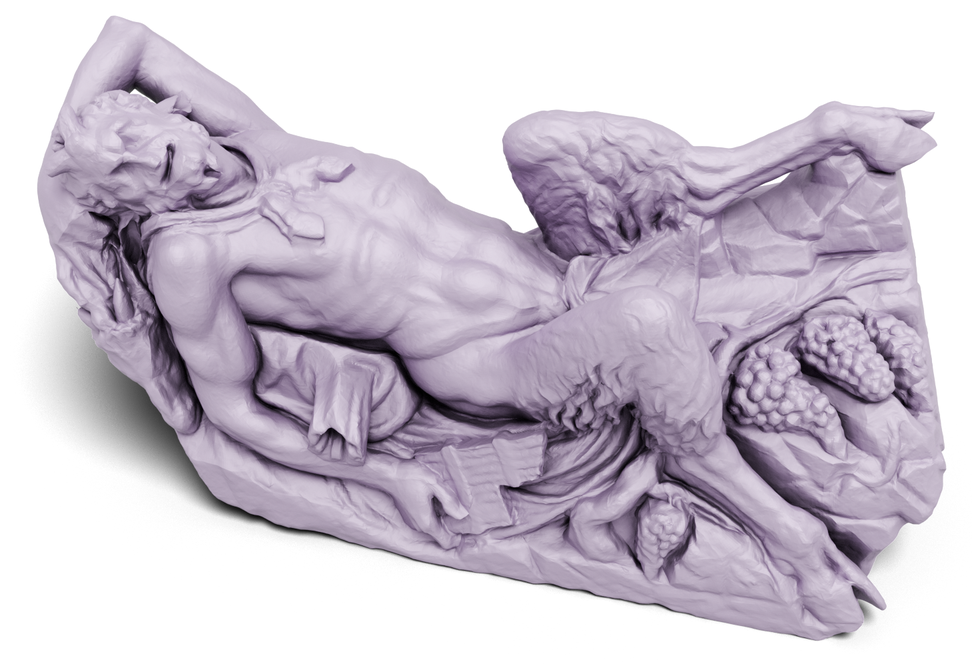}} &
\multicolumn{2}{c}{\includegraphics[width=0.24\linewidth]{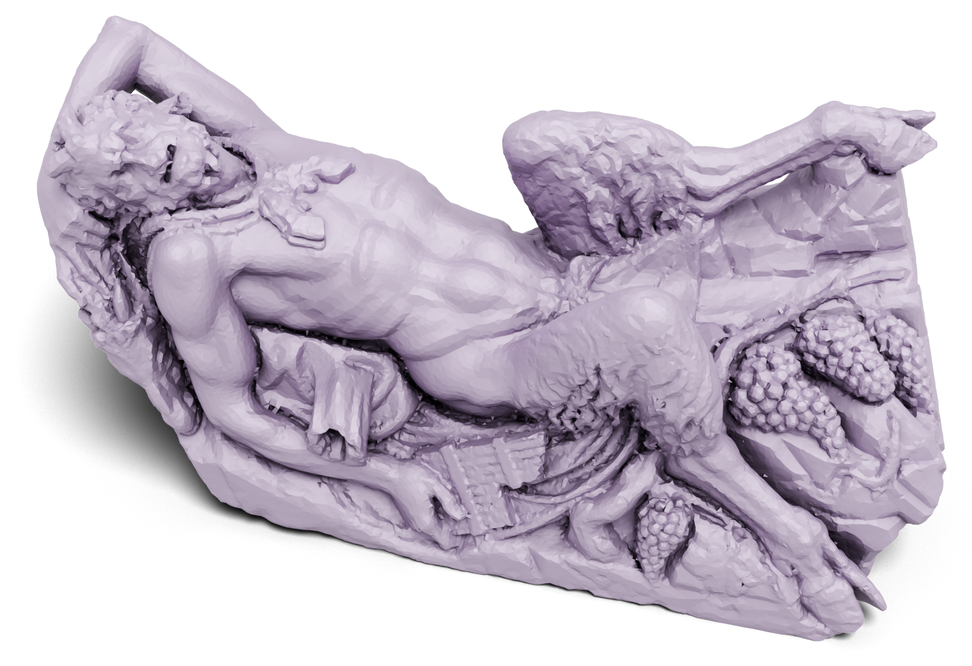}} \\

\includegraphics[width=0.12\linewidth]{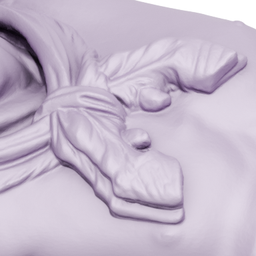} & 
\includegraphics[width=0.12\linewidth]{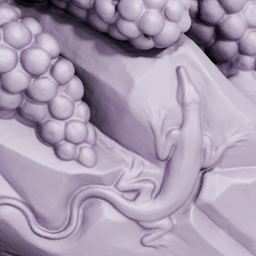} & 
\includegraphics[width=0.12\linewidth]{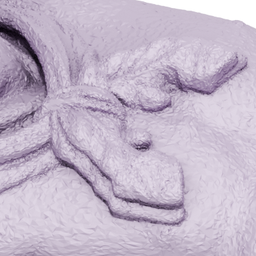} & 
\includegraphics[width=0.12\linewidth]{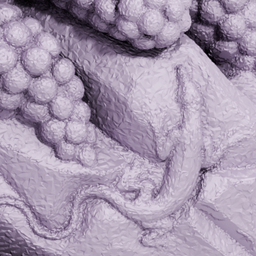} & 
\includegraphics[width=0.12\linewidth]{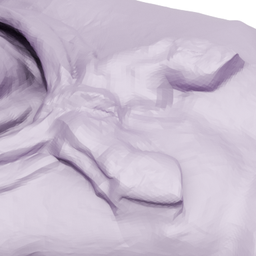} & 
\includegraphics[width=0.12\linewidth]{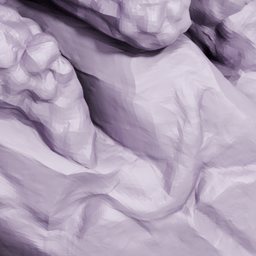} & 
\includegraphics[width=0.12\linewidth]{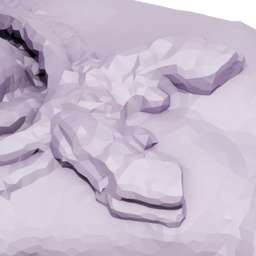} & 
\includegraphics[width=0.12\linewidth]{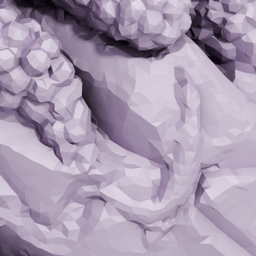} \\

\multicolumn{2}{c}{Reference} &
\multicolumn{2}{c}{Draco \cite{draco2017}} &
\multicolumn{2}{c}{NGF \cite{vs2024ngfs}} &
\multicolumn{2}{c}{\ourmethod{} - 64K, no SH} \\

\multicolumn{2}{c}{37.4 MB} &
\multicolumn{2}{c}{2.3 MB} &
\multicolumn{2}{c}{267.3 kB} &
\multicolumn{2}{c}{319.8 kB}
    \end{tabular}
	\caption{Comparison of our representation against two mesh-compression techniques, Draco \cite{draco2017} and Neural Geometry Fields (NGF) \cite{vs2024ngfs}, on a high-resolution statue. The reference mesh occupies 37.4MB with single precision. Draco compresses it to 2.3MB by quantizing vertex connectivity, but suffers from high-frequency artifacts. NGF requires only 267.3kB but produces non-adaptive patches, can over-smoothen the result and self-intersect. Our approach achieves a 319.8kB file size by discarding non-active points, using 16-bit floats, and inferring connectivity from Delaunay triangulation.}
	\label{fig:mesh_compression}
    \Description{Compression figure}
\end{figure*}

\begin{figure*}[htbp!]
	\centering
	\small
	\setlength{\tabcolsep}{1pt}
	\begin{tabular}{cccccc}
\includegraphics[width=0.16\linewidth]{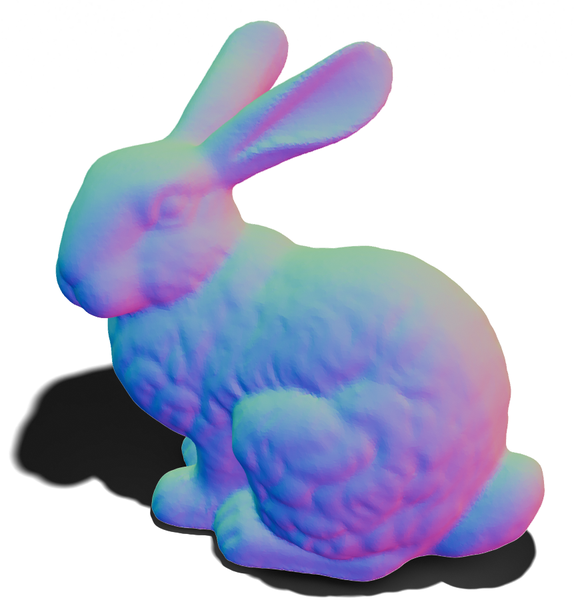} &
\includegraphics[width=0.16\linewidth]{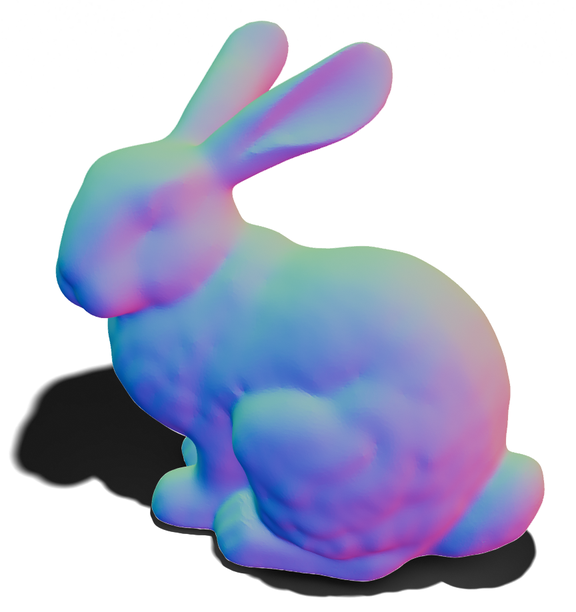} &
\includegraphics[width=0.16\linewidth]{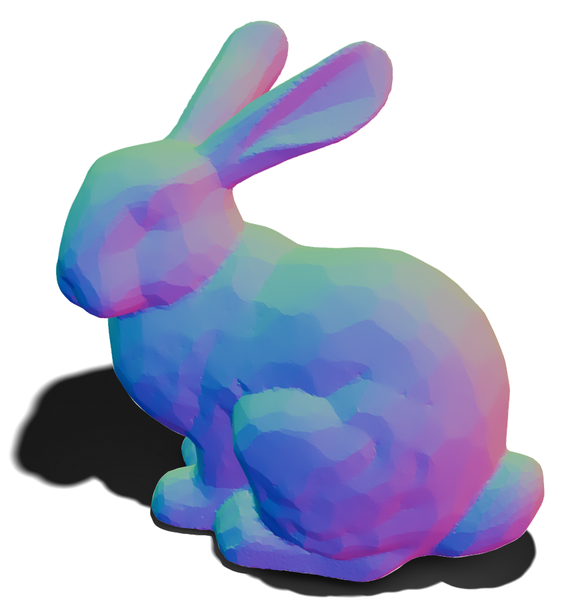} &
\includegraphics[width=0.16\linewidth]{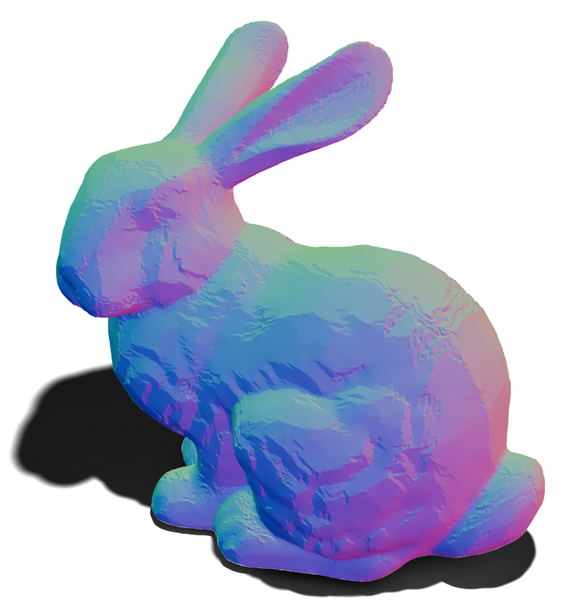} &
\includegraphics[width=0.16\linewidth]{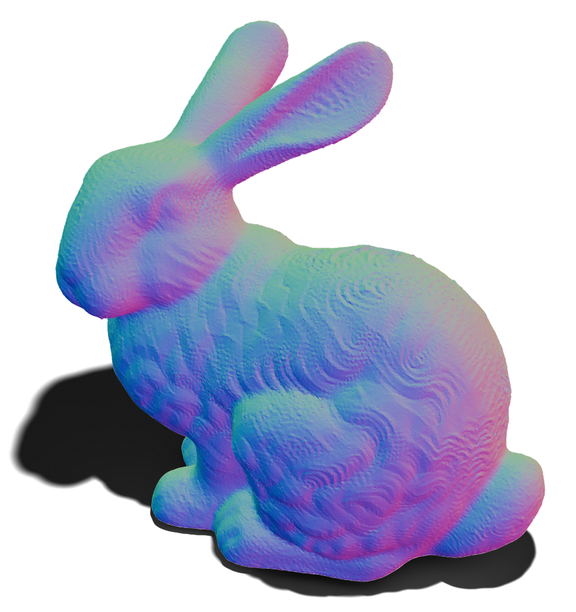} &
\includegraphics[width=0.16\linewidth]{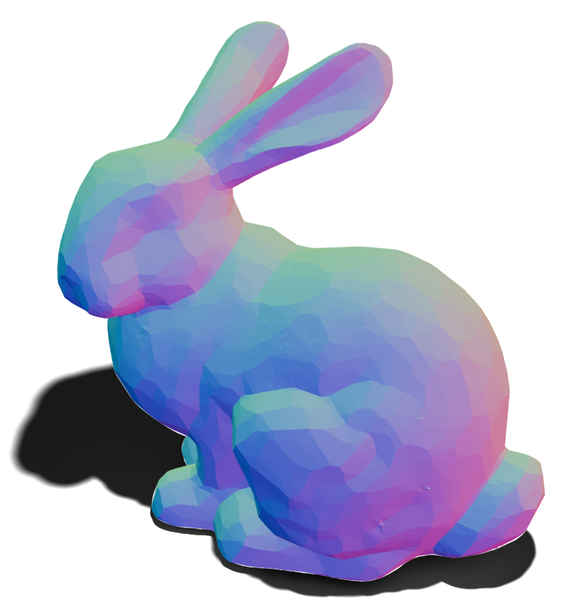}
\\

\ourmethod{} & Guided filter & $L_0$-Smoothing & Image quantization & CLIP Loss & Developability \\
initialization & \cite{GuidedFilterHe2013} & \cite{Xu2011Lzero} & K-means \cite{kmeans}& \cite{radford2021learning} & \cite{Stein:2018:DSF} 
    \end{tabular}
	\caption{We showcase how gradient‐based optimization allows for creating geometric textures by fine‐tuning high‐frequency details. The reference is obtained by optimizing our representation of the Stanford bunny.
    Using gradient-based optimization, we can create geometric textures by applying a style to high-frequency details of a mesh. Using Paparazzi's approach \cite{Liu:Paparazzi:2018}, we can use any image filter for stylization. We display the stylization using Guided filter, $L_0$-smoothing, and K-means color quantization. Our method also allows us to backpropagate gradients from any energy function.  We apply a CLIP loss on the rendered normal maps with the prompt "a wave-like style", and a developability energy \cite{Stein:2018:DSF} directly on the reconstructed mesh, yielding piecewise developable surfaces.}
	\label{fig:geometric_textures}
    \Description{Geometric textures figure.}
\end{figure*}

\section{Applications}
\label{sec:applications}

This section provides examples of gradient-based mesh optimization applications using our method, namely mesh compression, geometric texture generation and photogrammetry.


\subsection{Mesh compression}
\label{sec:compression}

One key advantage of our representation is its relatively low memory footprint{, because we} do not store connectivity. Instead, we infer it using Delaunay Triangulation and a single point can generate multiple vertices in the final mesh. 
{For $N$ points, the size of our representation amounts to $B(4+q)N$ bytes, where $q$ is the number of spherical harmonic coefficients and $B=4$ or $B=8$, depending on whether we use single or double precision. Given that spherical harmonics of degree 1 introduce three additional coefficients per point, opting not to use them can maximize compression efficiency with only minor visual degradation, as evidenced in the compression of \figref{fig:compression_SH}. 
After fitting our representation to a given shape, to maximize the compression rate, we discard the non-active vertices of the tetrahedral grid, \emph{i.e.}, vertices that are not connected to at least one vertex whose base SDF value has an opposite sign. While this operation is not strictly lossless, the difference is almost imperceptible, as shown in \figref{fig:compression_simplification}.}
We compare \ourmethod{} in \figref{fig:mesh_compression} with Neural Geometry Fields (NGF) \cite{vs2024ngfs}, an appearance-based compression method, and Draco \cite{draco2017}, an industry-standard scheme that reduces bit rates in vertex connectivity.
NGF compresses geometry more aggressively but exhibits notable drawbacks compared to our method: its patch-based tessellation can be unfair because patches vary in size, it is not adaptive, and it has no mechanism to prevent self-intersections.
Although Draco is fast, it only focuses on quantizing geometry and is prone to artifacts, so highly tessellated meshes can still consume high memory in applications that do not require such dense sampling.


\begin{figure*}[htbp!]
	\centering
	\small
	\setlength{\tabcolsep}{1pt}
	\begin{tabular}{cccccc}
\includegraphics[width=0.16\linewidth,trim={33 67 33 47},clip]{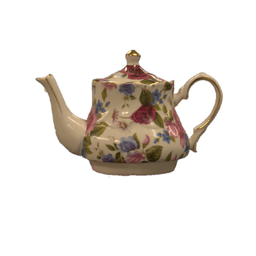} & \includegraphics[width=0.16\linewidth,trim={33 67 33 47},clip]{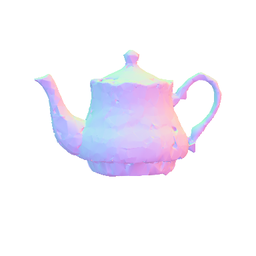}  & \includegraphics[width=0.16\linewidth,trim={33 67 33 47},clip]{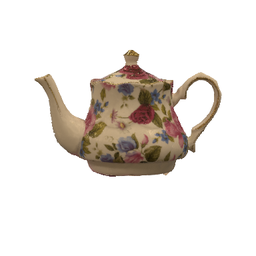} & \includegraphics[width=0.16\linewidth,trim={80 127 93 47},clip]{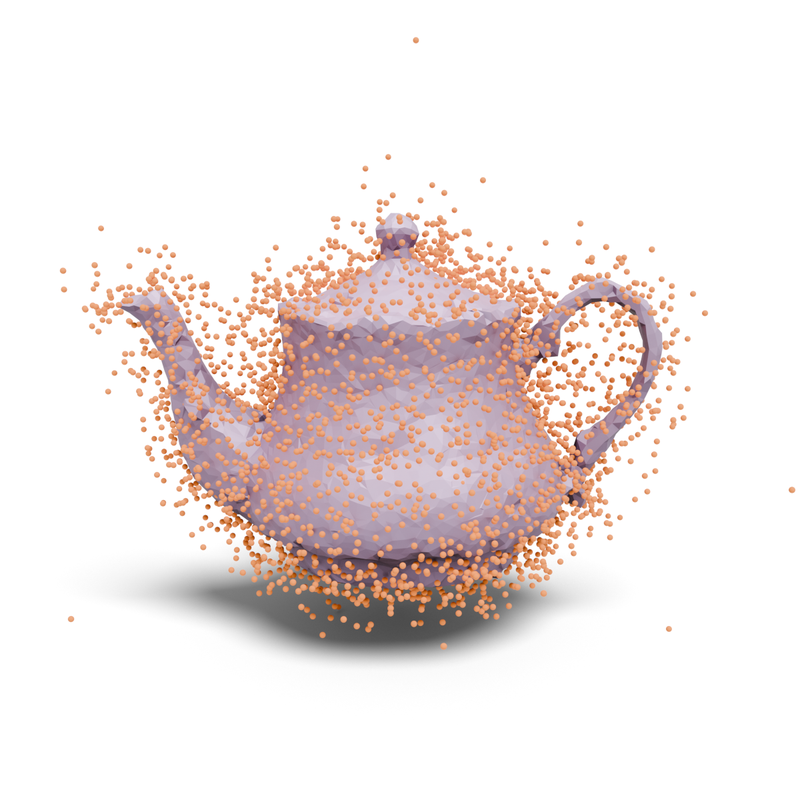} & \includegraphics[width=0.16\linewidth,trim={33 67 33 47},clip]{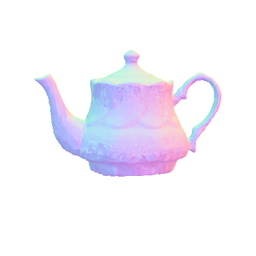} & \includegraphics[width=0.16\linewidth,trim={33 67 33 47},clip]{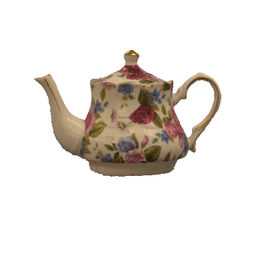} \\
Target & \multicolumn{2}{c}{Coarse reconstruction} & Point cloud adaptivity & \multicolumn{2}{c}{High-frequency refining} 
    \end{tabular}
	\caption{We demonstrate the applicability of our method to photogrammetry by integrating it into the NVDiffRec pipeline \cite{Munkberg_2022_CVPR}. NVDiffRec uses differentiable rasterization to jointly optimize PBR material properties, the environment map, and geometry to match a target rendered shape. In our approach, we employ our method as the geometry representation in a multi-stage process. We initialize a blue-noise point cloud and jointly optimize the point positions along with their SDF values, compute the tetrahedral grid on the fly, and reconstruct a coarse shape. Second, we fit our representation to the coarse mesh, resulting in a point cloud that closely matches the reconstructed shape. Finally, we refine high-frequency details by optimizing only the SDF values in the last stage.}
	\label{fig:photogrammetry}
    \Description{Photogrammetry figure}
\end{figure*}

\subsection{Geometric texture}

Gradient-based mesh optimization can be used for geometric texture generation. Adopting a similar pipeline to Paparazzi \cite{Liu:Paparazzi:2018}, we render normal maps of our generated mesh, $\tilde{I}$, apply an image filter $f$, and substitute the gradient with $f(\tilde{I}) - \tilde{I}$ to update the parameters of \ourmethod{}.
{In practice, we first fit our representation to a reference shape using spherical harmonics of degree $2$. Once the fit is complete, we fix the point positions and SDF values. Using a chosen filter, we render the shape’s normal map from a randomly sampled camera view, apply the filter, and compute gradients to backpropagate, updating only the spherical harmonic coefficients. Since geometric textures primarily affect high-frequency details, adjusting spherical harmonics is sufficient to tilt the normals into the desired orientation without deviating significantly from the reference shape. Additionally, we use a loss on the mask, depth map, and normal map, compared to the reference mesh, to ensure that the generated shape remains close to the reference and does not degenerate. To avoid being constrained too much by the initial reference, we periodically update the reference shape with the current mesh during optimization.
We demonstrate results using various image filters in \figref{fig:geometric_textures}: namely Guided filter \cite{GuidedFilterHe2013}, $L_0$-Smoothing \cite{Xu2011Lzero}, and K-means image quantization \cite{kmeans}.

\ourmethod{} can be adapted to any gradient-based mesh optimization pipelines. Therefore, \figref{fig:geometric_textures} also showcases examples where the gradients are given by a user-defined rendering or mesh energy. From this energy, we can use PyTorch's automatic differentiation to update \ourmethod{}’s parameters and match the target energy. For instance, we demonstrate how we can apply a CLIP loss \cite{radford2021learning} on the normal map to produce a wave-like effect.} Similar to FlexiCubes, our approach supports a developability energy \cite{Stein:2018:DSF} for crafting piecewise developable surfaces, though both methods tend to yield a lot of patches.

\begin{figure*}[t]
\small
\setlength{\tabcolsep}{0pt}
\begin{tabular}{l c c c c c c}
  & \includegraphics[width=0.155\linewidth,trim={53 87 80 117},clip]{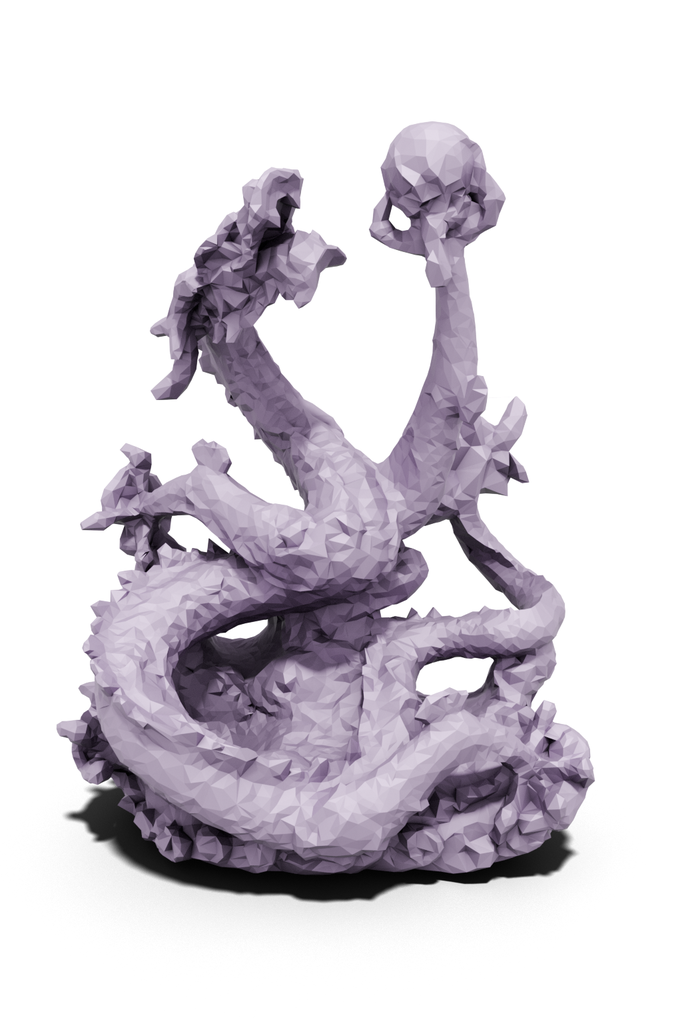} &
  \includegraphics[width=0.155\linewidth,trim={53 87 80 117},clip]{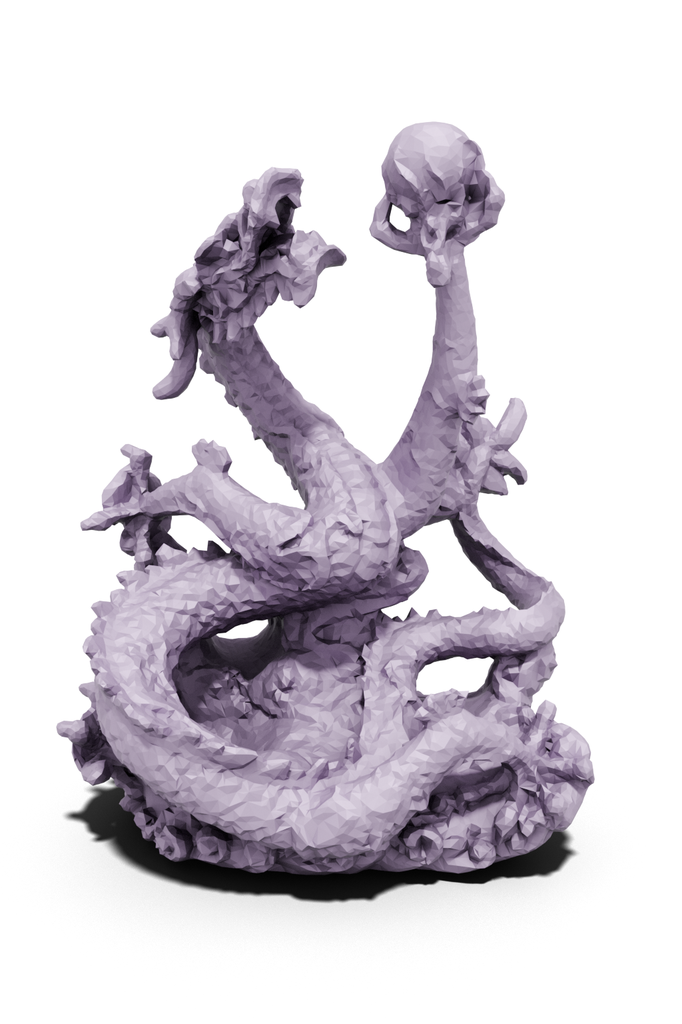} &
  \includegraphics[width=0.155\linewidth,trim={53 87 80 117},clip]{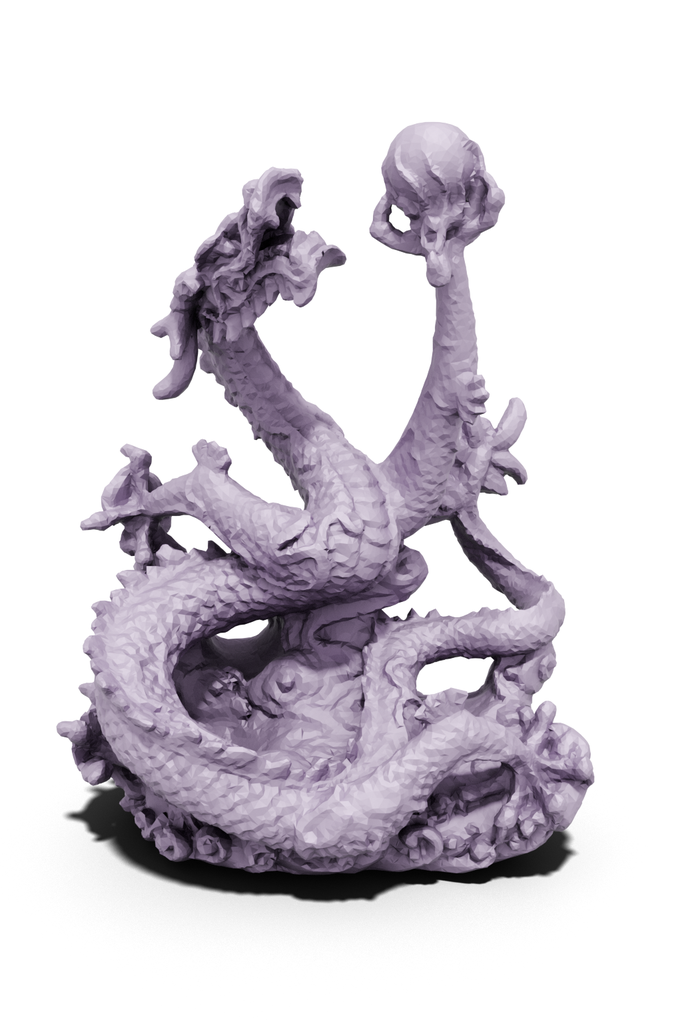} &
  \includegraphics[width=0.155\linewidth,trim={53 87 80 117},clip]{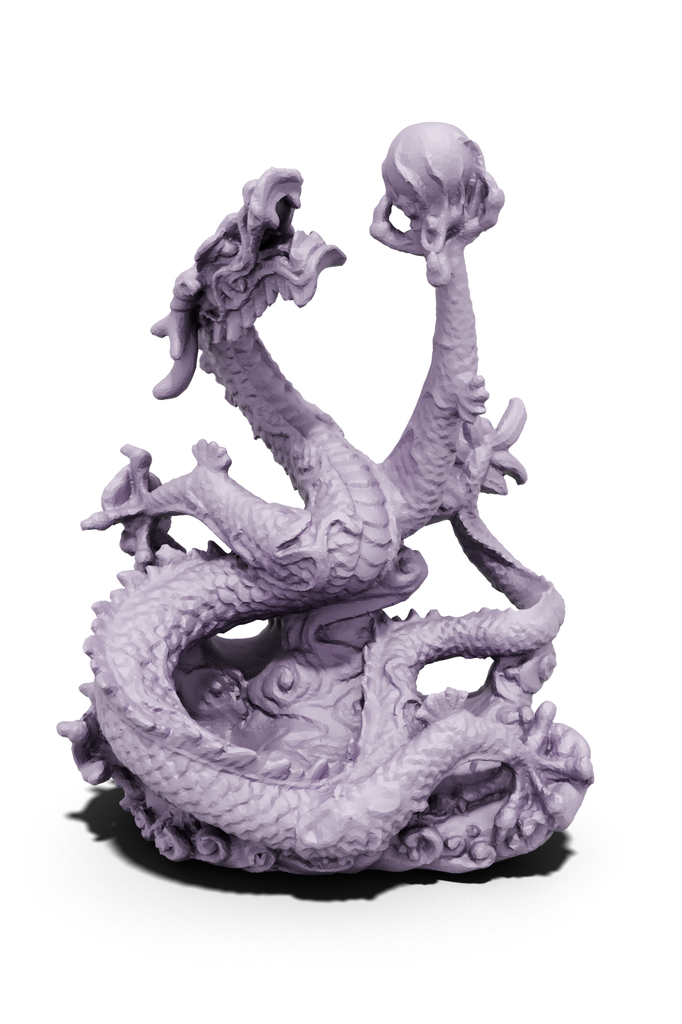} &
  \includegraphics[width=0.155\linewidth,trim={53 87 80 117},clip]{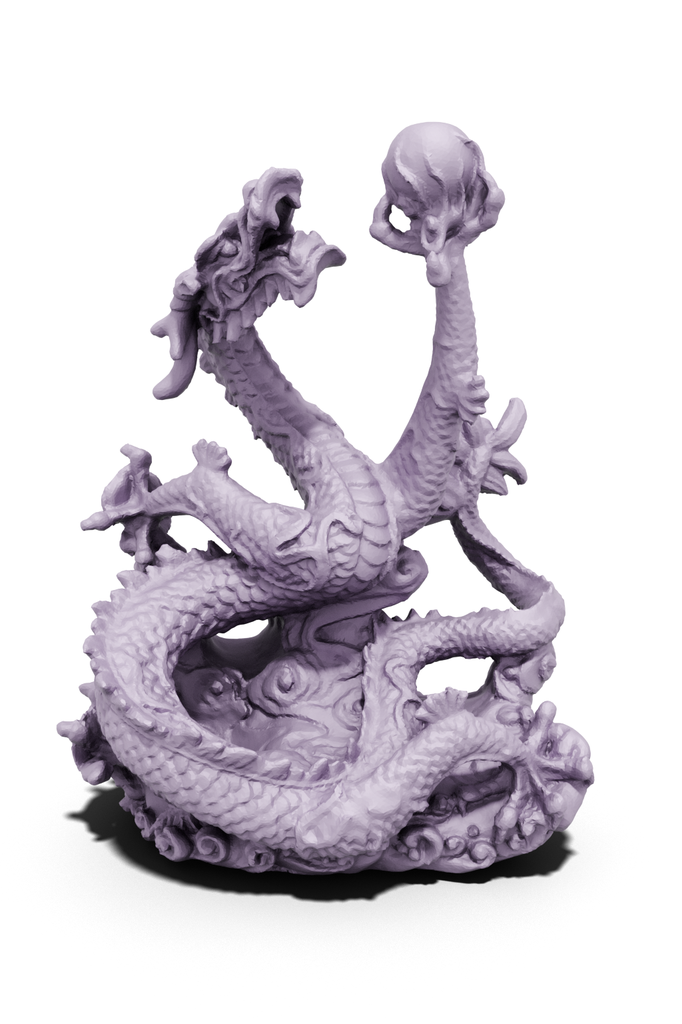} &
  \includegraphics[width=0.155\linewidth,trim={53 87 80 117},clip]{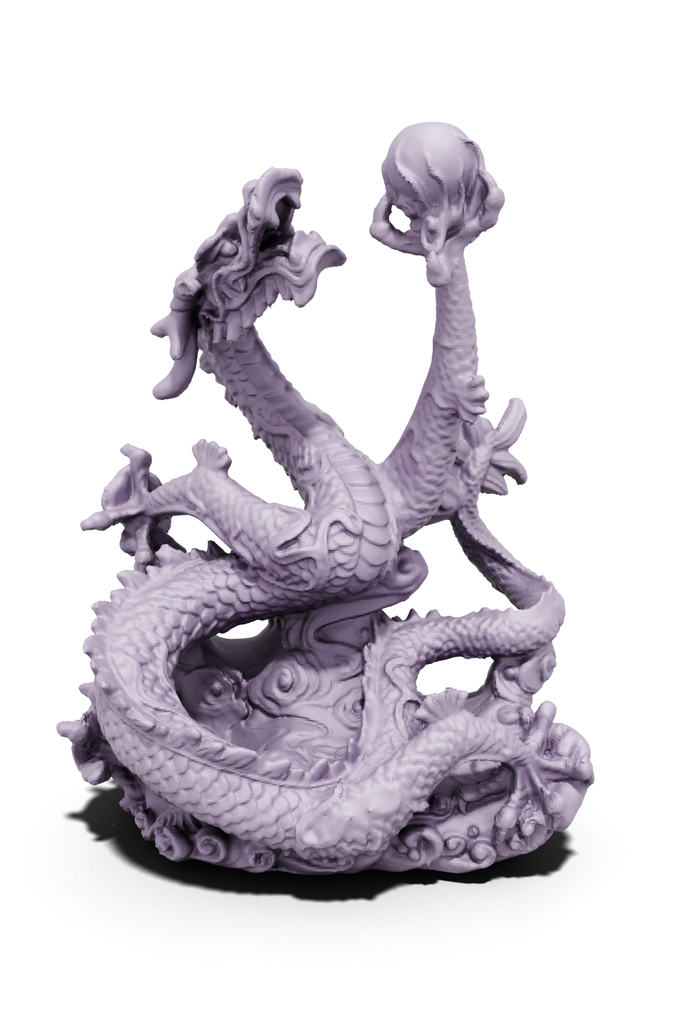} \\
  grid points & 8K / 6.6K & 16K / 13.2K & 32K / 24.4K & 64K / 43.0K & 128K / 76.6K & input \\
  vertices & $15$K & $29$K & $53$K & $91$K & $163$K & $450$K \\
   runtime & $137$s & $157$s & $181$s & $276$s & $456$s & \\
   memory & $108$kB & $213$kB & $393$kB & $690$kB & $1.2$MB & $15.5$ MB
\end{tabular}
  \Description{\ourmethod{} - shape used for performance evaluation}
  \caption{We present the shape utilized in our performance experiments at varying grid point densities. The table shows the target number of grid points used during the optimization and the final number of active points that remained. We also provide details of the final mesh resolution and the practical runtime required to generate each shape. Additionally, we provide the file size of the model, where only active points remained and are stored in float16 format. We used SH with degree 1. The shape showcased here is sourced from Thingi10K \cite{zhou2016thingi10kdataset100003dprinting}.}
  \label{fig:teaser_supmat}
\end{figure*}

\begin{figure}[t]
    \centering
    \includegraphics[width=0.8\linewidth]{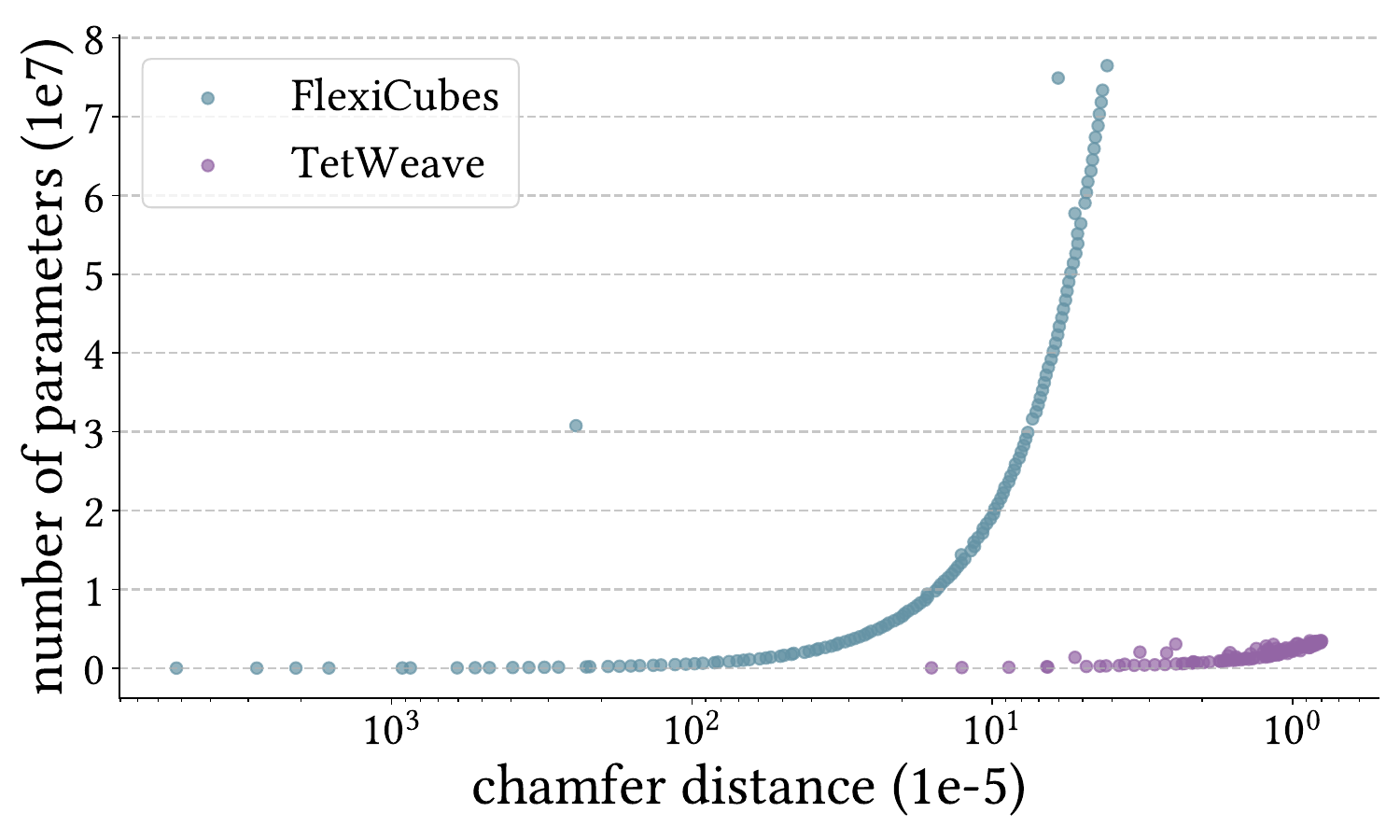}    
    \caption{Relationship between the achieved chamfer distance and the number of parameters required for FlexiCubes and \ourmethod{}. The graph demonstrates that FlexiCubes does not scale as effectively as \ourmethod{}. The maximum grid size for FlexiCubes represented in this figure is $145^3$, while \ourmethod{} is shown utilizing up to 500K points. These experiments were conducted on the shape shown in \figref{fig:teaser_supmat}.}
    \label{fig:scaling}
    \Description{TODO}
\end{figure}
\subsection{Photogrammetry}
\label{sec:photogrammetry}

Photogrammetry involves recovering scene parameters from a set of photographs. NVDiffRec \cite{Munkberg_2022_CVPR} introduces a differentiable rendering framework that jointly reconstructs PBR materials, environment lighting, and geometry. The original implementation employs DMTet for its reconstruction pipeline and the authors provide an implementation of the framework with FlexiCubes. We evaluate each method on the recent benchmark dataset, Stanford ORB~\cite{kuang2023stanfordorb}.

As noted in \secref{sec:stages}, NVDiffRec can output noisy geometry when the reconstruction problem is underconstrained. {\citet{Munkberg_2022_CVPR} addressed similar issues in DMTet by regularizing SDF values through an MLP, by truncating the positional encoding to lower the frequency.}
\ourmethod{} has more degrees of freedom and optimizes the geometry and the background grid concurrently. This is a strength when it comes to expressivity, but can be challenging when optimizing a highly ambiguous inversion problem, where lighting, geometry and materials are unknown. {Thus, rather than using our method as a direct replacement for DMTet \cite{shen2021dmtet}, we employ a multi-stage approach. We initialize TetWeave's point cloud using a precomputed blue noise sampling and construct a tetrahedral grid via Delaunay triangulation. This serves as the foundation for our method during the main stage, where we update point positions and SDF values to extract a coarse mesh. Next, we fit an adaptive grid of $8$K points around the coarse mesh using our standard mesh fitting algorithm described in \autoref{sec:results}. In the late stage, we focus on refining the adaptive point cloud by exclusively optimizing the SDF values. This stage focuses on capturing fine-grained details and eliminating artifacts introduced in earlier stages. The gentle introduction of the optimization of the background grid makes it easier to converge to a correct solution. We observe that increasing the number of points to achieve more detailed meshes can lead to artifacts. To reduce additional degrees of freedom, we avoid using spherical harmonics for this application.}

{
The results for the photogrammetry application of \ourmethod{} are shown in \figref{fig:photogrammetry} and the quantitative metrics are displayed in \tableref{tab:benchmark}, following the benchmark established by Stanford ORB \cite{kuang2023stanfordorb}. While the full benchmark is presented, our discussion focuses on grid-adaptive isosurface extraction techniques, as these are most comparable to our method and were also used with NVDiffRec. Other methods improve on the rendering model and methods to estimate lighting and materials. In our experiments, we find that \ourmethod{} achieves better quality regarding novel lighting- and novel view synthesis. We are also able to retrieve better triangulations and higher resolution meshes. However, we observe that our method does not surpass FlexiCubes and DMTet in geometry metrics. This can be attributed to the fact that competing methods produce lower-resolution meshes that are less susceptible to high-frequency artifacts. Additionally, their grid structures inherently act as regularizers, whereas the high adaptivity of our point cloud prioritizes local geometric detail.
}

\section{Discussion}

This section assesses \ourmethod{}’s performance. We analyze its memory efficiency and runtime characteristics, and examine its limitations to provide a balanced perspective. We conclude by exploring future research directions to further advance unstructured mesh representations.

\subsection{Performance}
\label{sec:performance}

We discuss the performance of our method, focusing on comparisons with FlexiCubes \cite{shen2023flexicubes} in terms of both memory usage and computational speed. Experiments were conducted on an NVIDIA RTX 3090 GPU.

\paragraph{Memory} Our shape representation is particularly memory-efficient. This is due to several reasons: we store only the points, their SDF values, and, optionally, spherical harmonics coefficients. In contrast, FlexiCubes \cite{shen2023flexicubes} requires storing SDF values, deformation parameters per point, and 21 coefficients per voxel. Moreover, our method avoids using a fixed grid and instead adaptively resamples passive points around the shape. This allows each point to contribute to a higher number of output triangles, significantly reducing the number of points needed to achieve a desired level of detail and Chamfer Distance. As illustrated in \figref{fig:scaling}, our approach achieves better Chamfer Distance results with far fewer points compared to FlexiCubes, translating to significantly lower memory requirements. The graph demonstrates that \ourmethod{} scales more efficiently than FlexiCubes, allowing us to reconstruct shapes at a higher resolution.

\begin{table}[b]
\centering
\small
\caption{Quantitative comparison of runtime performance between \ourmethod{} and FlexiCubes \cite{shen2023flexicubes}. The forward pass of \ourmethod{} is divided into two components: the Delaunay triangulation step for generating the tetrahedral grid \cite{tetgenSi2025} and our implementation of the Marching Tetrahedra algorithm \cite{doi_marchingtet}, adapted to incorporate spherical harmonics coefficients (of degree 1 in this experiment).}
\resizebox{1.0\linewidth}{!}{
\begin{tabular}{l l c c c}
\toprule
Method & & \multicolumn{2}{c}{Forward Time (ms)} & Backward Time (ms)\\
\midrule
\multirow{3}{*}{FlexiCubes} & $32^3$ & \multicolumn{2}{c}{$5.88$} & $3.58$\\
 & $64^3$ & \multicolumn{2}{c}{$6.66$} & $5.07$\\
 & $128^3$ & \multicolumn{2}{c}{$9.63$} & $15.25$ \\
\midrule
& & Delaunay Triangulation & Marching Tetrahedra & \\
\cmidrule(l){3-3}\cmidrule(l){4-4}
\multirow{6}{*}{\ourmethod{}} & 8K & 22.66 & 1.48 & 5.25 \\
& 16K & 40.25 & 1.68 & 6.11 \\
& 32K & 76.77 & 2.04 & 8.54 \\
& 64K & 162.24 & 2.94 & 15.72 \\
& 128K & 323.64 & 4.52 & 29.88 \\
& 256K & 658.95 & 8.09 & 57.34 \\
\bottomrule
\end{tabular}
}
\label{table:runtime}
\end{table}

\paragraph{Runtime} Our memory efficiency comes with the trade-off of increased runtime. In \tableref{table:runtime}, we present the average runtime for both FlexiCubes and \ourmethod{}, measured on the shape shown in \figref{fig:teaser_supmat}. 
For our method, most of the time in the forward pass is spent on Delaunay triangulation. Although the theoretical time complexity of Delaunay triangulation is $\bigo (n\ln(n))$, this reflects the worst-case scenario and, in practice, the runtime of Tetgen \cite{tetgenSi2025} appears to scale linearly. Similarly, our backward pass runtime also exhibits linear scaling, which can be attributed to the linear increase in the number of vertices generated by our approach.
In practice, the runtime does not exhibit strict linear scaling due to several factors. First, the number of points is progressively increased throughout the optimization process. All point clouds in \figref{fig:teaser_supmat} were initialized with $8$K points, which were gradually increased to the target number through our resampling strategy. Second, the tetrahedral grid is not recomputed at every step; instead, gradients are accumulated over several iterations. Point positions and the tetrahedral grid are only updated every five steps. Lastly, during late-stage refinement, point positions and the triangulation remain fixed, significantly reducing the runtime for the final optimization steps. As a result, the runtime can vary from two minutes at $8$K points to less than eight minutes at $128$K.

\subsection{Limitations}

\begin{figure}[t!]
    \centering
    \small
    \setlength{\tabcolsep}{0pt}
    \begin{tabular}{c c c}
    \includegraphics[width=.33\linewidth]{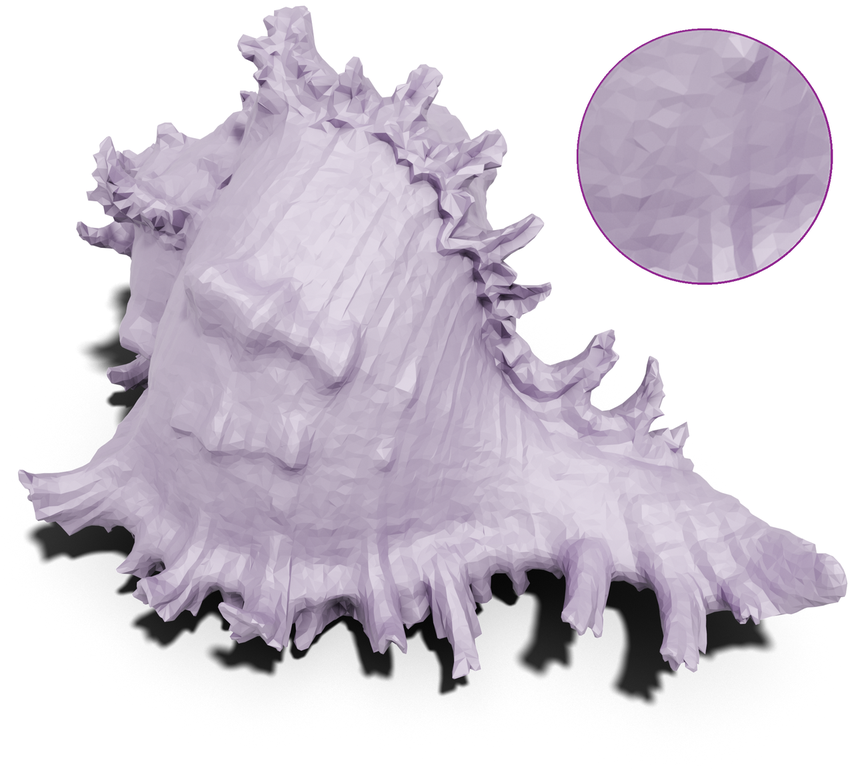} & 
    \includegraphics[width=.33\linewidth]{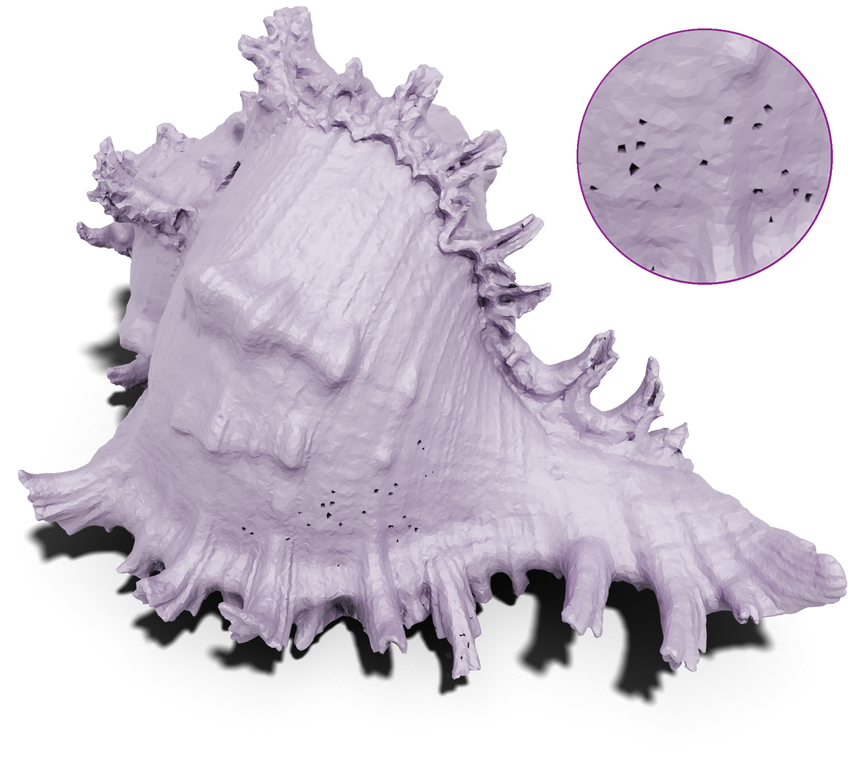}& 
    \includegraphics[width=.33\linewidth]{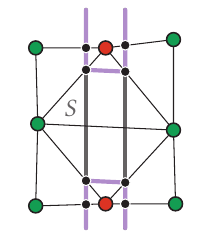} \\
    FlexiCubes & \ourmethod{} & visual interpretation
    \end{tabular}
    \caption{We evaluate our method against FlexiCubes on a shape featuring a particularly thin structure. In such cases, the Delaunay triangulation can sometimes link two points with similar signed distance function values, resulting in an inaccurate hole that deviates from the target geometry $S$.}
    \label{fig:thin_structure}
    \Description{TODO}
\end{figure}

Isosurface extraction methods are generally unsuitable for handling thin structures, and our approach is particularly sensitive to such cases. As illustrated in \figref{fig:thin_structure}, when dealing with meshes containing thin structures, the Delaunay triangulation may connect points with the same sign, resulting in holes in the final shape.

\subsection{Conclusion}
\label{sec:conclusion}

This paper presents \ourmethod{}, a scalable isosurface representation for gradient-based mesh optimization that constructs tetrahedral grids on the fly via Delaunay triangulation. Our method produces watertight, manifold, and intersection-free meshes, excelling at handling complex shapes with lower memory overhead. However, several challenges remain. While the final mesh resolution grows nearly linearly with the number of points, Delaunay triangulation can exhibit superlinear runtime; although this did not pose problems in our experiments, it implies that time complexity may become the limiting factor before memory. Furthermore, like other grid-adaptive approaches, internal cavities can arise unless regularization is applied, and thin structures remain problematic, to which our method is especially sensitive.

We believe that unstructured representations for mesh-based pipelines still have many opportunities for development. Since Marching Tetrahedra only requires consistently oriented tetrahedral grids, devising a technique to infer non-Delaunay connectivity on the fly could offer additional flexibility. Another promising avenue is a direct pipeline to convert arbitrary meshes into our representation, which could immediately enable near-lossless compression and provide a convenient parameter space for learning-based tasks—particularly valuable for generative models. Determining the feasibility and practical implementation of such a pipeline remains an open challenge.
\begin{acks}
We thank the anonymous reviewers for their constructive feedback and Marcel Padilla for his careful reading of our manuscript. The open source codebases of FlexiCubes and DMTet have been instrumental in the development of \ourmethod{}. This work was supported in part by the European Research Council (ERC) under the European Union’s Horizon 2020 research and innovation program (grant agreement No. 101003104, ERC CoG MYCLOTH).
\end{acks}

\balance
\bibliographystyle{ACM-Reference-Format}
\bibliography{references}

\newpage
\clearpage

\appendix
\section{Implementation Details}
\label{sec:implementation_details}

In this section, we detail the implementation of our shape reconstruction pipeline, as described in \secref{sec:results}. Differentiable rasterization is performed using nvdiffrast \cite{Laine2020diffrast}, and many geometric operations are used from the Kaolin library \cite{KaolinLibrary}.

\paragraph{Initialization}Before starting the reconstruction process, the shapes in our dataset are normalized such that their bounding boxes fit within the cube $[-0.9, 0.9]^3$, similar to the preprocessing used in FlexiCubes \cite{shen2023flexicubes}. For grid initialization, we randomly sample $8000$ points within a sphere of radius $\sqrt{3}$. {Spherical harmonics coefficients are initialized to $0$, which leads to an initial state where $\hat{s}_i(e) = s_i$ due to the application of the $1+tanh$ function.} Thanks to our iterative resampling approach, \ourmethod{} is robust to the resolution of the initial sampling, as long as it adequately covers the target shape. In contrast, methods like FlexiCubes and DMTet \cite{shen2021dmtet} are more sensitive to the initial grid scale, requiring precise tuning to ensure adequate coverage of the target shape. If the grid is too small, some parts of the shape could be missed. If the grid is too large, the resulting mesh lacks detail.

\paragraph{Delaunay triangulation} To create a Delaunay tetrahedral grid from a point cloud, we use TetGen \cite{tetgenSi2025}. In practice, we observe that TetGen can stall when points are too close to one another. To remove this issue we perturb points that are too close to each other.

\begin{figure}[t]
    \centering
    \includegraphics[width=0.9\linewidth]{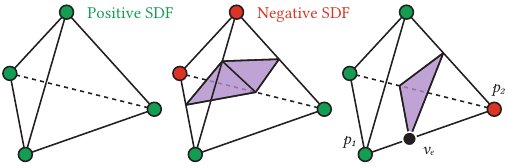}
    
    \caption{Marching Tetrahedra lookup configurations. Despite $2^4$ possible configurations, they amount to three distinct possibilities up to rotation and sign inversion displayed here. Contrary to Marching Cubes, there is no ambiguous case, and the resulting mesh is guaranteed intersection-free and 2-manifold.}
    \label{fig:marching_tets_configurations}
    \Description{TODO}
\end{figure}

\paragraph{Resampling and adaptive meshing} 

Our approach involves resampling passive points while progressively adding new points during optimization. We gradually increase the number of points from the initial $8000$ to the desired target count. For the mesh reconstruction experiment, we sample camera positions uniformly on a sphere using a Fibonacci lattice. We compute the normal maps of both the target shape and the current shape and evaluate the per-pixel difference using an $L_1$ loss. To facilitate adaptive resampling, we divide the space around the current shape into voxels using a grid resolution of $32^3$. For each voxel we determine a target number of sample points \(k\). To achieve this, we start with a precomputed blue noise point set which we scale such that \(k\) points fall within the voxel.

\paragraph{Parameters} {The final optimization objective can be written as

\begin{equation*}
    \mathcal{L} = \mathcal{L}_{\text{recons}} + \lambda_{\text{fairness}} \mathcal{L}_{\text{fairness}} + \lambda_{\text{ODT}} \mathcal{L}_{\text{ODT}} + \lambda_{\text{sign}} \mathcal{L}_{\text{sign}}
\end{equation*}

In our base experiment, we use $\lambda_M = 10$, $\lambda_D = 250$ $\lambda_N = 1$, $\lambda_{\text{fairness}} = 0.35$, $\lambda_{\text{ODT}} = 0.1$, and $\lambda_{\text{sign}} = 1.0$.} We use spherical harmonics of degree $2$. {Because the computation of the vertices of the generated mesh is differentiable with respect to the grid points, the signed distance function and the spherical harmonics coefficients, we can leverage PyTorch's automatic differentiation \cite{paszke2019pytorchimperativestylehighperformance} and update parameters using the AdamW optimizer \cite{loshchilov2019decoupledweightdecayregularization}.} The step size ($0.002$ for SDF values, $0.0003$ for point positions) remains constant throughout the optimization process. 

\section{Optimal Delaunay Triangulation}
\label{sec:odt_proof}

In this section, we provide an in-depth exploration of the Optimal Delaunay Triangulation loss $E_{\text{ODT}}$ \cite{ODTChen2004,alliez2005variationaltetmeshing}. We explain how it is computed and prove that this energy is minimized on a regular tetrahedron.
While \citeauthor{ODTChen2004} derive the update step to minimize the ODT loss, we are interested in computing the loss itself, to be included in a gradient-descent optimizer with autograd. This allows us to simultaneously optimize mesh quality and reconstruction losses. 

\subsection{Computation details}
The vertex positions of a tetrahedron $T$ are given as $\textsc{v}_0, \dots, \textsc{v}_3$. The ODT energy for $T$ is given as 
\begin{equation}
E_{\text{ODT}}(T) = \vert M_{S_{T}} - M_{T} \vert.
\label{eq:odt}
\end{equation}
where $ M_{T} $ represents the sum of the principal moments of $T$ relative to its circumcenter $c_{T}$. $M_{S_{T}}$ denotes the moment of inertia of $S_{T}$, defined as the circumsphere of $T$ and having an equivalent mass. Assuming unit mass density, the mass of $T$ is simply its volume $V_T$. 
With
$$
\mathbf{a} = \textsc{v}_1 - \textsc{v}_0, \quad \mathbf{b} = \textsc{v}_2 - \textsc{v}_0, \quad \mathbf{c} = \textsc{v}_3 - \textsc{v}_0, \quad D = \det([\mathbf{a}, \mathbf{b}, \mathbf{c}]).
$$
We obtain $$V_T = \frac{\vert D \vert }{6}.$$
The circumcenter $ c_T $ is computed as \cite{weisstein_circumsphere}
$$
c_T = \textsc{v}_0 + \frac{\|\mathbf{a}\|^2  (\mathbf{b} \times \mathbf{c}) + \|\mathbf{b}\|^2  (\mathbf{c} \times \mathbf{a}) + \|\mathbf{c}\|^2  (\mathbf{a} \times \mathbf{b})}{2D}.
$$
The circumradius is given by $R_T = \|c_T - \textsc{v}_0\|$. Hence, the moment of inertia of $S_T$ is given as
\begin{equation}
\label{eq:mst}
M_{S_T} = \frac{2}{5}V_T R_T^2.
\end{equation}

An explicit formula for computing $M_T$ is given by \citet{Tonon2005TetExplicit}. We define the relative coordinates of the vertices with respect to the circumcenter $c_T$:
\[
x_i = \textsc{v}_i^x - c_T^x, \quad y_i = \textsc{v}_i^y - c_T^y, \quad z_i = \textsc{v}_i^z - c_T^z, \quad \text{for } i = 0, 1, 2, 3.
\]

The coordinate quadratic sums are given by
\[
S_x = \sum_{i=0}^3 x_i^2 + \sum_{0 \leq i < j \leq 3} x_i x_j.
\]

We define $S_y$ and $S_z$ similarly for the $y$ and $z$ coordinates. Computing the moments of inertia along each axis yields:
\[
I_x = \frac{V_T}{10} (S_y + S_z), \quad
I_y = \frac{V_T}{10} (S_x + S_z), \quad
I_z = \frac{V_T}{10} (S_x + S_y),
\]
where $V_T$ is the volume of the tetrahedron. Hence, the sum of the principal moments of inertia is given by
\begin{equation}
\label{eq:mt}
M_T = I_x + I_y + I_z = \frac{1}{5}V_T\left( S_x + S_y + S_z\right).
\end{equation}

\subsection{Theoretical Analysis} We show the following property:

\begin{proposition}
\label{prop:min_odt}
For any regular tetrahedron $T$, $E_{\text{ODT}}(T) = 0$.
\end{proposition}

\begin{proof}
$M_{S_T}$ and $M_T$ are invariant under rigid transformations of $T$. The regular tetrahedron $T_\alpha$ is given by $$\textsc{v}_0 = \left(\alpha, \alpha, \alpha \right), \, \textsc{v}_1 = \left(\alpha, -\alpha, -\alpha \right), \, \textsc{v}_2 = \left(-\alpha, \alpha, -\alpha \right), \, \textsc{v}_3 = \left(-\alpha, -\alpha, \alpha \right).$$

Because the circumcenter of a regular tetrahedron is its barycenter, we have $c_{T_\alpha} = \left(0, 0, 0 \right)$ and therefore $R_{T_\alpha}^2 = 3\alpha^2$. This evaluates to $S_x = S_y = S_z = 2\alpha^2$. Combining these with Eq. \eqref{eq:mst} and \eqref{eq:mt}, we obtain $$M_{S_{T_\alpha}} = M_{T_\alpha} = \frac{6}{5}V_{T_\alpha}\alpha^2.$$
Since $E_{\text{ODT}}$ is the absolute value of the difference of these quantities, it attains the value zero for the regular tetrahedron.
\end{proof}

\section{Evaluation details}

We explain the metrics used in our evaluation and provide additional rendering-based metrics.

\subsection{Metric definition}

\label{sec:metric_def}

In this section we define the metrics used for our evaluation. They are consistent with the ones used in FlexiCubes \cite{shen2023flexicubes}.
The chamfer distance (\textbf{CD}) quantifies the similarity between two point clouds. For each point in one of the point sets we compute the distance to the closest point in the other set. Averaging all distances yields the chamfer distance. To compare meshes, we sample points on both the ground-truth and reconstructed mesh, resulting in a point cloud with one million points for each mesh. We also compute the F1-score (\textbf{F1}) as $$F_1 = \frac{2TP}{2TP + FP + FN}$$
based on the sampled points. If the nearest point has a distance below a certain threshold, we count it as a true positive ($TP$). If its distance is above the threshold, it is counted as a false positive ($FP$) if it belongs to the reconstructed mesh, and a false negative ($FN$) if it belongs to the ground truth. We use a threshold of $0.001$. The edge chamfer distance (\textbf{ECD}) and edge F1-score (\textbf{EF1}) are similar to the chamfer distance and F1 score, but apply to edge points \cite{chen2022ndc}, \emph{i.e.} points whose normal has an average dot product with neighboring points' normals that falls below a threshold of $0.2$. The inaccurate-normals metric (\textbf{IN} $\mathbf{> 5^\circ}$) captures the percentage of points for which the angle difference between predicted and ground truth normals exceeds 5 degrees. Normals of sampled points are copied from the face containing the sample point. The nearest point pairs between the predicted and ground truth meshes are identified to compute these angular differences. For a triangle, the aspect ratio (\textbf{AR}) is defined as the ratio of the longest edge to the shortest altitude, while the radius ratio (\textbf{RR}) is defined as the ratio of the inradius to the circumradius. Both metrics assess triangle quality on the extracted mesh, with lower values indicating better triangle regularity. The small angles metric (\textbf{SA}$\mathbf{<10^\circ}$) calculates the percentage of triangles with smallest internal angle below $10^\circ$, which is a proxy for the amount of sliver triangles. We also compute the percentage of intersecting faces (\textbf{SI}) using PyMeshLab \cite{pymeshlab}.

\subsection{Rendering metrics}

The rendering metrics presented in \tableref{fig:rendering_metrics} evaluate the fidelity of the reconstructed meshes to the ground truth through image-based comparisons, rather than purely geometric metrics. These metrics -- mask error, depth error, and normals error -- offer insight into how well the reconstructed shape aligns perceptually with the target shape when rendered from various viewpoints.
\ourmethod{} demonstrates competitive performance compared to FlexiCubes \cite{shen2023flexicubes} and DMTet \cite{shen2021dmtet}, consistently achieving lower errors across all metrics as the number of points is scaled up. This indicates that \ourmethod{} generates surfaces that are more visually aligned with the target shape in terms of silhouette, depth, and normal accuracy.

\begin{table}[b!]
\centering
\footnotesize
\setlength{\tabcolsep}{2pt}
\caption{Quantitative evaluation on the mesh reconstruction task with rendering metrics. We sample two thousand points using Fibonacci sampling over a sphere to position cameras and render the mask, depth map, and normal map of the reconstructed mesh and ground truth. We use an $L_2$ distance as comparison metric. The mask error is scaled by $10^3$, the depth error by $10^2$, and the normal error by $10^1$.}

\resizebox{1.0\linewidth}{!}{
\begin{tabular}{l c c c}
\toprule
Method & Mask Error $\times 10^{3}$ $\downarrow$ & Depth Error $\times 10^{2}$ $\downarrow$ & Normals Error $\times 10^{1}$ $\downarrow$ \\
\midrule
DMTet ($128^3$) & 0.56 & 1.12 & 2.15 \\
FlexiCubes ($128^3$) &0.29 & 0.54 & 1.51 \\
TetWeave ($16$K) & 0.39 & 0.70 & 1.79 \\
TetWeave ($64$K) & 0.23 & 0.36 & 1.30 \\
TetWeave ($128$K) & 0.17 & 0.26 & 1.10 \\
\bottomrule
\end{tabular}
}
\label{fig:rendering_metrics}
\end{table}

\begin{table*}[ht!]
\centering
\small
\setlength{\tabcolsep}{4pt}
\caption{Expanded version of \tableref{fig:metrics} on quantitative evaluation on the mesh reconstruction task. We refer to the text in the Appendix for a more detailed explanation of the different quantities. The chamfer distance (CD) is scaled by 1e5, and the edge chamfer distance (ECD) by 1e2. Ablation measures are done with respect to \ourmethod{} with 64K grid points.}
\resizebox{1.0\linewidth}{!}{
\begin{tabular}{l c c c c c c c c c c c c}
\toprule
Method & CD $\downarrow$ & F1 $\uparrow$ & ECD $\downarrow$ & EF1 $\uparrow$ & NC $\uparrow$ & IN$>5^{\circ}$(\%) $\downarrow$ & AR$>4$(\%) $\downarrow$ & RR$>4$(\%) $\downarrow$ & SA$<10^{\circ}$(\%) $\downarrow$ & SI(\%) $\downarrow$ & \#V & \#F \\
\midrule

DMTet ($32^3$) & 28.054 & 0.132 & 4.384 & 0.150 & 0.895 & 71.095 & 12.755 & 12.649 & 13.044 & 0.000 & 1355 & 27102 \\

DMTet ($64^3$) & 7.036 & 0.219 & 2.983 & 0.187 & 0.936 & 61.504 & 11.613 & 11.425 & 11.908 & 0.000 & 5030 & 10066 \\

DMTet ($128^3$) & 1.043 & 0.339 & 1.681 & 0.272 & 0.965 & 48.393 & 12.026 & 11.826 & 12.351 & 0.000 & 20677 & 41364 \\
\midrule
FlexiCubes ($32^3$) & 12.350 & 0.210 & 3.784 & 0.191 & 0.931 & 62.187 & 7.275 & 8.714 & 6.362 & 0.775 & 1776 & 3556 \\
FlexiCubes ($64^3$) & 2.900 & 0.329 & 2.378 & 0.257 & 0.961 & 49.814 & 6.055 & 7.327 & 5.236 & 0.341 & 7086 & 14181 \\
FlexiCubes ($128^3$)  & 0.752 & 0.416 & 1.254 & 0.393 & 0.979 & 36.911 & 5.418 & 6.701 & 4.588 & 0.203 & 28430 & 56873 \\ 
\midrule

\ourmethod{} (8K) &  1.176 & 0.376 & 1.952 & 0.283 & 0.967 & 48.635 & 2.511 & 3.534 & 1.922 & 0.000 & 14715 & 29468 \\
\ourmethod{} (16K) & 0.517 & 0.409 & 1.475 & 0.353 & 0.974 & 43.380 & 2.350 & 3.344 & 1.743 & 0.000 & 26484 & 53015 \\
\ourmethod{} (32K) & 0.489 & 0.431 & 1.157 & 0.433 & 0.979 & 38.367 & 2.292 & 3.287 & 1.675 & 0.000 & 46640 & 93330 \\
\ourmethod{} (64K) &  0.419 & 0.446 & 0.962 & 0.518 & 0.984 & 33.700 & 2.251 & 3.252 & 1.616 & 0.000 & 81027 & 162102\\
\ourmethod{} (128K) & 0.393 & 0.455 & 0.708 & 0.588 & 0.987 & 29.361 & 2.507 & 3.556 & 1.829 & 0.000 & 146514 & 293074 \\
\midrule
\multicolumn{13}{c}{\ourmethod{} - 64K} \\
\midrule
No SH & 0.415 & 0.436 & 0.919 & 0.466 & 0.981 & 37.712 & 2.180 & 3.147 & 1.556 & 0.000 & 83866 & 167785 \\
SH - degree 1 &  0.402 & 0.446 & 0.787 & 0.528 & 0.984 & 33.633 & 2.254 & 3.245 & 1.624 & 0.000 & 80865 & 161775 \\
SH - degree 2 &  0.419 & 0.446 & 0.962 & 0.518 & 0.984 & 33.700 & 2.251 & 3.252 & 1.616 & 0.000 & 81027 & 162102 \\
SH - degree 3 & 0.411 & 0.446 & 0.916 & 0.507 & 0.984 & 33.864 & 2.245 & 3.253 & 1.611 & 0.000 & 81141 & 162329 \\
SH - degree 4 & 0.412 & 0.446 & 1.212 & 0.508 & 0.984 & 33.949 & 2.261 & 3.271 & 1.623 & 0.000 & 81699 & 163446 \\
\midrule
\multicolumn{13}{c}{\ourmethod{} - 64K} \\
\midrule
No Fairness & 0.781 & 0.455 & 1.119 & 0.511 & 0.987 & 28.866 & 17.091 & 19.353 & 15.611 & 0.000 & 127428 & 255107 \\
No ODT & 0.419 & 0.446 & 1.071 & 0.517 & 0.984 & 33.733 & 2.253 & 3.248 & 1.617 & 0.000 & 81320 & 162682 \\
Uniform & 0.427 & 0.436 & 1.051 & 0.427 & 0.980 & 36.453 & 1.567 & 2.514 & 0.975 & 0.000 & 61794 & 123635 \\
\bottomrule
\end{tabular}
}
\label{tab:metrics_extended}
\end{table*}
\begin{table*}[h!]
\centering
\setlength{\tabcolsep}{3pt}
\scriptsize
\caption{Comparison on Stanford ORB \cite{kuang2023stanfordorb} using NVDiffRec \cite{Munkberg_2022_CVPR} as a backbone.
Depth SI-MSE $\times 10^{-3}$. Shape Chamfer distance $\times 10^{-3}$.
We highlight in bold the best-performing techniques overall, as well as the top-performing technique specifically within the isosurface extraction category, which is the primary focus of our analysis.}
\label{tab:benchmark}
\resizebox{1.0\linewidth}{!}{
\begin{tabular}{lccccccccccc}
\toprule
 \multirow{2}{*}{}  
 & \multicolumn{3}{c}{Geometry}
 & \multicolumn{4}{c}{Novel Scene Relighting}            & \multicolumn{4}{c}{Novel View Synthesis} 
 \\
 \cmidrule(l){2-4} \cmidrule(l){5-8} \cmidrule(l){9-12}
& Depth$\downarrow$ & Normal$\downarrow$ & Shape$\downarrow$ & PSNR-H$\uparrow$ & PSNR-L$\uparrow$ & SSIM$\uparrow$ & LPIPS$\downarrow$ & PSNR-H$\uparrow$ & PSNR-L$\uparrow$ & SSIM$\uparrow$ & LPIPS$\downarrow$\\\midrule
IDR~\cite{yariv2020multiview} & $0.35$ & $0.05$ & $\mathbf{0.30}$ & \multicolumn{4}{c}{N/A} & $\mathbf{30.11}$ & $\mathbf{39.66}$ & $\mathbf{0.990}$ & $\mathbf{0.017}$\\
NeRF~\cite{mildenhall2020nerf} & $2.19$ & $0.62$ & $62.05$ & \multicolumn{4}{c}{N/A} & $26.31$ & $33.59$ & $0.968$ & $0.044$\\
\midrule
Neural-PIL~\cite{boss2021neuralpil} & $0.86$ & $0.29$ & $4.14$ & \multicolumn{4}{c}{N/A} & $25.79$ & $33.35$ & $0.963$ & $0.051$\\
PhySG~\cite{physg2021} & $1.90$ & $0.17$ & $9.28$ & $21.81$ & $28.11$ & $0.960$ & $0.055$ & $24.24$ & $32.15$ & $0.974$ & $0.047$\\
NeRD~\cite{boss2021nerd} & $1.39$ & $0.28$ & $13.70$ & $23.29$ & $29.65$ & $0.957$ & $0.059$ & $25.83$ & $32.61$ & $0.963$ & $0.054$\\
NeRFactor~\cite{zhang2021nerfactor} & $0.87$ & $0.29$ & $9.53$ & $23.54$ & $30.38$ & $0.969$ & $0.048$ & $26.06$ & $33.47$ & $0.973$ & $0.046$\\
InvRender~\cite{Wu_2023_CVPR} & $0.59$ & $0.06$ & $0.44$ & $23.76$ & $30.83$ & $0.970$ & $0.046$ & $25.91$ & $34.01$ & $0.977$ & $0.042$\\
NVDiffRecMC~\cite{hasselgren2022nvdiffrecmc} & $0.32$ & $\mathbf{0.04}$ & $0.51$ & $24.43$ & $31.60$ & $0.972$ & $0.036$ & $28.03$ & $36.40$ & $0.982$ & $0.028$\\
\midrule
DMTet~\cite{shen2021dmtet} & $\mathbf{0.31}$ & $0.06$ & $0.62$ & $22.91$ & $29.72$ & $0.963$ & $0.039$ & $21.94$ & $28.44$ & $0.969$ & $0.030$\\
FlexiCubes ~\cite{shen2023flexicubes} & $0.32$ & $\mathbf{0.05}$ & $\mathbf{0.49}$ & $23.26$ & $29.99$ & ${0.964}$ & $\mathbf{0.037}$ & ${22.21}$ & ${28.72}$ & $\mathbf{0.970}$ & $\mathbf{0.028}$\\
\textbf{\ourmethod{}} & $0.35$ & $0.07$ & $0.58$ & $\mathbf{23.48}$ & $\mathbf{30.30}$ & $\mathbf{0.965}$ & $0.038$ & $\mathbf{22.30}$ & $\mathbf{28.98}$ & $\mathbf{0.970}$ & $0.031$ \\
\bottomrule
\end{tabular}
}
\end{table*}
\end{document}